\documentclass[letterpaper,11pt]{article}
\pdfoutput=1 
\usepackage[utf8]{inputenc}

\usepackage[margin=1in]{geometry}
\usepackage{amsfonts}
\usepackage{amsmath}
\usepackage{amssymb}
\usepackage{amsthm}
\usepackage{graphicx}
\usepackage{float}
\usepackage[font=small,labelfont=bf]{caption} 
\allowdisplaybreaks

\usepackage{bold-extra} 

\usepackage{hyperref}
\usepackage[svgnames]{xcolor}
\hypersetup{colorlinks={true},urlcolor={blue},linkcolor={DarkBlue},citecolor=[named]{DarkGreen}}

\usepackage{microtype}
\usepackage[capitalise,nameinlink,noabbrev]{cleveref}

\usepackage{xspace}

\usepackage{tikz}
\usepackage{pgfplots}
\pgfplotsset{compat=newest}
\usetikzlibrary{arrows.meta}
\usetikzlibrary{matrix,arrows,shapes,positioning,calc,snakes,decorations.markings,math}
\usetikzlibrary{arrows.meta}
\tikzset{>={Latex[width=5mm,length=5mm]}}

\pgfdeclarelayer{background}
\pgfsetlayers{background,main}

\usetikzlibrary{intersections}

\usepackage{doi}

\theoremstyle{definition}
\newtheorem{definition}{Definition}

\theoremstyle{plain}
\newtheorem{theorem}{Theorem}[section]
\newtheorem{lemma}[theorem]{Lemma}

\newtheorem{proposition}[theorem]{Proposition}
\newtheorem{claim}{Claim}

\newtheorem*{lemma*}{Lemma}
\newtheorem*{theorem*}{Theorem}
\newtheorem*{desiredlemma*}{Desired Lemma}

\theoremstyle{remark}
\newtheorem{remark}{Remark}
\newtheorem{example}{Example}

\Crefname{claim}{Claim}{Claims}

\newcommand{\size}{\textup{\textsf{size}}}

\DeclareMathOperator{\trunc}{\textup{\textsf{T}}}
\newcommand{\conv}{\textup{\textsf{conv}}}

\def\fp/{\textup{\textsf{FP}}}
\def\p/{\textup{\textsf{P}}}
\def\np/{\textup{\textsf{NP}}}
\def\conp/{\textup{\textsf{co-NP}}}
\def\fnp/{\textup{\textsf{FNP}}}
\def\tfnp/{\textup{\textsf{TFNP}}}
\def\ptfnp/{\textup{\textsf{PTFNP}}}
\def\ppa/{\textup{\textsf{PPA}}}
\def\ppad/{\textup{\textsf{PPAD}}}
\def\ppads/{\textup{\textsf{PPADS}}}
\def\ppp/{\textup{\textsf{PPP}}}
\def\pwpp/{\textup{\textsf{PWPP}}}
\def\pls/{\textup{\textsf{PLS}}}
\def\cls/{\textup{\textsf{CLS}}}
\def\ppadpls/{\textup{$\textsf{PPAD} \cap \textsf{PLS}$}}
\def\ppapls/{\textup{$\textsf{PPA} \cap \textsf{PLS}$}}
\def\eopl/{\textup{\textsf{EOPL}}}
\def\sopl/{\textup{\textsf{SOPL}}}
\def\ueopl/{\textup{\textsf{UEOPL}}}
\def\fixp/{\textup{\textsf{FIXP}}}
\def\bu/{\textup{\textsf{BU}}}
\def\linearfixp/{\textup{\textsf{Linear-FIXP}}}
\def\pspace/{\textup{\textsf{PSPACE}}}

\def\qKKT{\textup{\textsc{quadratic-KKT}}\xspace}
\def\linearKKT{\textup{\textsc{2D-linear-KKT}}\xspace}

\newcommand{\xbar}{\overline{x}}
\newcommand{\zbar}{\overline{z}}
\newcommand{\fbar}{\overline{f}}
\newcommand{\circuit}{\mbox{\ensuremath{\mathcal{C}}}}
\newcommand{\circuitbar}{\overline{\circuit}}

\newcommand{\eps}{\varepsilon}
\newcommand{\reals}{\mathbb{R}}

\newcommand{\rationals}{\mathbb{Q}}

\newcommand{\sigmap}[1]{\sigma_{\textsf{#1}}}
\newcommand{\decode}{\textsc{Decode}\xspace}
\newcommand{\eval}{\textsc{Eval}\xspace}
\newcommand{\extract}{\textsc{ExtractBits}\xspace}
\newcommand{\bm}{\textsc{BitMultiply}\xspace}
\newcommand{\ctb}{\textsc{ContTimesBin}\xspace}

\newcommand{\affine}{\textsc{Affine}\xspace}

\newcommand{\ta}{\widetilde{a}}
\newcommand{\tg}{\widetilde{g}}
\newcommand{\tG}{\widetilde{G}}
\newcommand{\tf}{\widetilde{f}}

\definecolor{our-green}{rgb}{0.0,0.5,0.0}
\definecolor{our-yellow}{rgb}{1.0,0.75,0.0}

\tikzset{cstyle/.style={shape=circle,fill=black,scale=0.7}}

\title{The Complexity of Computing KKT Solutions of\\
Quadratic Programs\thanks{A preliminary version of this paper appeared at STOC 2024.}}

\author{
\begin{tabular}{cc}
& \\
\textbf{John Fearnley} & \textbf{Paul W. Goldberg}\\
\small{University of Liverpool, UK} & \small{University of Oxford, UK}\\
\href{mailto:john.fearnley@liverpool.ac.uk}{\small{\texttt{john.fearnley@liverpool.ac.uk}}} & \href{mailto:paul.goldberg@cs.ox.ac.uk}{\small{\texttt{paul.goldberg@cs.ox.ac.uk}}}\\
& \\
\textbf{Alexandros Hollender} & \textbf{Rahul Savani}\\
\small{University of Oxford, UK} & \small{Alan Turing Institute and University of Liverpool, UK}\\
\href{mailto:alexandros.hollender@cs.ox.ac.uk}{\small{\texttt{alexandros.hollender@cs.ox.ac.uk}}} & \href{mailto:rahul.savani@liverpool.ac.uk}{\small{\texttt{rahul.savani@liverpool.ac.uk}}}\\
& \\
\end{tabular}
}

\date{}
 
\begin{document}

\maketitle

\setcounter{page}{0}
\thispagestyle{empty}

\begin{abstract}
It is well known that solving a (non-convex) quadratic program is \np/-hard. We show that the problem remains hard even if we are only looking for a Karush-Kuhn-Tucker (KKT) point, instead of a global optimum. Namely, we prove that computing a KKT point of a quadratic polynomial over the domain $[0,1]^n$ is complete for the class \cls/ = \ppadpls/.
\end{abstract}

\newpage

\section{Introduction}

Quadratic programming (QP) is the problem of optimising a quadratic function of a collection of real variables, subject to linear constraints on those variables. It has widespread applications, numerous software implementations, and an extensive literature on its theoretical analysis, dating back more than 50 years. A fairly standard formulation is the following:
\begin{equation}\label{eq:qp}
\min_{x\in \mathbb{R}^n} f(x) := x^\intercal Qx + c^\intercal x  ~~~{\rm subject~to}~~~ a^\intercal_i x \leq b_i, \forall i \in \{1,\ldots ,m\}\ .
\end{equation}
In (\ref{eq:qp}) the matrix $Q$ is usually (and without loss of generality) taken to be symmetric, and there has been much work on restrictions of the problem based on assumed properties of $Q$, some of which we touch on below. The main result of the present paper is that for QP it is computationally hard to compute a {\em Karush-Kuhn-Tucker} (KKT) point, an important kind of solution consisting of one that is locally optimal with respect to gradient descent. Moreover, our hardness result applies to a special case of interest known as ``box constraints'' (e.g., \cite{AndjouhB22,BurerL09}), in which the feasible region (i.e., the region $a^\intercal_i x \leq b_i, \forall i \in \{1,\ldots ,m\}$) consists of an axis-aligned hypercuboid; here we use $[0,1]^n$. (Indeed, the general version of the problem sometimes assumes that the linear constraints include $x\in [0,1]^n$ (e.g., \cite{BellareR95}), which guarantees that the feasible region is compact.)

\begin{paragraph}{KKT solutions, and other types of solution.}
Informally, a KKT point is one that constitutes a local optimum of gradient descent. It may be a point at which the gradient is zero (a stationary point), or one where the gradient is non-zero, but further downhill progress is obstructed by one or more of the boundary constraints. A key feature of KKT solutions of (\ref{eq:qp}) is that they have concise certificates: roughly, the gradient together with the binding constraints at a point of interest. Moreover, if the domain is compact, there is guaranteed to be at least one KKT point, since the global optimum is a KKT point. These two observations indicate that the problem of searching for a KKT point belongs to the complexity class \tfnp/, search problems in which easily-checkable solutions must exist. In particular, this means that the problem is not expected to be \np/-hard, unless $\np/ = \conp/$~\cite{MegiddoP1991-tfnp}.

Apart from KKT points, the other main solution concepts for continuous optimisation problems are the following:
\begin{itemize}
\item Global optimum, $x$ for which $f(x)\leq f(x')$ for all $x'$ in the feasible region;
\item Stationary point (a.k.a.\ critical point), where $\nabla f(x)=0$;
\item Local minimum, $x$ where for some $\eps>0$ we have $f(x)\leq f(x')$ for all $x'$ within $\eps$ of $x$; at a {\em strict} local minimum, we have $f(x)< f(x')$ for all $x'$ in the feasible region within $\eps$ of $x$.
\end{itemize}
For QPs of the form (\ref{eq:qp}), stationary points are not guaranteed to exist, and global optima and local minima are \np/-hard to compute. KKT points relate to other solution concepts as follows. Global or local minima are KKT points, thus KKT points are at least as easy to compute as global/local minima. Any stationary point is a KKT point, but stationary points are not guaranteed to exist, even for the box-constrained feasible regions that we consider here. To see that stationary points can be searched for in polynomial time, note that they are given by $\nabla f = 0$, hence define a linear subspace, and checking for points in the subspace that satisfy the boundary constraints $a^\intercal_i x \leq b_i$ amounts to solving an LP. In the unconstrained case (in which there are no boundary constraints) stationary points are the same as KKT points, so then the problem of searching for either kind of solution is tractable.
A stationary point need not be a local minimum: for the problem of minimising $f(x)=-x^2$ over the interval $[0,1]$, $x=0$ is a stationary point but not a local minimum.
\end{paragraph}

\begin{paragraph}{Hardness results for global/local optima.}
There has been much work studying the circumstances in which one can efficiently compute a solution of one of the above kinds, also on determining whether a given point $x$ is one of the above solutions.

The global optimisation problem for quadratic programming is amongst the earliest problems to be shown \np/-hard \cite{Sahni72,Sahni74}\footnote{Indeed, even before \np/-completeness, \cite{MotzkinS65} reduce MAX CLIQUE to a version of QP in which the feasible region is a simplex. In Sahni's reduction (from SUBSET SUM) the feasible region is a box.}, although containment in \np/ had to wait until substantially later \cite{Vavasis90-QP-in-NP}\footnote{\cite{Vavasis90-QP-in-NP} showed in particular that there exist optimal solutions having polynomial bit complexity.}. \np/-hardness has also been established for various restrictions of the problem, for example \cite{PardalosV91} obtain \np/-hardness when $Q$ has rank 1 with one negative eigenvector (in a sense the simplest kind of nonconvex program that can be expressed as a QP). A simple reduction from MAX-CLIQUE to QP \cite{MotzkinS65,PardalosYH91,AhmadiZ22-polytope} yields \np/-hardness for QPs that are square-free quadratic forms (diagonal entries of the matrix $Q$ in (\ref{eq:qp}) are zero, and there are no linear terms $c^\intercal x$); the feasible region is a simplex as opposed to a box.

Regarding local optima, there are hardness results known for computing them, as well as for checking whether a given point is locally optimal. It is shown in \cite{AhmadiZ22-polytope} that it is hard to find an approximation to a local minimum. The problems of checking whether a given point $x$ is a local optimum, or a strict local optimum, are \np/-hard \cite{MurtyK1987,PardalosS88} (in particular, when the feasible region is the unit box $[0,1]^n$ and $x$ is the origin). In the unconstrained case ($m=0$) \cite{AhmadiZ22-unconstrained} show that it is possible in polynomial time to determine whether any version of local optimal solution exists.

There are strong hardness results even for computing approximations to the global optimum of QP. \cite{BellareR95} obtained hardness of approximation for QP by reducing a two-prover one-round proof (with polylogarithmic communication) to QP in quasi-polynomial time. From this it follows that assuming \np/ problems are not solvable in quasi-polynomial time, there is no non-trivial constant factor approximation algorithm for QP. \cite[Theorem 1.3]{BellareR95} also show that for some constant $\mu\in(0,1)$, QP is \np/-hard to approximate within a factor $\mu$; these hardness results even apply assuming numbers are given in unary. Also in the context of two-prover one-round interactive proof systems, \cite{FeigeL92} show how the search for an optimal strategy for the provers can be expressed as a QP. They study a relaxed version of the QP that corresponds with an upper bound on the value of the game played by the provers (and is poly-time computable); this leads to a general-purpose heuristic for problems in \np/, and in turn a general kind of algorithm for diverse problems in \p/. \cite{FeigeK94} (Corollary 4) use this to conclude that (unless $\p/ = \np/$) there is no constant-factor approximation algorithm for QP.
\end{paragraph}

\subsection{\np/ total search problems and the class \cls/}

As a solution concept, KKT points have two appealing properties: guaranteed existence (provided the feasible region is bounded), and polynomial-time checkability (we can efficiently verify that a point is KKT). These properties mean that the problem of {\em computing} one belongs to the complexity class \tfnp/: total (as opposed to partial) functions that belong to \np/. Problems that belong to \tfnp/ are classified by various syntactic subclasses associated with the proof principle underlying the existence guarantee. Here, the relevant classes are \pls/ \cite{JPY88}, \ppad/ \cite{Pap94}, and \cls/ \cite{DaskalakisP2011-CLS}, the latter having been shown to be equal to \ppadpls/ \cite{FGHS22}.

The complexity class \cls/ (for ``continuous local search'') was introduced by \cite{DaskalakisP2011-CLS} in an effort to understand the complexity of certain seemingly-hard search problems that belong to both \ppad/ and \pls/. The problems they list include the search for a KKT point of a given polynomial over a domain given by linear constraints. Such problems are unlikely to be complete for either \ppad/ or \pls/, since such a result would indicate that one of \ppad/ or \pls/ contains the other. Recently \cite{FGHS22} showed that \cls/ is equal to the intersection of \ppad/ and \pls/, in the process showing that it is \cls/-complete to find KKT solutions of piecewise polynomial functions defined by a certain class of arithmetic circuits. Building on these results, \cite{BabR21} showed that computing a (possibly mixed) Nash equilibrium of a congestion game is \cls/-complete and furthermore (of more relevance to the present paper) that local optimisation (in the KKT sense) of degree-5 polynomials is also \cls/-complete. Since the \cls/-completeness results of \cite{FGHS22,BabR21}, other problems in game theory have been shown \cls/-complete \cite{ElkindGG22,TOCG23} via comparatively direct reductions. The main result of \cite{FGHS22} indicates that \cls/-complete problems are unlikely to have polynomial-time algorithms. Moreover, the hardness of \cls/ can also be based on various cryptographic assumptions such as particular versions of indistinguishability obfuscation~\cite{HubacekY2017-CLS}, soundness of Fiat-Shamir~\cite{ChoudhuriHKPRR19-Fiat-Shamir}, or Learning With Errors~\cite{JawaleKKZ21-PPAD-LWE}.

\begin{paragraph}{}
Regarding our main result and its significance, we have noted that quadratic programming is a fundamental problem of general interest. Our main result answers an open question raised in \cite{PardalosV92} (see problem 3) and reiterated in \cite{AhmadiZ22-polytope}. The problem {\sc Polynomial KKT} \cite{DaskalakisP2011-CLS} is a generalization of (\ref{eq:qp}) in which $f$ is allowed to be any polynomial, written down as a sum of monomials. \cite{BabR21} showed that the {\sc Polynomial KKT} problem is \cls/-complete for degree-5 polynomials, which naturally raises the question, pointed out in \cite{FGHS22}, of whether such a result holds for lower degree. Here we identify the lowest degree for which a hardness result holds, since for degree 1 the problem is linear programming. On the other hand, our result does not hold for some versions of interest, such as taking a standard simplex as the feasible region,\footnote{In subsequent work, our result was extended to also hold for such simplex constraints, by presenting a direct reduction from box constraints to simplex constraints~\cite{GhoshH24-symmetric-common-payoff}.} e.g., \cite{BomzeDKRQT00}. Another important class of QPs that differs from the one studied here involves optimising quadratic functions over a unit sphere, or an intersection of spheres, e.g., \cite{Ye92-sphereQP,VavasisZ90-sphereQP,Ye99}.

Our main result is for the computation of {\em exact} (as opposed to approximate) solutions. Fortunately, any problem instance has rational-valued KKT solutions whose bit complexity is polynomial in the bit complexity of numbers appearing in the problem instance\footnote{This does not hold for objective functions of degree 3 or more. The distinction is analogous to the distinction between Nash equilibria of 2-player games versus 3-player games.}. If we consider natural notions of approximation, computation of exact solutions is polynomial-time equivalent to the computation of $\eps$-approximate solutions for inverse-exponential $\eps$.
\cite{Ye98} gives an algorithm that computes $\eps$-KKT points, whose runtime dependence on $\eps$ is $O((1/\eps)\log(1/\eps)\log(\log(1/\eps)))$ (there is also polynomial dependence on $n$, and a factor representing the difference between the maximal and minimal objective values). So we give a negative answer to the question of whether a logarithmic dependence on $\eps$ is possible. Finally, our hardness result also highlights a contrast with convex optimisation, in which KKT points and global optima coincide, and many algorithms are known that find $\eps$-approximate solution in time $O(\log (1/\eps)$) \cite{Bubeck14}.
\end{paragraph}

\begin{theorem*}[Main Result]
It is \cls/-complete to compute KKT solutions of (\ref{eq:qp}), even when the feasible region consists of the unit box $[0,1]^n$.
\end{theorem*}

\subsection{Technical Overview}

Our result is proved in two steps. First, we present a reduction from the problem of computing (some sort of) a KKT point of a type of arithmetic circuit to the problem of computing a KKT point of a QP with box constraints. Namely, we consider {\em linear arithmetic circuits} that compute piecewise linear functions using a single kind of (fairly general, multi-purpose) gate. In the second step, we show that computing a KKT solution of such a circuit is a \cls/-hard problem by reducing from a version of the problem for more general circuits, that is known to be \cls/-hard~\cite{FGHS22}. Together, these two reductions establish the \cls/-hardness of computing a KKT point of a QP.

While the first step is certainly the most interesting part of this combined reduction, the second step is surprisingly technical and requires a certain number of new ideas, which are likely to be useful in future works. We now present the main challenges in both parts, as well as the new ideas that were needed to overcome them.

\subsubsection*{Step 1: Reducing from a circuit to a QP}

The first challenge in this part is the following main obstacle.

\paragraph{\bf Challenge 1: We can only use terms of degree at most two.}

Unsurprisingly, the techniques used in \cite{FGHS22} to show \cls/-hardness of finding a KKT point of a general arithmetic circuit are of no use here, since we are reducing to an explicit polynomial. Rather, just as in \cite{BabR21}, we will just use the result of \cite{FGHS22} as a starting point for reductions.

The techniques used in \cite{BabR21} to reduce to a degree-5 polynomial are much more relevant here. Their reduction is highly non-trivial and also relies on some older ideas used in the context of proving \pls/-hardness of a version of local-max-SAT~\cite{Krentel90}. However, the restriction here to degree-2 polynomials makes a big difference and we mostly cannot re-use their ideas.\footnote{One notable exception to this is the idea of simulating the evaluation of a circuit by constructing an objective function that consists of a sum of terms, one for each gate of the circuit, with gates deeper down in the circuit having smaller weights. This idea of exponentially-decreasing weights, already used in \cite{Krentel90} in the context of discrete local optimisation, ensures that gates are correctly simulated and that their output is not biased by other gates that use it as an input.} We encounter a fundamental obstacle to the use of \emph{guide variables} (called guide players in \cite{BabR21}, since their reduction is presented in terms of a game). Very briefly, the role of these guide variables, which were already used in \cite{Krentel90} in a somewhat simpler form, is to be able to ``deactivate'' some interactions between two (or more) other variables. This deactivation is absolutely crucial for the approach of \cite{BabR21}, as it already was in \cite{Krentel90}. Given that in any reasonable construction of a quadratic polynomial the interaction between two variables would yield a quadratic term (e.g., $x_ix_j$, or perhaps $(x_i-x_j)^2$), the addition of a guide variable on top of that immediately takes us up to degree 3.

The inability to use guide variables, or any of the other involved machinery from \cite{BabR21}, forces us to start from the ground up. As a toy example to illustrate our approach, consider a circuit~$\circuit$ that takes two inputs $x_1, x_2 \in [0,1]$ and consists of only two gates. The first gate computes $x_3 := x_1 + x_2$, and then the second gate, which is the output of the circuit, computes $x_4 := -2x_3$. Thus, the circuit~$\circuit$ simply computes the function $f: [0,1]^2 \to \mathbb{R}, (x_1,x_2) \mapsto -2(x_1+x_2)$. This is a linear function and so finding a KKT point over the simple domain $[0,1]^2$ is very easy: the only KKT point (for the minimization problem) is at $(1,1)$. Now let us attempt to simulate this circuit by a quadratic polynomial that implements each gate separately.\footnote{Obviously, there is a trivial reduction here that just lets the quadratic polynomial be the linear function $f$ itself. However, we are interested in a construction that implements each gate separately, because we will ultimately need to implement (slightly) more general gates. Indeed, a circuit that only consists of linear gates represents a linear function, and it is easy to find KKT points of such functions.} Consider the polynomial
$$p(x_1,x_2,x_3,x_4) := (x_3-x_1-x_2)^2 + (x_4 + 2x_3)^2$$
which consists of one squared term for each gate.\footnote{Because we use such squared terms, our polynomial will not be multilinear. In particular, our result has no implications for games, unlike \cite{BabR21}.} Intuitively, minimizing $p$ will force $x_3 = x_1+x_2$ and $x_4 = -2x_3$.
To simplify the exposition in this part of the overview, we think of $x_3$ and $x_4$ as being unconstrained, while $x_1, x_2 \in [0,1]$. This makes the KKT conditions for $x_3$ and $x_4$ easier to work with,\footnote{The arguments still work if $x_3$ and $x_4$ are constrained to lie in some sufficiently large interval. For example, we could impose the constraints $x_3 \in [-10,10]$ and $x_4 \in [-30,30]$, since this ensures that $x_3$ and $x_4$ never lie on the boundary at a solution. In any case, we will later revert to $[0,1]$ constraints for all variables and these will indeed be used in a very crucial way in our construction.} as they just correspond to $\frac{\partial p}{\partial x_3} = 0$ and $\frac{\partial p}{\partial x_4} = 0$.
Now, the partial derivative of $p$ with respect to $x_4$ is
$$\frac{\partial p}{\partial x_4} = 2(x_4+2x_3)\ .$$
At a KKT point this must be zero, so we obtain $x_4 = -2x_3$. Next, we have
$$\frac{\partial p}{\partial x_3} = 2(x_3-x_1-x_2) + 4(x_4+2x_3)\ .$$
Setting this to zero, and using $x_4 = -2x_3$, we obtain $x_3 = x_1+x_2$ as desired. Thus, any KKT point $(x_1,x_2,x_3,x_4)$ of $p$ satisfies $x_4 = f(x_1,x_2)$. In other words, we have correctly simulated the evaluation of the circuit. However, this is not enough. We want any KKT point $(x_1,x_2,x_3,x_4)$ of $p$ to yield a KKT point $(x_1,x_2)$ of $f$, and this is currently not the case. What is missing is that the QP is not ``aware'' of the fact that it should attempt to minimize the output of the circuit, namely~$x_4$. An initial attempt to fix this by redefining
$$p(x_1,x_2,x_3,x_4) := (x_3-x_1-x_2)^2 + (x_4 + 2x_3)^2 + x_4$$
fails because it introduces big errors in the evaluation of the gates. This can be mitigated by using the idea of exponentially-decreasing weights from \cite{Krentel90} (which was also heavily used in \cite{BabR21}). Indeed, we can instead define
$$p(x_1,x_2,x_3,x_4) := (x_3-x_1-x_2)^2 + \delta(x_4 + 2x_3)^2 + \delta^2x_4$$
where $\delta > 0$ is small. Performing the same analysis as above, yields $x_4 = -2x_3 - \delta/2$, and then $x_3 = x_1+x_2 + \delta^2$. So the gates are no longer evaluated exactly, but have some additive error. Fortunately, the error can be made arbitrarily small by making $\delta$ sufficiently small.

The interesting observation however is that
$$\frac{\partial p}{\partial x_1} = -2(x_3-x_1-x_2) = -2 \delta^2$$
and similarly $\frac{\partial p}{\partial x_2} = -2 \delta^2$. This forces any KKT point $(x_1,x_2,x_3,x_4)$ of $p$ to satisfy $(x_1,x_2) = (1,1)$, which is indeed the correct KKT point of the original function $f$! Moreover, notice that $-2 \delta^2$ is equal to $\delta^2 \frac{\partial f}{\partial x_1}$, i.e., it is proportional to the partial derivative of the original function $f$. In other words, this seemingly arbitrary error term actually carries useful information. This is not a coincidence. The errors, starting from the error in the output gate due to the new term $\delta^2x_4$, propagate backwards in the circuit evaluation, until they reach the input variables $x_1$ and $x_2$. In doing so, every traversed gate modifies the error in a very particular way, until, finally, the signal seen by the input variables corresponds to the gradient of the original function $f$. Indeed, at every gate the error is modified in a way that corresponds to applying the rules for computing the gradient of a circuit using the \emph{backpropagation} technique. This technique, widely used in machine learning, computes the gradient of a function by starting from the output and repeatedly applying the chain rule for differentiation until the inputs are reached. Indeed, the following can be proved by induction over the depth of the circuit:

\begin{lemma*}[Linear Backpropagation Lemma]
When $p$ is constructed from a depth-$m$ circuit $\circuit$ with linear gates computing a function $f$, any KKT point of $p$ satisfies
$$\frac{\partial p}{\partial x_i} = \delta^m \cdot \frac{\partial f}{\partial x_i}$$
for all input variables $x_i$.
\end{lemma*}

\paragraph{\bf Challenge 2: Circuits with linear gates are easy.}

The Linear Backpropagation Lemma implies that any KKT point of $p$ must yield a KKT point of the original function $f$. Unfortunately, this is not enough to prove that our problem is intractable, because computing a KKT point of a linear function is an easy problem. In order to reduce from existing hard functions~\cite{FGHS22,BabR21} we would at least need the circuit $\circuit$ to also consist of multiplication gates $x_k := x_ix_j$. But to implement such a gate we would need terms of the form $(x_k - x_ix_j)^2$, which have degree four.

The crucial observation here is that we have not yet used the boundary of the domain in any way to implement gates. The boundary constraints suggest a natural generalization of linear gates. Indeed, if we consider a term such as $(x_3-x_1-x_2)^2$ and now -- unlike we did before -- also constrain $x_3 \in [0,1]$, then we see that any KKT point of this term must satisfy
$$x_3 = \trunc (x_1 + x_2)$$
where $\trunc: \reals \to [0,1]$ denotes truncation to the $[0,1]$ interval, i.e., $\trunc(z) = \min\{1,\max\{0,z\}\}$. More generally, we can simulate any gate of the form $x_k := \trunc(ax_i+bx_j+c)$ by the term $(x_k-ax_i-bx_j-c)^2$. We call such a gate a \emph{truncated linear gate}. No efficient algorithm for computing KKT points of such circuits seems to be known, so there is hope that we might be able to prove intractability.

Unfortunately, before we can start considering proving such an intractability result, there is a more pressing issue: this reduction only works for a single gate. Although we still have correct (approximate) evaluation of the circuit by picking a sufficiently small $\delta$, the errors no longer correctly simulate backpropagation when truncation occurs. In order to restore the behavior that we observed in the setting without truncation, we need to find a way to simulate truncated linear gates that also works with backpropagation.

We modify the simulation as follows. A truncated linear gate $x_k := \trunc(ax_i+bx_j+c)$ is now simulated by the term
$$(x_k + z^+ - z^- - ax_i - bx_j - c)^2 +2 z^+ z^- + 2z^+(1-x_k) +2z^-x_k$$
where $z^+$ and $z^-$ are new auxiliary variables. Intuitively, $z^+$ is here to ``pick up the slack'' between $x_k$ and $ax_i+bx_j+c$, when the latter is strictly larger than $1$. The variable $z^-$ has a similar function when $ax_i+bx_j+c < 0$. Note that the derivative of the new term with respect to $x_k$ is the same as the derivative of $(x_k-ax_i-bx_j-c)^2$. So from the point of view of $x_k$ this new term is the same as the old one; in particular, $x_k$ will again (approximately) take the value $\trunc(ax_i+bx_j+c)$. The difference is in what the variables $x_i$ and $x_j$ see. The derivative of the old term with respect to $x_i$ was just $-2a(x_k-ax_i-bx_j-c)$, so when truncation occurred this derivative would essentially correspond to the truncation gap, whereas we would want it to be $0$, which is the correct backpropagation signal (because a small change in $x_i$ would not change $x_k$ if truncation occurs). On the other hand, the derivative of the new term with respect to $x_i$ is $-2a(x_k + z^+ - z^- - ax_i - bx_j - c)$. In this derivative, the variables $z^+$ and $z^-$ fill the gap between $x_k$ and $ax_i+bx_j+c$ when truncation occurs, and thus we obtain the correct backpropagation signal $0$. If truncation does not occur, then $z^+ = z^- = 0$ and the new term behaves just like the old term.

It is thus tempting to try to establish the following lemma.

\begin{desiredlemma*}[Ideal Backpropagation Lemma]
When $p$ is constructed from a depth-$m$ circuit~$\circuit$ with truncated linear gates computing a function $f$, any KKT point of $p$ satisfies
$$\left(\frac{\partial p}{\partial x_1}(x_1,x_2), \frac{\partial p}{\partial x_2}(x_1,x_2)\right) \in \delta^m \cdot \partial f(x_1,x_2)$$
where $x_1$ and $x_2$ are the input variables, and $\partial f$ is the generalized gradient\footnote{At a point where $f$ is differentiable, the generalized gradient is the singleton set consisting of the gradient; where $f$ is not differentiable, it is the set of convex combinations of well-defined gradients close to that point.} of the (almost everywhere differentiable) function $f$.
\end{desiredlemma*}

\paragraph{\bf Challenge 3: The Ideal Backpropagation Lemma does not hold.}

Unfortunately, this lemma fails to hold. The reason for this is quite fundamental: backpropagation does not really work for such circuits. Consider the following simple example: the circuit $\circuit$ has a single input $x_1$, and outputs the value $\trunc[2x_1]/2 + \trunc[x_1-1/2]$. This can easily be implemented by using three truncated linear gates. Note that the circuit computes the (linear!) function $[0,1] \to [0,1], x_1 \mapsto x_1$, which has derivative $1$ everywhere. Now consider performing backpropagation when the input is $x_1=1/2$. Since this is the threshold for both truncations, the value $0$ is a valid derivative for both of those gates. As a result, the gradient computed by backpropagation could theoretically output $0$. In \cref{sec:lacqp} we provide a slightly more involved example (\cref{ex:backprop-counter}) where this indeed happens in our QP construction: the correct gradient value is $1/2$ everywhere, but at $x_1 \approx 1/2$ we have $\frac{\partial p}{\partial x_1}(x_1) = 0$. In particular, our construction would incorrectly output $x_1 \approx 1/2$ as a KKT point. This shows that the Ideal Backpropagation Lemma cannot hold, even in some approximate version where we would also consider points in the vicinity of $(x_1,x_2)$.

However, the example suggests that a weaker statement might hold. Indeed, the issue occurs because both truncations have a threshold at $1/2$. If we were to slightly perturb these thresholds, then we would indeed see a derivative of $0$ appear. In other words, it seems reasonable to think that the backpropagation that occurs computes some convex combination of gradients of various perturbed versions of the circuit $\circuit$. To be more precise, in a perturbed version of $\circuit$, every gate $x_k := \trunc(ax_i+bx_j+c)$ is replaced by a gate $x_k := \trunc(ax_i+bx_j+c + \pi_k)$ for some $\pi_k \in \mathbb{R}$. By picking~$\delta$ sufficiently small, we can ensure that all $\pi_k$ are as small as required. Indeed, we can prove the following result.

\begin{lemma*}[Backpropagation Lemma (informal)]
When $p$ is constructed from a depth-$m$ circuit $\circuit$ with truncated linear gates, any KKT point of $p$ satisfies
$$\left(\frac{\partial p}{\partial x_1}(x_1,x_2), \frac{\partial p}{\partial x_2}(x_1,x_2)\right) \in \delta^m \cdot \textup{conv}\left\{ \nabla \tilde{f} (x_1,x_2) : \text{$\tilde{f}$ computed by small perturbation of $\circuit$} \right\}$$
where $x_1$ and $x_2$ are the input variables.
\end{lemma*}

This leads us to define the new notion of a generalized gradient \emph{of a linear arithmetic circuit} to capture this behavior. See the subsequent preliminaries section for more details on this.

With the Backpropagation Lemma in hand, we can now leave quadratic polynomials behind us and focus on circuits with truncated linear gates. We will also refer to these by the simpler name \emph{linear arithmetic circuits} (which is usually used to refer to circuits with $+, -, c, \times c, \min, \max$ gates~\cite{FGHS22}), since our circuits can easily simulate such circuits and vice-versa. In the next step, we construct a class of such circuits that is robust to perturbations, and for which it is \cls/-hard to find a KKT point (with respect to the new definition of generalized gradient).

\subsubsection*{Step 2(a): Designing a robust function: the mesa construction}

In order to show \cls/-hardness of this problem, we have to reduce from an existing \cls/-hard problem. We reduce from the problem of computing an approximate KKT point of a smooth function defined on the two-dimensional grid $[0,1]^2$, which is known to be \cls/-complete when the function is represented by an arithmetic circuit with more general gates~\cite{FGHS22}. Indeed, being able to work on a two-dimensional domain (as opposed to a high-dimensional one if we used \cite{BabR21} instead) allows us to avoid having to unnecessarily complicate the construction.

At a high level, given such a smooth function defined on $[0,1]^2$, we would like to construct a piecewise linear function that has (approximately) the same KKT solutions as the original function. Importantly, our piecewise linear function must be represented by a linear arithmetic circuit. This is already challenging, even if we ignore the perturbations.

\paragraph{\bf Challenge 1: Existing linear circuit constructions give no guarantees about the gradient.}

Interpolations of continuous functions by linear arithmetic circuits have been given in prior work~\cite{CDT09,Rubinstein18-Nash-inapproximability,FGHS22} and indeed these have been pivotal in proving important results in this field. However, these constructions only aim to obtain a piecewise linear function that closely approximates the original function in terms of \emph{function value}. The usage of the \emph{averaging trick}~\cite{DGP09}, which is common to all of these, means that the (generalized) gradient of the piecewise linear function can be wildly different from the original gradient and introduce spurious KKT solutions.

We are thus forced to move away from this type of construction. Putting circuits aside for a bit, there is a standard interpolation by a piecewise linear function that does (approximately) maintain the gradient. Simply pick a sufficiently fine standard triangulation of $[0,1]^2$, define the value of the function at the vertices of the triangulation to agree with the original function, and interpolate linearly within each triangle. This interpolation would be sufficient for our purposes, because it is not hard to show that any KKT point of the new function must correspond to an approximate KKT point of the original function.

The ``catch'' is that we would have to construct a linear arithmetic circuit that represents this piecewise linear interpolation. Existing techniques, which all use the averaging trick, are unable to achieve this. However, it turns out that using some new ideas (most of which we end up using in our final construction and which are highlighted below) it is in fact possible to construct a linear arithmetic circuit that \emph{exactly} computes this piecewise linear interpolation. At this point, if the Ideal Backpropagation Lemma stated earlier was true, we would be done. Unfortunately, this is not the case, and we also have to argue about what happens to the circuit when gates are slightly perturbed.

\paragraph{\bf Challenge 2: The standard piecewise linear interpolation is not robust to perturbations.}

Unfortunately, this circuit is not robust to perturbations, i.e., perturbed versions of the circuit would introduce new solutions that did not appear in the original function. In order to illustrate the issues that can occur, as well as explain how they can be overcome, it is useful to take a step back and think about what would happen if the domain was one-dimensional, instead of two-dimensional.

\begin{figure}
\begin{center}
\resizebox{0.4\linewidth}{!}{

\def\axlabelstretch{1.2}

\begin{tikzpicture}[
	dot/.style = {circle, fill, minimum size=1mm, inner sep=0pt, outer sep=0pt},
	sq/.style = {rectangle, fill, minimum size=1.22mm, inner sep=0pt, outer sep=0pt}
]

\tikzmath{
	\base = 1.15;
	\interc = 9.5;
	\startx = -1;
	\starty = 0;
	\widthsteep = 0.65;
    \widthtop = 4.5;
	\redxc = \startx + \widthtop; 
	\bluexc = \startx + 2*\widthtop; 
	\greenxc = \startx + 3*\widthtop; 
	\redx0= \redxc - 0.5 * \widthtop - \widthsteep;
	\redx1= \redxc - 0.5 * \widthtop;
	\redx2= \redxc + 0.5 * \widthtop;
	\redx3= \redxc + 0.5 * \widthtop + \widthsteep;
	\bluex0= \bluexc - 0.5 * \widthtop - \widthsteep;
	\bluex1= \bluexc - 0.5 * \widthtop;
	\bluex2= \bluexc + 0.5 * \widthtop;
	\bluex3= \bluexc + 0.5 * \widthtop + \widthsteep;
	\greenx0= \greenxc - 0.5 * \widthtop - \widthsteep;
	\greenx1= \greenxc - 0.5 * \widthtop;
	\greenx2= \greenxc + 0.5 * \widthtop;
	\greenx3= \greenxc + 0.5 * \widthtop + \widthsteep;
	\redyl = \interc-\base^\redx1; 
	\redyr = \interc-\base^\redx2; 
	\blueyl = \interc-\base^\bluex1;
	\blueyr = \interc-\base^\bluex2;
	\greenyl = \interc-\base^\greenx1;
	\greenyr = \interc-\base^\greenx2;
}

\begin{axis}[
    	 xshift=4cm,
         scale=0.5,
		 xmin=0,ymin=0,xmax=17,ymax=10, 
		 samples=100,
		 axis lines=center,
		 grid=both,
		 xtick = \empty,
		 ytick = \empty,
		]


\coordinate (A) at (\bluex0,\starty);
\coordinate (B) at (\bluex1,\blueyl);
\draw[thick,blue,dotted] ($(A)!-.7!(B)$) coordinate (A') -- ($(B)!-.7!(A)$) coordinate (B');
\draw[thick,blue] (\bluex0,\starty) -- (\bluex1,\blueyl);

\coordinate (C) at (\bluex1,\blueyl);
\coordinate (D) at (\bluex2,\blueyr);
\draw[thick,blue,dotted] ($(C)!-.7!(D)$) coordinate (C') -- ($(D)!-.7!(C)$) coordinate (D');
\draw[thick,blue] (\bluex1,\blueyl) -- (\bluex2,\blueyr);

\coordinate (E) at (\bluex2,\blueyr);
\coordinate (F) at (\bluex3,\starty);
\draw[thick,blue,dotted] ($(E)!-.7!(F)$) coordinate (E') -- ($(F)!-.7!(E)$) coordinate (F');
\draw[thick,blue] (\bluex2,\blueyr) -- (\bluex3,\starty);

\end{axis}

\end{tikzpicture}}
\end{center}
\caption{Creating a one-dimensional mesa as the minimum of three lines.}
\label{fig:1dmesa}
\end{figure}
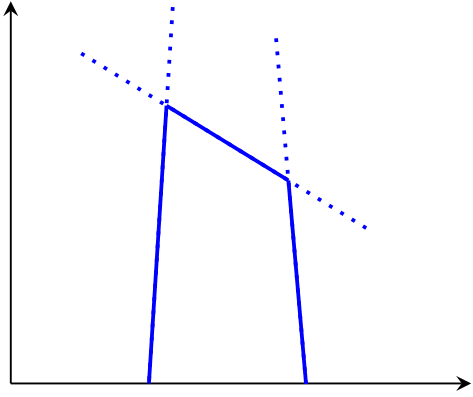

\begin{figure}
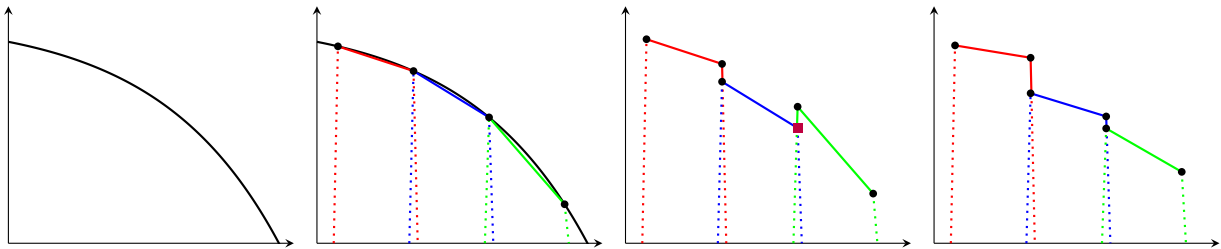

\begin{center}
\resizebox{1\linewidth}{!}{\input figures/1d.tex}
\end{center}
\caption{Left to right: smooth function we want to interpolate (with no stationary points); 
	standard interpolation as the maximum of adjacent mesas without perturbations (still no stationary points); additive perturbations
	introduce an unwanted stationary point (indicated as a brown square); this unwanted stationary point does not arise when the 
	halved-gradient trick is applied.}
\label{fig:halfgrad}
\end{figure}

In the one-dimensional case, with domain $[0,1]$, the standard interpolation is indeed very simple. Given some sufficiently fine discretization of $[0,1]$, we simply interpolate linearly between adjacent points. Implementing this using a linear arithmetic circuit is still non-trivial (because the discretization is exponentially fine), but it is possible using the new ideas hinted at above. Without going into too much detail, this piecewise linear function is constructed by taking the maximum of an exponential\footnote{Of course this would not be efficient, so the actual construction has to do something more clever. Nevertheless, this (incomplete) description is sufficient for this part of the technical overview.} number of simple functions. Every simple function implements a single linear segment of the interpolation and then very quickly decreases in value as soon as we leave the segment. Each of the simple functions can be constructed by taking the minimum of the three linear functions that they each consist of, see \cref{fig:1dmesa}. Taking the maximum of all these simple functions then indeed correctly yields the desired interpolation. 

Now, let us see what happens when we allow small perturbations at gates of the circuit. Of course, it is hard to think about all the ways in which perturbations can interfere with a construction (especially since we have not given many details here), but a good starting point, and in a certain sense, the absolute minimum requirement is: if we slightly perturb each of the linear functions that are used to construct each of the simple functions by some small additive term, no new solutions should appear.

Unfortunately, the construction already fails this simple test, as can be seen in the example in \cref{fig:halfgrad}. Furthermore, note that the issue appears as soon as we allow non-zero perturbations; requiring the perturbations to be very small does not help. However, this example suggests a somewhat different approach: instead of always trying to faithfully approximate the gradient of the original function, we can relax this requirement as long as we do not introduce any new KKT points. Indeed, we can avoid the issue by using the following ``halved-gradient'' trick: for each linear segment, halve its slope while keeping the same value at the middle of the segment. As long as the perturbations are kept sufficiently small, this simple trick ensures that linear pieces with ``bad'' gradients are no longer visible, i.e., they disappear when we take the maximum over all simple functions. See \cref{fig:halfgrad} for an example. The only case where such a ``bad'' piece might appear is if the slope of the linear segment is very flat. But in that case, the original function must have an approximate KKT point there.

This approach indeed works for the one-dimensional setting. It turns out that the easiest way to generalize this to two dimensions is not to try to apply this to a triangulation of the unit square, but rather to a grid over the unit square. For any point on the grid, we construct a corresponding square segment with a gradient that is half the gradient of the original function at that point. When we leave that square, the function value decreases very quickly. We call this simple function a mesa due to its shape which is reminiscent of a flat-topped hill with steep sides, see \cref{fig:2dmesa}. The final function is then obtained by taking the maximum of (an exponential number of such) mesa functions, one at each grid point.

\begin{figure}
\begin{center}
\includegraphics[width=8cm]{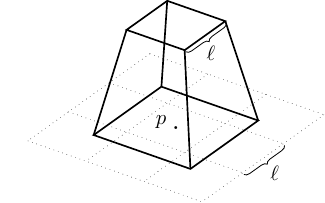}
\end{center}
\caption{An illustration of a two-dimensional mesa function centered at point $p$.}
\label{fig:2dmesa}
\end{figure}

We show that this mesa construction does not introduce any new KKT points, except in the vicinity of approximate KKT points of the original function. Importantly, this continues to hold even if we add an arbitrary, but small, perturbation to each linear piece of each mesa.

In the last part of this technical overview, we give some details about how the mesa construction can be implemented by using only the gates available in a linear arithmetic circuit (namely, truncated linear gates, as well as other gates that can be simulated by them, such as $\max$ and $\min$). Furthermore, we need to make sure that the perturbations, which appear in any gate, can only impact the construction in the way described above (i.e., perturbing each linear piece of each mesa), and not in any other way.

\subsubsection*{Step 2(b): Robustly implementing the mesa construction with a linear circuit}

From Step 2(a) we get a grid of points $G$ covering $[0, 1]^2$, and Boolean
circuits that, for each point $y \in G$ define a mesa centered at $y$. If $m(x,
y)$ is the height of the mesa centered at $y$ at the point $x$, then 
we must implement a linear circuit that computes 
\begin{equation*}
f(x) = \max_{y \in G}\ m(x, y).
\end{equation*}

Since $G$ contains exponentially many points, it is clearly infeasible to
compute $f$ directly. Instead we 
first perform a bit extraction on $x$ to obtain binary encodings of a small set
$S \subseteq G$ of
nearby grid points, and we then evaluate  $f(x) = \max_{y \in S} m(x, y)$
instead. This works because each mesa can only achieve the maximum used in $f$
in a small radius around its center, so we can disregard the mesas that are far
from $x$ when computing $f$.

While the technique of extracting bits from $x$ to succinctly compute some
function $f(x)$ has been used before, we must overcome several challenges to
make this work for our setting.

\paragraph{\bf Challenge 1: Dealing with bit extraction failures.}

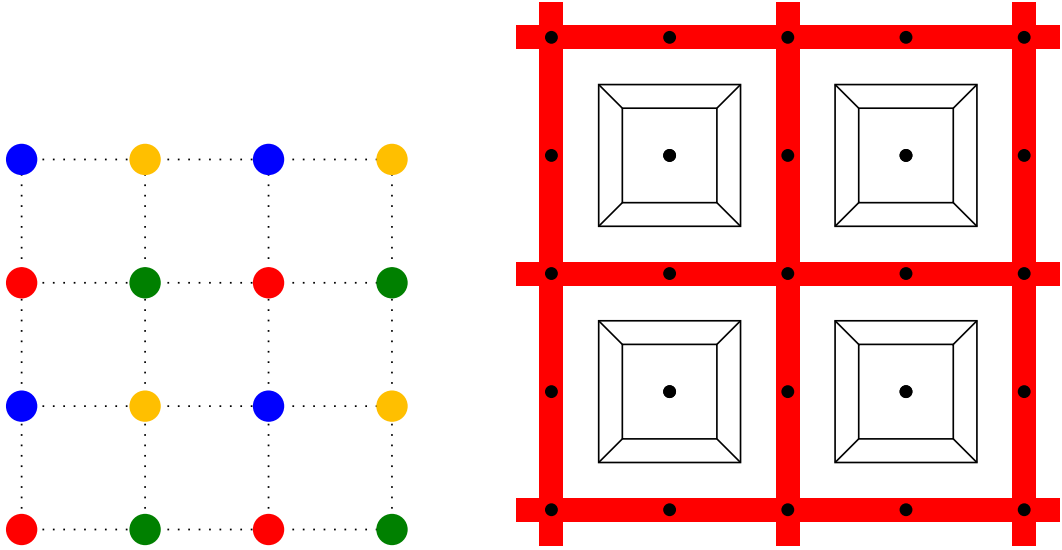
\begin{figure}
\begin{center}
\resizebox{0.4\linewidth}{!}{\def\blueredrows{1,5}
\def\blueyellowcols{1,5}
\def\yellowgreenrows{3,7}
\def\redgreencols{3,7}

\tikzset{cstyle/.style={shape=circle,fill=black,scale=0.7}}

\begin{tikzpicture}[scale=1]

\draw[step=1.0,black,thin,dotted] (0,0) grid (3,3);

\node[cstyle, red] at (0,0){};
\node[cstyle, blue] at (0,1){};
\node[cstyle, our-green] at (1,0){};
\node[cstyle, our-yellow] at (1,1){};
\node[cstyle, red] at (2,0){};
\node[cstyle, blue] at (2,1){};
\node[cstyle, our-green] at (3,0){};
\node[cstyle, our-yellow] at (3,1){};
\node[cstyle, red] at (0,2){};
\node[cstyle, blue] at (0,3){};
\node[cstyle, our-green] at (1,2){};
\node[cstyle, our-yellow] at (1,3){};
\node[cstyle, red] at (2,2){};
\node[cstyle, blue] at (2,3){};
\node[cstyle, our-green] at (3,2){};
\node[cstyle, our-yellow] at (3,3){};

\end{tikzpicture}}
\hskip 1cm
\resizebox{0.45\linewidth}{!}{\tikzset{cstyle/.style={shape=circle,fill=black,scale=0.6}}

\begin{tikzpicture}[scale=2]


\foreach \x in {0,...,4}
\foreach \y in {0,...,4}
{
	\node[cstyle] at (\x,\y) {};
}


\def\e{0.1}
\def\extra{0.2}

\fill[fill=red,opacity=1] (0-\e,0-\e-\extra) rectangle (0+\e,4+\e+\extra);

\fill[fill=red,opacity=1] (2-\e,0-\e-\extra) rectangle (2+\e,4+\e+\extra);

\fill[fill=red,opacity=1] (4-\e,0-\e-\extra) rectangle (4+\e,4+\e+\extra);

\fill[fill=red,opacity=1] (0-\e-\extra,0-\e) rectangle (4+\e+\extra,0+\e);

\fill[fill=red,opacity=1] (0-\e-\extra,2-\e) rectangle (4+\e+\extra,2+\e);

\fill[fill=red,opacity=1] (0-\e-\extra,4-\e) rectangle (4+\e+\extra,4+\e);

\def\small{0.4}
\def\large{1.6}
\def\f{0.2}

\draw[thick] (\small,\small)--(\large,\small)--(\large,\large)--(\small,\large)--cycle; 
\draw[thick] (\small+\f,\small+\f)--(\large-\f,\small+\f)--(\large-\f,\large-\f)--(\small+\f,\large-\f)--cycle; 
\draw[thick] (\small,\small)--(\small+\f,\small+\f);
\draw[thick] (\large,\small)--(\large-\f,\small+\f);
\draw[thick] (\large,\large)--(\large-\f,\large-\f);
\draw[thick] (\small,\large)--(\small+\f,\large-\f); 

\def\xshift{2}

\draw[thick] (\xshift+\small,\small)--(\xshift+\large,\small)--(\xshift+\large,\large)--(\xshift+\small,\large)--cycle; 
\draw[thick] (\xshift+\small+\f,\small+\f)--(\xshift+\large-\f,\small+\f)--(\xshift+\large-\f,\large-\f)--(\xshift+\small+\f,\large-\f)--cycle; 
\draw[thick] (\xshift+\small,\small)--(\xshift+\small+\f,\small+\f);
\draw[thick] (\xshift+\large,\small)--(\xshift+\large-\f,\small+\f);
\draw[thick] (\xshift+\large,\large)--(\xshift+\large-\f,\large-\f);
\draw[thick] (\xshift+\small,\large)--(\xshift+\small+\f,\large-\f); 

\def\yshift{2}

\draw[thick] (\xshift+\small,\yshift+\small)--(\xshift+\large,\yshift+\small)--(\xshift+\large,\yshift+\large)--(\xshift+\small,\yshift+\large)--cycle; 
\draw[thick] (\xshift+\small+\f,\yshift+\small+\f)--(\xshift+\large-\f,\yshift+\small+\f)--(\xshift+\large-\f,\yshift+\large-\f)--(\xshift+\small+\f,\yshift+\large-\f)--cycle; 
\draw[thick] (\xshift+\small,\yshift+\small)--(\xshift+\small+\f,\yshift+\small+\f);
\draw[thick] (\xshift+\large,\yshift+\small)--(\xshift+\large-\f,\yshift+\small+\f);
\draw[thick] (\xshift+\large,\yshift+\large)--(\xshift+\large-\f,\yshift+\large-\f);
\draw[thick] (\xshift+\small,\yshift+\large)--(\xshift+\small+\f,\yshift+\large-\f); 

\def\xshift{0}

\draw[thick] (\xshift+\small,\yshift+\small)--(\xshift+\large,\yshift+\small)--(\xshift+\large,\yshift+\large)--(\xshift+\small,\yshift+\large)--cycle; 
\draw[thick] (\xshift+\small+\f,\yshift+\small+\f)--(\xshift+\large-\f,\yshift+\small+\f)--(\xshift+\large-\f,\yshift+\large-\f)--(\xshift+\small+\f,\yshift+\large-\f)--cycle; 
\draw[thick] (\xshift+\small,\yshift+\small)--(\xshift+\small+\f,\yshift+\small+\f);
\draw[thick] (\xshift+\large,\yshift+\small)--(\xshift+\large-\f,\yshift+\small+\f);
\draw[thick] (\xshift+\large,\yshift+\large)--(\xshift+\large-\f,\yshift+\large-\f);
\draw[thick] (\xshift+\small,\yshift+\large)--(\xshift+\small+\f,\yshift+\large-\f); 

\foreach \x in {0,...,4}
\foreach \y in {0,...,4}
{
	\node[cstyle] at (\x,\y) {};
}

\end{tikzpicture}}
\end{center}
\caption{Left: the division of the grid into four sets. Right: one of the four
$g_i$ functions.}
\label{fig:grid}
\end{figure}

The bit extraction process requires us to implement an inherently discontinuous
function, and since linear circuits can compute only continuous functions, we
are forced to rely on bit extractors that can fail for a small subset of the
inputs. This is a well-known problem, and prior work has addressed it through
the use of the  ``averaging trick'', in which one extracts bits for $x$ and
also a large number of points that are close to~$x$. By arranging the process
such that only a small number of the bit extractions can fail, then the results
over all of the points can be averaged, giving us a function $\widetilde{f}$ with $| \widetilde{f}(x)
- f(x) | \le \eps$ for some small $\eps$.

For prior work, which was usually interested only in ensuring that $\widetilde{f}(x)$ was
far from zero whenever $f(x)$ was also far from zero, this approach was good
enough. However for our purposes, any deviation in the desired output, no matter how
small, may fatally undermine the construction by introducing new gradients, which could create new solutions.

To overcome this, we apply averaging in a new way that ensures that any errors
arising from bit extraction failures do not make their way into the output.
Specifically, we divide the points in the grid into four sets $S_1, S_2, S_3,
S_4 \subseteq G$ where each set contains points from a grid of double the width
of the original. The four sets are shown in four different colours on the left
side of \cref{fig:grid}.
For each set $S_i$ we build a function $g_i(x) = \max_{y \in S_i} m(x, y)$
which takes the maximum only over the mesas whose centers are in $S_i$.  We then set $f(x) =
\max(g_1(x), g_2(x), g_3(x), g_4(x))$.

The function $g_i$ is shown on the right of \cref{fig:grid}.
The red lines in the figure show the locations at which a bit extraction might
fail. Importantly, these regions occur only in places at which $m(x, y) \le 0$ for all mesa
centers $y \in S_i$. We ensure that at most two bit extractions can fail, and
that when they do fail they result in an output that is at most 1. So by
averaging over 12 points near $x$ we can ensure that $g_i(x) \le 1/6$ whenever
a bit extraction failure occurs. 

However, we also ensure that $f(y) \ge 1/3$ for all points $y$. 
Therefore, the function $g_j$ that is responsible for the mesa that defines $f(x)$ will satisfy $g_j(x) \ge 1/3$. Since $f$ is defined to be the maximum of the $g$ functions, this means that any errors arising from failed bit decodings in $g_i$ will be masked completely by another function $g_j$ that has correctly decoded its input. 
Ultimately, this means that 
any spurious gradient arising from a failed
bit extraction is well below the value of $f(x)$, and so these gradients can never make their way into the output.

\paragraph{\bf Challenge 2: Multiplying two variables.}

We are given Boolean circuits that define the mesas that we must output. This
means that a Boolean circuit will tell us to output an affine function at $x =
(x_1, x_2)$ with gradient $g = (g_1, g_2)$ and additive value $a$, in which
case we will need to output
\begin{equation*}
x_1 \cdot g_1 + x_2 \cdot g_2 + a.
\end{equation*}
This is problematic however, because
$x$ and $g$ are both variables, and a linear circuit does not allow us to multiply two
variables together.

We circumvent this by showing that it is possible to compute $x \cdot y$ when
$x$ is a continuous variable and $y$ is a variable encoded in binary.
So we represent gradients in binary in our circuit, and this allows us to
precisely control the gradients that we output.

\paragraph{\bf Challenge 3: Perturbations.}

It is not enough to create a linear circuit that implements $f$, because our
linear circuit must also be robust to perturbations. These perturbations will
slightly alter the values that are outputted by min, max, and truncation
operations, and we need to ensure that they do not alter the gradients of the
mesas that appear in the output of the function. 

We are able to show that evaluating a Boolean circuit and decoding a binary
value can be carried out exactly with no errors even in the presence of
perturbations, while introducing perturbations in the bit extraction process
only slightly increases the region in which the bit extraction will fail. 
Then we show that each mesa can be computed with gradients that are exactly
correct, but where each of the five pieces might be additively perturbed by a
small amount. Finally, we show that the step of splitting the points into the $g_i$
functions and then maximizing over them to produce $f$ only adds an additional
small perturbation, while not altering the gradients of any of the mesas.

\section{Preliminaries}

All numbers appearing as inputs in our problem are assumed to be rational. A rational number $x$ is represented by its numerator and denominator (in binary) of the irreducible fraction for $x$. We let $\size(x)$ denote the number of bits needed to represent $x$ in this way. We also extend this notation in the natural way to the case where $x$ is a vector with rational entries.

\begin{definition}
The \qKKT problem is defined as follows. We are given a degree-2 polynomial $p$ over $n$ variables, and the goal is to compute a KKT point of the following optimization problem:
\begin{equation}\label{eq:qKKT}
\begin{split}
\min \quad &p(x) \\
\text{ s.t.} \quad  
& 0 \leq x_i \leq 1 \quad \forall i \in [n]
\end{split}
\end{equation}
A point $x \in [0,1]^n$ is a KKT point of \eqref{eq:qKKT} if, for all $i \in [n]$,
\begin{itemize}
    \item if $x_i > 0$, then $\frac{\partial p}{\partial x_i} (x) \leq 0$, and
    \item if $x_i < 1$, then $\frac{\partial p}{\partial x_i} (x) \geq 0$.
\end{itemize}
\end{definition}

The \qKKT problem lies in the class \cls/, even for more general domains,\footnote{We omit the definition of a KKT point for more general domains; see, e.g., \cite{FGHS22}.} namely any non-empty compact domain given by linear inequalities~\cite{FGHS22}. This is because the problem can be solved (inefficiently) by gradient descent. The problem is guaranteed to always admit at least one rational solution with polynomially bounded bit complexity; see \cref{app:approx-to-exact} for a proof sketch. Our main result is the following theorem.

\begin{theorem}\label{thm:main-formal}
The \qKKT problem is \cls/-complete.
\end{theorem}

The problem remains \cls/-complete even if we only ask for an $\eps$-KKT point, where $\eps > 0$ is allowed to be exponentially small (i.e., is given in binary). Indeed, by standard arguments, finding an exact solution reduces to finding an approximate solution; see \cref{app:approx-to-exact}. For the definition of $\eps$-KKT points, we just replace ``$\leq 0$'' and ``$\geq 0$'' by ``$\leq \eps$'' and ``$\geq -\eps$'' (respectively) in the definition of \qKKT.

We prove \cref{thm:main-formal} by reducing from the \linearKKT problem, which we introduce below. In \cref{sec:lacqp} we reduce from \linearKKT to \qKKT, and in \cref{sec:linear-KKT-hardness,sec:circuit-implementation} we show that \linearKKT is \cls/-hard.

\subsection{KKT points of linear arithmetic circuits}

A linear arithmetic circuit is a circuit that consists of gates implementing piecewise linear operations. As a result, the function represented by a linear circuit is a piecewise linear function. Such circuits can be evaluated in polynomial time~\cite{FGHS22}.

In this paper, we consider linear arithmetic circuits that consist of a single\footnote{Of course, various other types of gates can be simulated using truncated linear gates and we will use this later in the paper.} type of gate: truncated linear gates. A truncated linear gate is defined by rational parameters $a,b,c \in \rationals$. The gate takes as input two variables $x_i, x_j$ of the circuit and outputs $x_k := \trunc(ax_i + bx_j + c)$, where $\trunc: \reals \to [0,1]$ denotes truncation (i.e., projection) to the $[0,1]$ interval.

\paragraph*{Generalized gradients.}
Let $\circuit$ be a linear arithmetic circuit with $n$ inputs and one output. We let $f: \mathbb{R}^n \to \mathbb{R}$ denote the function computed by the circuit $\circuit$. This is a piecewise linear function that is almost everywhere differentiable. The generalized gradient of $f$ at point $y$ can be defined as
\begin{equation*}
\begin{split}
\partial f(y) := \conv\Big\{\lim_{t \to \infty} \nabla f(y_t) : \text{$(y_t)_t$ converging to $y$ such that} &\text{ $f$ is differentiable at $y_t$}\\
&\text{and $\nabla f(y_t)$ also converges}\Big\}.
\end{split}
\end{equation*}

For our purposes we have to introduce a new more general notion of generalized gradient \emph{of a circuit}. Let $\circuit$ be a linear arithmetic circuit consisting of $m$ truncated linear gates. For any $\pi = (\pi_i)_{i \in [m]} \in \mathbb{R}^{m}$, we let $\circuit^\pi$ denote the circuit $\circuit$ perturbed by $\pi$, namely, for each $i \in [m]$ the $i$th gate $x_i := \trunc (a x_j + b x_k + c)$ is replaced by $x_i := \trunc (a x_j + b x_k + c + \pi_i)$. We let $f^{\pi}: \mathbb{R}^n \to \mathbb{R}$ denote the function represented by the perturbed circuit $\circuit^\pi$.

For any $\delta > 0$ and any such circuit $\circuit$, the $\delta$-generalized circuit gradient of $\circuit$ at point $y \in \mathbb{R}^n$ is defined as
$$\widetilde{\partial}_\delta \circuit(y) := \conv\big\{\nabla f^\pi(y): \pi \in [-\delta,\delta]^m \text{ such that $f^\pi$ is differentiable at $y$}\big\}.$$
It can be shown that $\partial f(y) \subseteq \widetilde{\partial}_\delta \circuit(y)$ for all $\delta > 0$. Although it is tempting to think that $\widetilde{\partial}_\delta \circuit(y) \to \partial f(y)$ as $\delta \to 0$, this is not the case. Indeed, \cref{ex:backprop-counter} together with the Backpropagation Lemma (\cref{lem:backpropagation}) provide a counter-example.

We can now define the intermediate computational problem which will act as a bridge between existing \cls/-hard problems and \qKKT.

\begin{definition}
The \linearKKT problem is defined as follows. We are given $\eps, \delta > 0$ and a linear arithmetic circuit $\circuit$ with two inputs and one output, and consisting only of truncated linear gates. The goal is to find a point $y \in [0,1]^2$ that satisfies the $\eps$-KKT conditions with respect to the $\delta$-generalized circuit gradient of $\circuit$, i.e., such that there exists $u \in \widetilde{\partial}_\delta \circuit(y)$ satisfying
\begin{itemize}
\item if $y_i > 0$, then $u_i \leq \eps$
\item if $y_i < 1$, then $u_i \geq - \eps$
\end{itemize}
for $i=1,2$.
\end{definition}

We note that it is not clear whether \linearKKT lies in \tfnp/, because it is not clear whether we can efficiently check if some given $y$ is a solution. Nevertheless, we establish in \cref{sec:linear-KKT-hardness} that this problem is \cls/-hard, which is all we require from this intermediate problem.

\section{Reduction from \linearKKT to \qKKT}\label{sec:lacqp}

The main result of this section is the following.

\begin{proposition}\label{prop:linear-to-quadratic}
There is a polynomial-time reduction from \linearKKT to \qKKT.
\end{proposition}

The remainder of this section proves this result. We begin with the detailed construction of the quadratic polynomial and a statement of the Backpropagation Lemma (\cref{lem:backpropagation}), and explain why it implies \cref{prop:linear-to-quadratic}. Then, we prove some simple properties of the construction, before the technical culmination of this section, namely the proof of the Backpropagation Lemma.

\subsection{Construction and Backpropagation Lemma}\label{sec:qkkt-construction}

Let $\circuit$ be a linear arithmetic circuit that has two inputs and one output, and that consists only of truncated linear gates.

Let $n$ denote the number of variables in the circuit $\circuit$. We use $x_i$ to denote the $i$th variable in the circuit, and assume that $x_1, \dots, x_n$ are ordered such that the gate computing $x_i$ uses inputs $x_{\ell(i)}, x_{r(i)}$ with $\ell(i) < i$ and $r(i) < i$. In particular, $x_1$ and $x_2$ are the input variables, and $x_n$ is the output variable. For every $i \in [n] \setminus \{1,2\}$, the $i$th gate of $\circuit$ is the gate computing $x_i$. It will be more convenient to write the $i$th gate's function $x_i = \trunc(a_i x_{\ell(i)} + b_i x_{r(i)} + c_i)$ as $x_i = \trunc(\sum_{j=1}^{i-1} a_{i j} x_j + c_i)$, where
$$a_{ij} = \left\{ \begin{tabular}{cl}
$a_i$ &if $j = \ell(i)$\\
$b_i$ &if $j = r(i)$\\
$0$ &otherwise
\end{tabular} \right.$$
Let also $K \geq 1$ be such that $K \geq \max_{i \in [n] \setminus \{1,2\}} (\sum_{j=1}^{i-1} |a_{ij}| + |c_i|)$.

We now construct a polynomial $p$ on $n + 2(n-2) = 3n - 4$ variables. In more detail, the polynomial will have the following variables:
\begin{itemize}
    \item For each $i \in [n]$, a variable $y_i$, corresponding to each variable $x_i$ of $\circuit$.
    \item For each $i \in [n] \setminus \{1,2\}$, two auxiliary variables $z_i^+$ and $z_i^-$ to help with the implementation of the $i$th gate, which computes $x_i$.
\end{itemize}
For each gate $i \in [n] \setminus \{1,2\}$ we construct a polynomial $q_i$ on variables $y=(y_1,\ldots,y_n)$, $z=(z_3^+,z_3^-,\ldots,z_n^+,z_n^-)$
$$q_i(y,z) := \left(y_i + K z_i^+ - K z_i^- - \sum_{j=1}^{i-1} a_{i j} y_j - c_i\right)^2 +2K^2 z_i^+ z_i^- + 2Kz_i^+(1-y_i) +2Kz_i^-y_i.$$
For a given $\delta \in (0,1)$, the final polynomial $p$ is then constructed as follows
$$p(y,z) := \delta^{n+1} y_n + \sum_{i=3}^n \delta^i q_i(y,z).$$

\paragraph*{\bf \qKKT instance.}
The instance of \qKKT we consider is thus
\begin{equation}\label{eq:KKT-construction}
\begin{split}
\min \quad &p(y,z) \\
\text{ s.t.} \quad  
& (y,z) \in [0,1]^{3n-4}
\end{split}
\end{equation}

We are now ready to state the main technical lemma of this section.

\begin{lemma}[Backpropagation Lemma]\label{lem:backpropagation}
Let $(y,z)$ be a KKT point of the constructed QP \eqref{eq:KKT-construction}, for some $\delta \in (0,1/16K^2)$. Then we have
$$\frac{1}{\delta^{n+1}} \cdot \left(\frac{\partial p}{\partial y_1}(y,z), \frac{\partial p}{\partial y_2}(y,z)\right) \in \widetilde{\partial}_{\delta'} \circuit(y_1,y_2)$$
where $\delta' = 8K^2\delta$.
\end{lemma}

Let us see how \cref{prop:linear-to-quadratic} follows from this lemma. Let $\eps'$, $\delta'$, and $\circuit$ be the inputs to a \linearKKT instance. We construct the polynomial $p$ described above with $\delta := \min\{\delta'/8K^2,1/32K^2\}$. Clearly, this can be done in polynomial time. Now, consider any KKT point $(y,z)$ of the resulting QP \eqref{eq:KKT-construction}. We claim that $(y_1,y_2)$ must be a solution to the original \linearKKT instance. Indeed, let
$$u := \frac{1}{\delta^{n+1}} \cdot \left(\frac{\partial p}{\partial y_1}(y,z), \frac{\partial p}{\partial y_2}(y,z)\right).$$
By the Backpropagation Lemma, we have that $u \in \widetilde{\partial}_{\delta'} \circuit(y_1,y_2)$. Furthermore, since $(y,z)$ is a KKT point of \eqref{eq:KKT-construction}, we in particular have for $i = 1,2$
\begin{itemize}
\item if $y_i > 0$, then $\frac{\partial p}{\partial y_i}(y,z) \leq 0$, and thus $u_i \leq 0$
\item if $y_i < 1$, then $\frac{\partial p}{\partial y_i}(y,z) \geq 0$, and thus $u_i \geq 0$.
\end{itemize}
In other words, $(y_1,y_2)$ satisfies the KKT conditions (and thus, in particular, the $\eps'$-KKT conditions) with respect to the $\delta'$-generalized circuit gradient of $\circuit$.

Before proceeding with the proof of the Backpropagation Lemma, we present an example showing that a stronger version of the lemma -- where we ask for the generalized gradient of $f$ at $(y_1,y_2)$ (or even of some point in the vicinity) to be zero -- fails.

\begin{example}\label{ex:backprop-counter}
Consider the circuit $\circuit$ that has one single input $x_1$ and computes $x_2 := \trunc(2x_1)$, $x_3 := \trunc(x_1-1/2)$, and outputs $x_4 := \trunc(x_2/2+x_3-x_1/2)$. It is easy to see that this circuit computes the function $f: [0,1] \to [0,1], x_1 \mapsto x_1/2$. Thus, the only KKT point of $f$ is at $x_1 = 0$. However, it can be checked that if we construct the polynomial $p$ as described above from $\circuit$, then, for any sufficiently small $\delta > 0$, the QP \eqref{eq:KKT-construction} will have a KKT point at $y_1 = 1/2 + \delta^2/4$ (and where we have $y_2 = 1$, $y_3 = 0$, and $y_4 = 1/4 -\delta^2/8 - \delta/2$). In particular, this means that the backpropagation computes gradient $0$ at that point, even though the actual gradient of $f$ is always $1/2$. As a result, no general backpropagation result can be proved for this kind of circuit without taking into account perturbed versions of the circuit.
\end{example}

\subsection{Properties of KKT points}

In this section we prove some simple properties that are satisfied by any KKT point of the \qKKT instance \eqref{eq:KKT-construction}. Recall that a point $(y,z) \in [0,1]^{3n-4}$ is a KKT point of \eqref{eq:KKT-construction} if, for all $i \in [n]$,
\begin{itemize}
    \item if $y_i > 0$, then $\frac{\partial p}{\partial y_i} (y,z) \leq 0$, and
    \item if $y_i < 1$, then $\frac{\partial p}{\partial y_i} (y,z) \geq 0$,
\end{itemize}
and similarly for the other variables $z_i^+$ and $z_i^-$ for all $i \in [n] \setminus \{1,2\}$.

\paragraph*{\bf Truncation.}
The following lemma \ref{lem:KKT-aux-variables} states that, at any KKT point, the auxiliary variables $z$ enforce truncation, in a certain sense.

\begin{lemma}\label{lem:KKT-aux-variables}
Let $(y,z)$ be a KKT point of \eqref{eq:KKT-construction}. Then for all $i \in [n] \setminus \{1,2\}$
$$\trunc\left(\sum_{j=1}^{i-1} a_{i j} y_j + c_i\right) = \sum_{j=1}^{i-1} a_{i j} y_j + c_i - Kz_i^+ + Kz_i^-.$$
\end{lemma}

\begin{proof}
We show the following stronger fact, namely that
\begin{equation}\label{eq:KKT-truncation+}
K z_i^+ = \max\left\{0, \left(\sum_{j=1}^{i-1} a_{i j} y_j + c_i\right) - \trunc\left(\sum_{j=1}^{i-1} a_{i j} y_j + c_i\right)\right\}
\end{equation}
and
\begin{equation}\label{eq:KKT-truncation-}
K z_i^- = \max\left\{0, \trunc\left(\sum_{j=1}^{i-1} a_{i j} y_j + c_i\right) - \left(\sum_{j=1}^{i-1} a_{i j} y_j + c_i\right)\right\}.
\end{equation}
In order to prove \eqref{eq:KKT-truncation+}, note that the variable $z_i^+$ only appears in $q_i$, and thus $\frac{\partial p}{\partial z_i^+} = \delta^i \frac{\partial q_i}{\partial z_i^+}$ and
\begin{equation*}
\begin{split}
\frac{\partial q_i}{\partial z_i^+}(y,z) &= 2 K \left(y_i + K z_i^+ - K z_i^- - \sum_{j=1}^{i-1} a_{i j} y_j - c_i\right) +2K^2 z_i^- + 2K(1-y_i)\\
&= 2 K \left(1 + K z_i^+ - \sum_{j=1}^{i-1} a_{i j} y_j - c_i\right).
\end{split}
\end{equation*}
We now consider two cases. If $\sum_{j=1}^{i-1} a_{i j} y_j + c_i \leq \trunc(\sum_{j=1}^{i-1} a_{i j} y_j + c_i)$, then it must be that $\sum_{j=1}^{i-1} a_{i j} y_j + c_i \leq 1$. As a result, $\frac{\partial p}{\partial z_i^+}(y,z) = \delta^i \frac{\partial q_i}{\partial z_i^+}(y,z) \geq \delta^i \cdot 2K^2 z_i^+$. By the KKT conditions it follows that $z_i^+ = 0$. Indeed, if $z_i^+ > 0$, then we would have $\frac{\partial p}{\partial z_i^+}(y,z) > 0$, which contradicts the KKT conditions.

If, on the other hand, $\sum_{j=1}^{i-1} a_{i j} y_j + c_i > \trunc(\sum_{j=1}^{i-1} a_{i j} y_j + c_i)$, then it must be that $\sum_{j=1}^{i-1} a_{i j} y_j + c_i > 1$. As a result, $\frac{\partial p}{\partial z_i^+}(y,z) = \delta^i \frac{\partial q_i}{\partial z_i^+}(y,z) < \delta^i \cdot 2K^2 z_i^+$. In particular, we cannot have $z_i^+ = 0$, since that would imply $\frac{\partial p}{\partial z_i^+}(y,z) < 0$, which is not allowed by the KKT conditions at $z_i^+ = 0$. We also cannot have $z_i^+ = 1$. Indeed, by the KKT conditions, that would imply that $\frac{\partial p}{\partial z_i^+}(y,z) \leq 0$, which translates to
$$\delta^i \cdot 2 K \left(1 + K - \sum_{j=1}^{i-1} a_{i j} y_j - c_i\right) \leq 0$$
which is impossible, since $K \geq 1$ was chosen such that $K \geq \sum_{j=1}^{i-1} |a_{ij}| + |c_i| \geq \sum_{j=1}^{i-1} a_{i j} y_j + c_i$. As a result, we must have $z_i^+ \in (0,1)$, which implies that the KKT condition is $\frac{\partial p}{\partial z_i^+}(y,z) = 0$. This yields
$$Kz_i^+ = \sum_{j=1}^{i-1} a_{i j} y_j + c_i - 1 = \sum_{j=1}^{i-1} a_{i j} y_j + c_i - \trunc\left(\sum_{j=1}^{i-1} a_{i j} y_j + c_i\right)$$
as desired. We have thus shown that \eqref{eq:KKT-truncation+} always holds at a KKT point.

In order to prove \eqref{eq:KKT-truncation-}, we again note that $\frac{\partial p}{\partial z_i^-} = \delta^i \frac{\partial q_i}{\partial z_i^-}$ and
\begin{equation*}
\begin{split}
\frac{\partial q_i}{\partial z_i^-}(y,z) &= -2 K \left(y_i + K z_i^+ - K z_i^- - \sum_{j=1}^{i-1} a_{i j} y_j - c_i\right) +2K^2 z_i^+ + 2Ky_i\\
&= 2 K \left(K z_i^- + \sum_{j=1}^{i-1} a_{i j} y_j + c_i\right).
\end{split}
\end{equation*}
and then perform a similar case analysis.
\end{proof}

\paragraph*{\bf Approximate evaluation.} The next lemma \ref{lem:KKT-main-variables} states that the gates of the circuit are correctly simulated at a KKT point of \eqref{eq:KKT-construction}, up to some small additive error depending on the parameter $\delta$. The lemma also gives a precise expression for the value of each variable $y_i$ at a KKT point, which will be useful for the next section. In order to state this precise expression, we first have to introduce some additional notation. We define, for any $i \in [n] \setminus \{1\}$,
$$p_i(y,z) := \delta^{n+1} y_n + \sum_{\ell=i+1}^n \delta^\ell q_\ell(y,z).$$
In particular, $p_2 = p$, and $p_n(y,z) = \delta^{n+1} y_n$. We are now ready to state the lemma.

\begin{lemma}\label{lem:KKT-main-variables}
Let $(y,z)$ be a KKT point of \eqref{eq:KKT-construction}. Then for any $i \in [n] \setminus \{1,2\}$
$$y_i = \trunc\left(\sum_{j=1}^{i-1} a_{i j} y_j + c_i - \frac{1}{2 \delta^i} \cdot \frac{\partial p_i}{\partial y_i}(y,z)\right) = \trunc\left(\sum_{j=1}^{i-1} a_{i j} y_j + c_i\right) \pm (2K\delta)^{n+1-i}.$$
\end{lemma}

\begin{proof}
Let us first prove the first equality. Note that
\begin{equation*}
\begin{split}
\frac{\partial p}{\partial y_i}(y,z) &= \delta^i \frac{\partial q_i}{\partial y_i}(y,z) + \frac{\partial p_i}{\partial y_i}(y,z)\\
&= \delta^i \left( 2\left(y_i + K z_i^+ - K z_i^- - \sum_{j=1}^{i-1} a_{i j} y_j - c_i\right) - 2Kz_i^+ +2Kz_i^- \right) + \frac{\partial p_i}{\partial y_i}(y,z)\\
&= 2\delta^i \left(y_i - \sum_{j=1}^{i-1} a_{i j} y_j - c_i\right) + \frac{\partial p_i}{\partial y_i}(y,z).
\end{split}
\end{equation*}
Hence if $y_i > \sum_{j=1}^{i-1} a_{i j} y_j - c_i - \frac{1}{2 \delta^i} \frac{\partial p_i}{\partial y_i}(y,z)$, then $\frac{\partial p}{\partial y_i}(y,z) > 0$, and by the KKT conditions we must have $y_i = 0$. If, on the other hand, $y_i < \sum_{j=1}^{i-1} a_{i j} y_j - c_i - \frac{1}{2 \delta^i} \frac{\partial p_i}{\partial y_i}(y,z)$, then $\frac{\partial p}{\partial y_i}(y,z) < 0$, and by the KKT conditions we must have $y_i = 1$. Thus, in all cases we have $y_i = \trunc(\sum_{j=1}^{i-1} a_{i j} y_j - c_i - \frac{1}{2 \delta^i} \frac{\partial p_i}{\partial y_i}(y,z))$.

In order to prove the second equality, we show that $\frac{1}{2\delta^i}|\frac{\partial p_i}{\partial y_i}(y,z)| \leq (2K\delta)^{n+1-i}$ by induction. For $i = n$, we have $\frac{\partial p_n}{\partial y_n}(y,z) = \delta^{n+1}$, and thus $\frac{1}{2\delta^n}|\frac{\partial p_n}{\partial y_n}(y,z)| = \delta/2 \leq 2K\delta$. Now, assume that the statement holds for $i+1, \dots, n$. We show that it also holds for $i$. We can bound
\begin{equation*}
\begin{split}
\frac{1}{2 \delta^i} \left|\frac{\partial p_i}{\partial y_i}(y,z)\right| \leq \frac{1}{2 \delta^i} \sum_{\ell=i+1}^n \delta^\ell \left|\frac{\partial q_\ell}{\partial y_i}(y,z)\right| &\leq \frac{1}{2 \delta^i} \sum_{\ell=i+1}^n \delta^\ell \cdot 2K \cdot (2K\delta)^{n+1-\ell}\\
&= \delta^{n+1-i} \cdot K \cdot \sum_{\ell=i+1}^n (2K\delta)^{n+1-\ell}\\
&\leq \delta^{n+1-i} \cdot K \cdot 2 \cdot (2K)^{n-i} = (2K\delta)^{n+1-i}
\end{split}
\end{equation*}
where we bound $|\frac{\partial q_\ell}{\partial y_i}(y,z)|$ by
\begin{equation*}
\begin{split}
\left|\frac{\partial q_\ell}{\partial y_i}(y,z)\right| &= \left|-2 a_{\ell i}\left(y_\ell + K z_\ell^+ - K z_\ell^- - \sum_{j=1}^{\ell-1} a_{\ell j} y_j - c_\ell \right)\right|\\
&= \left|-2 a_{\ell i}\left(y_\ell - \trunc\left(\sum_{j=1}^{\ell-1} a_{\ell j} y_j - c_\ell\right) \right)\right|\\
&\leq 2K \left|y_\ell - \trunc\left(\sum_{j=1}^{\ell-1} a_{\ell j} y_j - c_\ell\right)\right|\\
&\leq 2K \cdot (2K\delta)^{n+1-\ell}
\end{split}
\end{equation*}
using \cref{lem:KKT-aux-variables}, $|a_{\ell i}| \leq K$, and the induction hypothesis for $\ell \geq i+1$.
\end{proof}

\subsection{Proof of the Backpropagation Lemma}

In this section we prove the Backpropagation Lemma. We begin by recalling some notation, as well as introducing some new notation. We let $f: \mathbb{R}^2 \to \mathbb{R}$ denote the function represented by the circuit $\circuit$. For any $\pi = (\pi_i)_{i \in [n] \setminus \{1,2\}} \in \mathbb{R}^{n-2}$, we let $\circuit^\pi$ denote the circuit $\circuit$ perturbed by $\pi$, namely, for each $i \in [n] \setminus \{1,2\}$ the $i$th gate $x_i := \trunc (\sum_{j=1}^{i-1} a_{i j} x_j + c_i)$ is replaced by $x_i := \trunc (\sum_{j=1}^{i-1} a_{i j} x_j + c_i + \pi_i)$. We let $f^{\pi}: \mathbb{R}^2 \to \mathbb{R}$ denote the function represented by the perturbed circuit $\circuit^\pi$. For any sign vector $s = (s_i)_{i \in [n] \setminus \{1,2\}} \in \{+1,-1\}^{n-2}$, we let $s \cdot \pi \in \mathbb{R}^{n-2}$ denote the coordinate-wise product of vector $s$ with vector $\pi$, i.e., $[s \cdot \pi]_i = s_i \pi_i$ for all $i \in [n] \setminus \{1,2\}$. Below we also use $\lambda_{-i}$ to denote $1 - \lambda_i$.

The Backpropagation Lemma is a consequence of the following technical lemma.

\begin{lemma}\label{lem:KKT-technical-lemma}
Let $(y,z)$ be a KKT point of QP \eqref{eq:KKT-construction}, for some $\delta \in (0,1/16K^2)$. Then there exists a perturbation vector $\pi = (\pi_i)_{i \in [n] \setminus \{1,2\}} \in \mathbb{R}^{n-2}$ satisfying
\begin{itemize}
	\item $|\pi_i| \leq 8K^2 \delta$ for all $i \in [n] \setminus \{1,2\}$
	\item for all $s \in \{+1,-1\}^{n-2}$, $f^{s \cdot \pi}$ is differentiable in a small neighborhood around $(x_1,x_2) = (y_1,y_2)$
\end{itemize}
In addition, there exists $\lambda = (\lambda_i)_{i \in [n] \setminus \{1,2\}} \in [0,1]^{n-2}$ such that for $k=1,2$
$$\frac{\partial p}{\partial y_k}(y,z) = \delta^{n+1} \sum_{s \in \{+1,-1\}^{n-2}} \left( \prod_{j=3}^{n} \lambda_{s_j \cdot j} \right) \frac{\partial f^{s \cdot \pi}}{\partial x_k} (y_1,y_2).$$
\end{lemma}

Before moving to the proof of the technical lemma, let us see why it implies the Backpropagation Lemma. From the two bullets we obtain by definition of the generalized circuit gradient that
$$\nabla f^{s \cdot \pi}(y_1,y_2) = \left(\frac{\partial f^{s \cdot \pi}}{\partial x_1} (y_1,y_2),\frac{\partial f^{s \cdot \pi}}{\partial x_2} (y_1,y_2)\right) \in \widetilde{\partial}_{\delta'} \circuit(y_1,y_2)$$
for all $s \in \{+1,-1\}^{n-2}$, and where we let $\delta' := 8K^2\delta$. As a result of the last part of the technical lemma we can write
\begin{equation*}
\begin{split}
\frac{1}{\delta^{n+1}} \cdot \left(\frac{\partial p}{\partial y_1}(y,z), \frac{\partial p}{\partial y_2}(y,z)\right) = \sum_{s \in \{+1,-1\}^{n-2}} \left( \prod_{j=3}^{n} \lambda_{s_j \cdot j} \right) \nabla f^{s \cdot \pi}(y_1,y_2) \in \widetilde{\partial}_{\delta'} \circuit(y_1,y_2)
\end{split}
\end{equation*}
since this is a convex combination of elements in $\widetilde{\partial}_{\delta'} \circuit(y_1,y_2)$, and this set is convex by definition. This is exactly the statement of the Backpropagation Lemma.

\subsubsection{Proof of the Technical Lemma}

Let $(y,z)$ be a KKT point of \eqref{eq:KKT-construction}. For $i \in [n] \setminus \{1\}$, let $\eps_i := (2K\delta)^{n-i}$.

\paragraph*{\bf Construction of $\pi$.}
Let $i \in [n] \setminus \{1,2\}$. We construct $\pi_i$ as follows
\begin{itemize}
\item If $|\sum_{j=1}^{i-1} a_{i j} y_j + c_i - 1| \leq 2K \eps_{i-1}$, then we set $\pi_i := - 4K \eps_{i-1}$.
\item If $|\sum_{j=1}^{i-1} a_{i j} y_j + c_i - 0| \leq 2K \eps_{i-1}$, then we set $\pi_i := 4K \eps_{i-1}$.
\item In all other cases we set $\pi_i := 0$.
\end{itemize}
Note that the two first cases cannot both occur, since $2K \eps_{i-1} \leq 2K \eps_{n-1} = 4K^2 \delta < 1/4$, since $\delta < 1/16K^2$. Furthermore, since $2K \eps_{i-1} < 1/4$, we also have that $\sum_{j=1}^{i-1} a_{i j} y_j + c_i + \pi_i \notin (-2K\eps_{i-1}, 2K\eps_{i-1}) \cup (1-2K\eps_{i-1},1+2K\eps_{i-1})$, and similarly $\sum_{j=1}^{i-1} a_{i j} y_j + c_i - \pi_i \notin (-2K\eps_{i-1}, 2K\eps_{i-1}) \cup (1-2K\eps_{i-1},1+2K\eps_{i-1})$.

\paragraph*{\bf Construction of $\lambda$.}
Let $i \in [n] \setminus \{1,2\}$. We construct $\lambda_i \in [0,1]$ as follows
\begin{itemize}
	\item If $\sum_{j=1}^{i-1} a_{i j} y_j + c_i < -2K \eps_{i-1}$, then set $\lambda_i := 0$.
	\item If $\sum_{j=1}^{i-1} a_{i j} y_j + c_i > 1 + 2K \eps_{i-1}$, then set $\lambda_i := 0$.
	\item If $\sum_{j=1}^{i-1} a_{i j} y_j + c_i \in (2K \eps_{i-1}, 1 - 2K \eps_{i-1})$, then set $\lambda_i := 1$.
	\item Otherwise, pick $\lambda_i \in [0,1]$ as a solution of the equation
	\begin{equation}\label{eq:def-lambda}
	\frac{1}{2 \delta^i} \cdot \frac{\partial p_i}{\partial y_i}(y,z) \cdot \lambda_i = \trunc\left(\sum_{j=1}^{i-1} a_{i j} y_j + c_i\right) - \trunc\left(\sum_{j=1}^{i-1} a_{i j} y_j + c_i - \frac{1}{2 \delta^i} \cdot \frac{\partial p_i}{\partial y_i}(y,z)\right).
	\end{equation}
	Note that such $\lambda_i \in [0,1]$ always exists.
\end{itemize}
In fact, it is not hard to see that $\lambda_i$ satisfies \eqref{eq:def-lambda} in all four cases. This follows from the fact that by the proof of \cref{lem:KKT-main-variables}
$$\left| \frac{1}{2 \delta^i} \cdot \frac{\partial p_i}{\partial y_i}(y,z) \right| \leq (2K\delta)^{n+1-i} = \eps_{i-1} \leq K\eps_{i-1}.$$
Furthermore, note that by \cref{lem:KKT-main-variables} we can rewrite the equation satisfied by $\lambda_i$ as
\begin{equation}\label{eq:def-lambda-y}
\frac{1}{2 \delta^i} \cdot \frac{\partial p_i}{\partial y_i}(y,z) \cdot \lambda_i = \trunc\left(\sum_{j=1}^{i-1} a_{i j} y_j + c_i\right) - y_i.
\end{equation}

Before stating and proving an important claim satisfied by $\pi$ and $\lambda$, we introduce some additional notation. For $i \in [n] \setminus \{1,2\}$ and any $s_i \in \{+1,-1\}$, we define the function $\phi_i^{s_i \cdot \pi_i}: \mathbb{R}^{i-1} \to \mathbb{R}$ by $\phi_i^{s_i \cdot \pi_i}(x_1, \dots, x_{i-1}) = \trunc(\sum_{j=1}^{i-1} a_{i j} x_j + c_i + s_i \cdot \pi_i)$. Recall that we use $\lambda_{-i}$ to denote $1 - \lambda_i$.

\begin{claim}\label{clm:pi-properties}
For $i \in [n] \setminus \{1,2\}$ and for any $s_i \in \{+1,-1\}$, the function $\phi_i^{s_i \cdot \pi_i}$ is differentiable (with respect to its inputs $x_1, \dots, x_{i-1}$) over $\prod_{j \in [i-1]} [y_j - \eps_{i-1}, y_j + \eps_{i-1}]$ and we have
$$\frac{\partial \phi_i^{s_i \cdot \pi_i}}{\partial x_k}(v_1, \dots, v_{i-1}) = \frac{\partial \phi_i^{s_i \cdot \pi_i}}{\partial x_k}(y_1, \dots, y_{i-1})$$
for all $k \in [i-1]$ and all $v \in \prod_{j \in [i-1]} [y_j - \eps_{i-1}, y_j + \eps_{i-1}]$. Furthermore, we also have
$$\lambda_{i} \cdot \frac{\partial \phi_i^{\pi_i}}{\partial x_k}(y_1, \dots, y_{i-1}) + \lambda_{-i} \cdot \frac{\partial \phi_i^{-\pi_i}}{\partial x_k}(y_1, \dots, y_{i-1}) = \lambda_i \cdot a_{ik}$$
for all $k \in [i-1]$.
\end{claim}

\begin{proof}
In order to show that $\phi_i^{s_i \cdot \pi_i}$ is differentiable over $\prod_{j \in [i-1]} [y_j - \eps_{i-1}, y_j + \eps_{i-1}]$, we will show that $\sum_{j=1}^{i-1} a_{i j} v_j + c_i + s_i \cdot \pi_i$ is far from both $0$ and $1$ for all $v \in \prod_{j \in [i-1]} [y_j - \eps_{i-1}, y_j + \eps_{i-1}]$. Indeed, $\pi_i$ has been constructed in order to ensure this. Namely, we have
$$\left| \sum_{j=1}^{i-1} a_{i j} v_j + c_i + s_i \cdot \pi_i - 1 \right| \geq \left| \sum_{j=1}^{i-1} a_{i j} y_j + c_i + s_i \cdot \pi_i - 1 \right| - K \eps_{i-1} \geq K \eps_{i-1}$$
where we used the fact that $\sum_{j=1}^{i-1} a_{i j} y_j + c_i + s_i \cdot \pi_i \notin (1-2K\eps_{i-1},1+2K\eps_{i-1})$. A very similar argument shows that $\sum_{j=1}^{i-1} a_{i j} v_j + c_i + s_i \cdot \pi_i$ is also far from $0$. As a result, the first part of the claim follows, since $\phi_i^{s_i \cdot \pi_i}$ restricted to $\prod_{j \in [i-1]} [y_j - \eps_{i-1}, y_j + \eps_{i-1}]$ is a linear affine function.

We prove the second part of the claim by a case analysis.
\begin{itemize}
	\item If $\pi_i=0$, then $\phi_i^{\pi_i} = \phi_i^{-\pi_i} = \phi_i^{0}$. Thus, since $\lambda_i + \lambda_{-i} = 1$, we can write
	$$\lambda_{i} \cdot \frac{\partial \phi_i^{\pi_i}}{\partial x_k}(y_1, \dots, y_{i-1}) + \lambda_{-i} \cdot \frac{\partial \phi_i^{-\pi_i}}{\partial x_k}(y_1, \dots, y_{i-1}) = \frac{\partial \phi_i^{0}}{\partial x_k}(y_1, \dots, y_{i-1}).$$
	Now we have three subcases:
	\begin{itemize}
		\item If $\sum_{j=1}^{i-1} a_{i j} y_j + c_i < -2K\eps_{i-1}$, then $\frac{\partial \phi_i^{0}}{\partial x_k}(y_1, \dots, y_{i-1}) = 0$ and $\lambda_i = 0$. So the desired equation holds.
		\item If $\sum_{j=1}^{i-1} a_{i j} y_j + c_i > 1 + 2K\eps_{i-1}$, then $\frac{\partial \phi_i^{0}}{\partial x_k}(y_1, \dots, y_{i-1}) = 0$ and $\lambda_i = 0$. So the equation holds.
		\item If $\sum_{j=1}^{i-1} a_{i j} y_j + c_i \in (2K \eps_{i-1}, 1 - 2K \eps_{i-1})$, then $\frac{\partial \phi_i^{0}}{\partial x_k}(y_1, \dots, y_{i-1}) = a_{ik}$ and $\lambda_i = 1$. So the equation holds.
	\end{itemize}
	\item If $\pi_i = -4K\eps_{i-1}$, then, by construction of $\pi_i$, we must have $\sum_{j=1}^{i-1} a_{i j} y_j + c_i \in [1-2K\eps_{i-1}, 1+2K\eps_{i-1}]$. But this implies that $\frac{\partial \phi_i^{\pi_i}}{\partial x_k}(y_1, \dots, y_{i-1}) = a_{ik}$ and $\frac{\partial \phi_i^{-\pi_i}}{\partial x_k}(y_1, \dots, y_{i-1}) = 0$. As a result, the desired equation again holds.
	\item If $\pi_i = 4K\eps_{i-1}$, then, by construction of $\pi_i$, we must have $\sum_{j=1}^{i-1} a_{i j} y_j + c_i \in [-2K\eps_{i-1}, 2K\eps_{i-1}]$. But this implies that $\frac{\partial \phi_i^{\pi_i}}{\partial x_k}(y_1, \dots, y_{i-1}) = a_{ik}$ and $\frac{\partial \phi_i^{-\pi_i}}{\partial x_k}(y_1, \dots, y_{i-1}) = 0$. So the desired equation holds.
\end{itemize}
\end{proof}

We are now ready to prove the technical lemma. We will prove a slightly stronger statement by induction. For this, we need some additional notation. For $i \in [n] \setminus \{1\}$, we let $\circuit_i^\pi$ denote the circuit $\circuit^\pi$ but where we have only kept the gates $i+1, \dots, n$. We think of $\circuit_i^\pi$ as having input variables $x_1, x_2, \dots, x_i$ (even though some of those variables might be unused and thus not have any impact on the output of the circuit). We let $f_i^\pi: \mathbb{R}^i \to \mathbb{R}$ denote the function represented by $\circuit_i^\pi$. Note that $f_2^\pi = f^\pi$ and $f_n^\pi(x_1, \dots, x_n) = x_n$.

\begin{claim}\label{clm:backpropagation-induction}
For any $i \in [n] \setminus \{1\}$ we have
\begin{enumerate}
	\item For any $s \in \{+1,-1\}^{n-2}$, the function $f_i^{s \cdot \pi}$ is differentiable (with respect to its inputs $x_1, x_2, \dots, x_i$) over $\prod_{j \in [i]} [y_j - \eps_i, y_j + \eps_i]$ and we have
	$$\frac{\partial f_i^{s \cdot \pi}}{\partial x_k}(v_1, \dots, v_i) = \frac{\partial f_i^{s \cdot \pi}}{\partial x_k}(y_1, \dots, y_i)$$
	for all $k \in [i]$ and all $v \in \prod_{j \in [i]} [y_j - \eps_i, y_j + \eps_i]$.
	\item For any $k \in [i]$ we have
	$$\frac{\partial p_i}{\partial y_k}(y,z) = \delta^{n+1} \sum_{s: s_j = 1 \; \forall j \leq i} \left(\prod_{j=i+1}^{n} \lambda_{s_j \cdot j}\right) \frac{\partial f_i^{s \cdot \pi}}{\partial x_k} (y_1, \dots, y_i).$$
\end{enumerate}
\end{claim}

The technical lemma (\cref{lem:KKT-technical-lemma}) simply follows from the claim by noting that for $i=2$ we have $f_2^{s \cdot \pi} = f^{s \cdot \pi}$ and $p_2 = p$. Furthermore, note that for all $i \in [n] \setminus \{1,2\}$ we have $|\pi_i| \leq 4K\eps_{i-1} \leq 4K \eps_{n-1} = 8K^2 \delta$, as desired. It remains to prove the claim.

\begin{proof}
We prove the claim by induction on $i$.

\paragraph{\bf Base case: $i = n$.}
Recall that $f_n^{s \cdot \pi} (v_1, \dots, v_n) = v_n$ and $p_n(y,z) = \delta^{n+1} y_n$. Clearly, $f_n^{s \cdot \pi}$ is differentiable with $\frac{\partial f_n^{s \cdot \pi}}{\partial x_n}(v_1, \dots, v_n) = 1$ and $\frac{\partial f_n^{s \cdot \pi}}{\partial x_k}(v_1, \dots, v_n) = 0$ for all $k \in [n-1]$ and for all $v \in \mathbb{R}^n$. Furthermore, $\frac{\partial p_n}{\partial y_n}(y,z) = \delta^{n+1}$ and $\frac{\partial p_n}{\partial y_k}(y,z) = 0$ for all $k \in [n-1]$. Thus, for all $k \in [n]$ we have
$$\frac{\partial p_n}{\partial y_k}(y,z) = \delta^{n+1} \sum_{s: s_j = 1 \; \forall j \leq n} \left( \prod_{j=n+1}^{n} \lambda_{s_j \cdot j} \right)  \frac{\partial f_n^{s \cdot \pi}}{\partial x_k} (y_1, \dots, y_n)$$
where we note that $\prod_{j=n+1}^{n} \lambda_{s_j \cdot j} = 1$, because it is an empty product, and that the sum has a single summand.

\paragraph{\bf Induction step.}
Let $i \in [n] \setminus \{1,2\}$, and assume that the induction hypothesis holds for $i$. We show that it also holds for $i-1$. Fix any $k \in [i-1]$. For any $s \in \{+1,-1\}^{n-2}$, by the definition of $f_{i-1}^{s \cdot \pi}$, we have that for all $v \in \mathbb{R}^{i-1}$
$$f_{i-1}^{s \cdot \pi}(v_1, \dots, v_{i-1}) = f_i^{s \cdot \pi}(v_1, \dots, v_{i-1}, \phi_i^{s_i \cdot \pi_i}(v_1, \dots, v_{i-1}))$$
where we recall that $\phi_i^{s_i \cdot \pi_i}(v_1, \dots, v_{i-1}) = \trunc(\sum_{j=1}^{i-1} a_{i j} v_j + c_i + s_i \cdot \pi_i)$.
Now, for any $v \in \prod_{j=1}^{i-1} [y_j - \eps_{i-1}, y_j + \eps_{i-1}]$, we have that
\begin{equation*}
\begin{split}
\left|y_i - \phi_i^{s_i \cdot \pi_i}(v_1, \dots, v_{i-1})\right| &= \left| y_i - \trunc\left(\sum_{j=1}^{i-1} a_{i j} v_j + c_i + s_i \cdot \pi_i\right) \right|\\
&\leq \left| y_i - \trunc\left(\sum_{j=1}^{i-1} a_{i j} y_j + c_i \right) \right| + K \eps_{i-1} + |\pi_i|\\
&\leq (2K\delta)^{n+1-i} + K \eps_{i-1} + 4K \eps_{i-1}\\
&\leq 6K \eps_{i-1} = 6K (2K\delta) \eps_i \leq \eps_i
\end{split}
\end{equation*}
where we used \cref{lem:KKT-main-variables}, $(2K\delta)^{n+1-i} = \eps_{i-1}$, $\eps_{i-1} = (2K\delta) \eps_i$, and $\delta < 1/16K^2$. As a result, we have that $(v_1, \dots, v_{i-1}, \phi_i^{s_i \cdot \pi_i}(v_1, \dots, v_{i-1})) \in \prod_{j=1}^{i} [y_j - \eps_i, y_j + \eps_i]$. By the induction hypothesis and by \cref{clm:pi-properties}, it follows using the chain rule that $f_{i-1}^{s \cdot \pi}$ is differentiable on $\prod_{j=1}^{i-1} [y_j - \eps_{i-1}, y_j + \eps_{i-1}]$, and that its partial derivative with respect to $x_k$ satisfies
\begin{equation*}
\begin{split}
\frac{\partial f_{i-1}^{s \cdot \pi}}{\partial x_k}(v) &= \frac{\partial f_i^{s \cdot \pi}}{\partial x_k}(v, \phi_i^{s_i \cdot \pi_i}(v)) + \frac{\partial \phi_i^{s_i \cdot \pi_i}}{\partial x_k}(v) \cdot \frac{\partial f_i^{s \cdot \pi}}{\partial x_i}(v, \phi_i^{s_i \cdot \pi_i}(v))\\
&= \frac{\partial f_i^{s \cdot \pi}}{\partial x_k}(y_{\leq i}) + \frac{\partial \phi_i^{s_i \cdot \pi_i}}{\partial x_k}(y_{\leq i-1}) \cdot \frac{\partial f_i^{s \cdot \pi}}{\partial x_i}(y_{\leq i})
\end{split}
\end{equation*}
for all $v \in \prod_{j=1}^{i-1} [y_j - \eps_{i-1}, y_j + \eps_{i-1}]$. Here we used $y_{\leq i}$ to denote $(y_1, \dots, y_i)$.

For any sign vector $s = (s_i)_{i \in [n] \setminus \{1,2\}}$ with $s_i = +1$, let $s' = (s_3, \dots, s_{i-1}, -s_i, s_{i+1}, \dots, s_n)$. We can write
\begin{equation*}
\begin{split}
\frac{\partial f_{i-1}^{s' \cdot \pi}}{\partial x_k}(y_{\leq i-1}) &= \frac{\partial f_i^{s' \cdot \pi}}{\partial x_k}(y_{\leq i}) + \frac{\partial \phi_i^{s_i' \cdot \pi_i}}{\partial x_k}(y_{\leq i-1}) \cdot \frac{\partial f_i^{s' \cdot \pi}}{\partial x_i}(y_{\leq i})\\
&= \frac{\partial f_i^{s \cdot \pi}}{\partial x_k}(y_{\leq i}) + \frac{\partial \phi_i^{- s_i \cdot \pi_i}}{\partial x_k}(y_{\leq i-1}) \cdot \frac{\partial f_i^{s \cdot \pi}}{\partial x_i}(y_{\leq i})
\end{split}
\end{equation*}
where we used the fact that $f_i^{s' \cdot \pi} = f_i^{s \cdot \pi}$, since the $i$th gate is not included in the corresponding arithmetic circuits. As a result,
\begin{equation}\label{eq:clm-induction-proof}
\begin{split}
&\lambda_{i} \frac{\partial f_{i-1}^{s \cdot \pi}}{\partial x_k}(y_{\leq i-1}) + \lambda_{-i} \frac{\partial f_{i-1}^{s' \cdot \pi}}{\partial x_k}(y_{\leq i-1})\\
= &\frac{\partial f_i^{s \cdot \pi}}{\partial x_k}(y_{\leq i}) + \left( \lambda_{i} \frac{\partial \phi_i^{\pi_i}}{\partial x_k}(y_{\leq i-1}) + \lambda_{-i} \frac{\partial \phi_i^{-\pi_i}}{\partial x_k}(y_{\leq i-1}) \right) \cdot \frac{\partial f_i^{s \cdot \pi}}{\partial x_i}(y_{\leq i})\\
= &\frac{\partial f_i^{s \cdot \pi}}{\partial x_k}(y_{\leq i}) + \lambda_i \cdot a_{ik} \cdot \frac{\partial f_i^{s \cdot \pi}}{\partial x_i}(y_{\leq i})
\end{split}
\end{equation}
where we used $\lambda_{s_i \cdot i} + \lambda_{-s_i \cdot i} = 1$ and \cref{clm:pi-properties}. Thus, we can now write
\begin{align*}
&\sum_{s: s_j = 1 \; \forall j \leq i-1} \left( \prod_{j=i}^{n} \lambda_{s_j \cdot j} \right) \frac{\partial f_{i-1}^{s \cdot \pi}}{\partial x_k} (y_{\leq i-1})\\
= &\sum_{s: s_j = 1 \; \forall j \leq i} \left( \prod_{j=i+1}^{n} \lambda_{s_j \cdot j} \right) \cdot \left( \lambda_i \frac{\partial f_{i-1}^{s \cdot \pi}}{\partial x_k} (y_{\leq i-1}) + \lambda_{-i} \frac{\partial f_{i-1}^{s' \cdot \pi}}{\partial x_k} (y_{\leq i-1}) \right)\\
= &\sum_{s: s_j = 1 \; \forall j \leq i} \left( \prod_{j=i+1}^{n} \lambda_{s_j \cdot j} \right) \cdot \left( \frac{\partial f_i^{s \cdot \pi}}{\partial x_k}(y_{\leq i}) + \lambda_i \cdot a_{ik} \cdot \frac{\partial f_i^{s \cdot \pi}}{\partial x_i}(y_{\leq i}) \right)\\
= &\sum_{s: s_j = 1 \; \forall j \leq i} \left( \prod_{j=i+1}^{n} \lambda_{s_j \cdot j} \right) \cdot \frac{\partial f_i^{s \cdot \pi}}{\partial x_k}(y_{\leq i}) + \lambda_i \cdot a_{ik} \cdot \sum_{s: s_j = 1 \; \forall j \leq i} \left( \prod_{j=i+1}^{n} \lambda_{s_j \cdot j} \right) \cdot \frac{\partial f_i^{s \cdot \pi}}{\partial x_i}(y_{\leq i})\\
= &\frac{1}{\delta^{n+1}} \cdot \frac{\partial p_i}{\partial y_k}(y,z) + \lambda_i \cdot a_{ik} \frac{1}{\delta^{n+1}} \frac{\partial p_i}{\partial y_i}(y,z)
\end{align*}
where we used \eqref{eq:clm-induction-proof} and the induction hypothesis for $i$. Thus, it remains to show that
$$\frac{\partial p_i}{\partial y_k}(y,z) + \lambda_i \cdot a_{ik} \frac{\partial p_i}{\partial y_i}(y,z) = \frac{\partial p_{i-1}}{\partial y_k}(y,z)$$
to establish that the induction hypothesis also holds for $i-1$. By definition of $p_{i-1}$, we have that
$$\frac{\partial p_{i-1}}{\partial y_k}(y,z) = \frac{\partial p_i}{\partial y_k}(y,z) + \delta^i \frac{\partial q_i}{\partial y_k}(y,z)$$
which means that it suffices to prove that
$$\frac{\partial q_i}{\partial y_k}(y,z) = \frac{1}{\delta^i} \cdot \lambda_i \cdot a_{ik} \frac{\partial p_i}{\partial y_i}(y,z).$$
We have
$$\frac{\partial q_i}{\partial y_k}(y,z) = -2 a_{ik} \left(y_i + K z_i^+ - K z_i^- - \sum_{j=1}^{i-1} a_{ij} y_j - c_i\right) = -2 a_{ik} \left(y_i - \trunc\left(\sum_{j=1}^{i-1} a_{ij} y_j + c_i\right)\right)$$
where we also used \cref{lem:KKT-aux-variables}. We can rewrite this as
$$\frac{\partial q_i}{\partial y_k}(y,z) = \frac{1}{\delta^i} \cdot \lambda_i \cdot a_{ik} \frac{\partial p_i}{\partial y_i}(y,z)$$
by using the fact that $\lambda_i$ satisfies \eqref{eq:def-lambda-y}, i.e.,
$$y_i - \trunc\left(\sum_{j=1}^{i-1} a_{i j} y_j + c_i\right) = -\frac{1}{2 \delta^i} \cdot \frac{\partial p_i}{\partial y_i}(y,z) \cdot \lambda_i .$$
Thus, the induction hypothesis also holds for $i-1$ and the proof is complete.
\end{proof}

\section{\cls/-hardness of \linearKKT}
\label{sec:linear-KKT-hardness}

In this section we prove the following lower bound.

\begin{proposition}\label{prop:linear-kkt-CLS-hard}
The \linearKKT problem is \cls/-hard.
\end{proposition}

The section is structured as follows. We begin with the formal description of the mesa functions (\cref{sec:mesa}), and then prove some key properties about these functions (\cref{sec:mesa-properties}). Finally, we use these to provide a proof of \cref{prop:linear-kkt-CLS-hard}. One key lemma (\cref{lem:circuit-construction}) in that proof, which pertains to how the mesa functions can be implemented by circuits in a robust way, is deferred to \cref{sec:circuit-implementation}.

\subsection{The Mesa Functions}
\label{sec:mesa}

Fix some side length $\ell > 0$, and a boundary slope $\Gamma > 0$.
A mesa function is defined by a center-point $p \in \mathbb{R}^2$, a value $a \in \mathbb{R}$, and a gradient $g \in \mathbb{R}^2$. Formally, we define the mesa function $M(\cdot;p,a,g): \mathbb{R}^2 \to \mathbb{R}$ as the minimum of five linear (affine) functions, namely for all $x \in \mathbb{R}^2$
$$M(x;p,a,g) = \min\big\{P_c(x;p,a,g), P_r(x;p,a,g), P_t(x;p,a,g), P_l(x;p,a,g), P_b(x;p,a,g)\big\}$$
where these linear functions are defined as follows:
\begin{itemize}
\item (center) The linear function with gradient $g = (g_1,g_2)$ that has value $a$ at point $p$. Formally, this is the function
$$P_c(x;p,a,g) = (x_1-p_1)g_1 + (x_2-p_2)g_2 + a.$$
\item (right) The linear function with gradient $(-\Gamma + g_1,g_2)$ that has value $P_c(p_1+\ell/2,p_2+\ell/2) = (g_1 + g_2)\ell/2 + a$ at point $(p_1+\ell/2,p_2+\ell/2)$. Formally, this is the function
\begin{equation*}
\begin{split}
P_r(x;p,a,g) &= (x_1-p_1-\ell/2)(-\Gamma + g_1) + (x_2-p_2-\ell/2)g_2 + (g_1 + g_2)\ell/2 + a\\
&= (x_1-p_1)(-\Gamma + g_1) + (x_2-p_2)g_2 + a + \Gamma \ell/2.
\end{split}
\end{equation*}
\item (top) The linear function with gradient $(g_1,-\Gamma + g_2)$ that has value $P_c(p_1+\ell/2,p_2+\ell/2) = (g_1 + g_2)\ell/2 + a$ at point $(p_1+\ell/2,p_2+\ell/2)$. Formally, this is the function
\begin{equation*}
\begin{split}
P_t(x;p,a,g) &= (x_1-p_1-\ell/2)g_1 + (x_2-p_2-\ell/2)(-\Gamma + g_2) + (g_1 + g_2)\ell/2 + a\\
&= (x_1-p_1)g_1 + (x_2-p_2)(-\Gamma + g_2) + a + \Gamma \ell/2.
\end{split}
\end{equation*}
\item (left) The linear function with gradient $(\Gamma + g_1,g_2)$ that has value $P_c(p_1-\ell/2,p_2-\ell/2) = -(g_1 + g_2)\ell/2 + a$ at point $(p_1-\ell/2,p_2-\ell/2)$. Formally, this is the function
\begin{equation*}
\begin{split}
P_l(x;p,a,g) &= (x_1-p_1+\ell/2)(\Gamma + g_1) + (x_2-p_2+\ell/2)g_2 -(g_1 + g_2)\ell/2 + a\\
&= (x_1-p_1)(\Gamma + g_1) + (x_2-p_2)g_2 + a + \Gamma \ell/2.
\end{split}
\end{equation*}
\item (bottom) The linear function with gradient $(g_1,\Gamma + g_2)$ that has value $P_c(p_1-\ell/2,p_2-\ell/2) = -(g_1 + g_2)\ell/2 + a$ at point $(p_1-\ell/2,p_2-\ell/2)$. Formally, this is the function
\begin{equation*}
\begin{split}
P_b(x;p,a,g) &= (x_1-p_1+\ell/2)g_1 + (x_2-p_2+\ell/2)(\Gamma + g_2) -(g_1 + g_2)\ell/2 + a\\
&= (x_1-p_1)g_1 + (x_2-p_2)(\Gamma + g_2) + a + \Gamma \ell/2.
\end{split}
\end{equation*}
\end{itemize}

\begin{remark}
It will be convenient to abuse notation and to also define the function $M(\cdot;p,A,g)$, where $A$ is a vector of values, namely $A = (a_c, a_r, a_t, a_l, a_b) \in \mathbb{R}^5$. In that case, we let
$$M(x; p, A, g) = \min\left\{P_c(x; p, a_c, g), P_r(x; p, a_r, g), P_t(x; p, a_t, g), P_l(x; p, a_l, g), P_b(x; p, a_b, g)\right\}$$
i.e., the values $a$ used by the different linear pieces are not (necessarily) the same and are given by the vector $A$. This will be useful to analyze the effect of perturbations.
\end{remark}

\subsection{Key Properties of the Mesa Construction}\label{sec:mesa-properties}

\begin{lemma}\label{lem:mesa-drop}
For any $\ell > 0$, $\Gamma > 0$, $p \in [0,1]^2$, $g \in [-1,1]^2$, and $A = (a_c, a_r, a_t, a_l, a_b) \in [0,1]^5$, we have that for all $x \in [0,1]^2$ with $\|x-p\|_\infty \geq \ell/2 + 3/\Gamma$
$$M(x; p, A, g) \leq 0.$$
\end{lemma}

For example, if $\Gamma \geq 12/\ell$, then $M(x; p, A, g) \leq 0$ for all $x \in [0,1]^2$ with $\|x-p\|_\infty \geq 3\ell/4$.

\begin{proof}
First, note that for any $x \in [0,1]^2$ we have
\begin{equation*}
\begin{split}
P_r(x;p,a_r,g) &= (x_1-p_1)(-\Gamma + g_1) + (x_2-p_2)g_2 + a_r + \Gamma \ell/2\\
&= -(x_1-p_1-\ell/2)\Gamma + (x_1-p_1)g_1 + (x_2-p_2)g_2 + a_r\\
&\leq -(x_1-p_1-\ell/2)\Gamma + 1 + 1 + 1
\end{split}
\end{equation*}
and thus $P_r(x;p,a_r,g) \leq 0$ whenever $x_1 \geq p_1 + \ell/2 + 3/\Gamma$. Similarly, one can also show that for all $x \in [0,1]^2$
\begin{itemize}
    \item $P_t(x;p,a_t,g) \leq 0$ whenever $x_2 \geq p_2 + \ell/2 + 3/\Gamma$,
    \item $P_l(x;p,a_l,g) \leq 0$ whenever $x_1 \leq p_1 - \ell/2 - 3/\Gamma$,
    \item $P_b(x;p,a_b,g) \leq 0$ whenever $x_2 \leq p_2 - \ell/2 - 3/\Gamma$.
\end{itemize}
Now, since the mesa function is the minimum of these four linear functions (and of the piece $P_c$), it follows that $M(x; p, A, g) \leq 0$ for all $x \in [0,1]^2$ with $\|x-p\|_\infty \geq \ell/2 + 3/\Gamma$, as claimed.
\end{proof}

\begin{lemma}\label{lem:mesa-value}
Let the following be given: $\ell = 1/N$ for some $N \in \mathbb{N}$, $\Gamma > 0$, and for all $p \in G$, $A^p = (a_c^p, a_r^p, a_t^p, a_l^p, a_b^p) \in [0.4,0.6]^5$ and $g^p \in [-0.01,0.01]$, where $G = \{0,\ell,2\ell,\cdots,1-\ell,1\}^2$. Then we have that for all $x \in [0,1]^2$
$$\max_{p \in G} M\big(x;p,A^p,g^p\big) \in [1/3,2/3].$$
\end{lemma}

\begin{proof}
We begin by proving the upper bound, namely $\max_{p \in G} M(x;p,A^p,g^p) \leq 2/3$. We will show that for all $p \in G$ and all $x \in [0,1]^2$, $M(x;p,A^p,g^p) \leq 2/3$. By construction of the mesa function, it suffices to show that $P_c(x;p,a_c^p,g^p) \leq 2/3$. For any $x \in [0,1]^2$ we have
$$P_c(x;p,a_c^p,g^p) = (x_1-p_1)g_1^p + (x_2-p_2)g_2^p + a_c^p \leq |g_1^p| + |g_2^p| + a_c^p \leq 0.62 \leq 2/3.$$

For the lower bound, fix any $x \in [0,1]^2$ and let $p \in G$ be such that $\|x-p\|_\infty \leq \ell/2$. (Such a $p$ must necessarily exist.) It suffices to show that $M(x;p,A^p,g^p) \geq 1/3$, since we take the maximum over all mesas. Let $a^* := \min\{a_c^p, a_r^p, a_t^p, a_l^p, a_b^p\}$. Then, we must have $M(x;p,A^p,g^p) \geq M(x;p,a^*,g^p)$, so now it suffices to show that $M(x;p,a^*,g^p) \geq 1/3$. By construction, we have that since $\|x-p\|_\infty \leq \ell/2$, $M(x;p,a^*,g^p) = P_c(x;p,a^*,g^p)$ and then
$$P_c(x;p,a^*,g^p) = (x_1-p_1)g_1^p + (x_2-p_2)g_2^p + a^* \geq -|g_1^p| -|g_2^p| + a^* \geq 0.38 \geq 1/3$$
which proves the lemma.
\end{proof}

\begin{lemma}\label{lem:mesa-bad-pieces}
Let the following be given: $\ell = 1/N$ for some $N \in \mathbb{N}$, $\Gamma \geq 6/\ell$, and for all $p \in G$, $A^p = (a_c^p, a_r^p, a_t^p, a_l^p, a_b^p) \in [0.4,0.6]^5$ and $g^p \in [-0.01,0.01]^2$, where $G = \{0,\ell,2\ell,\cdots,1-\ell,1\}^2$. Furthermore, assume that for all $p, p' \in G$ that are adjacent (i.e., $\|p-p'\|_\infty \leq \ell$) we have
\begin{enumerate}
    \item $\|g^p - g^{p'}\|_\infty \leq \ell$,
    \item for all $i,j \in \{c,r,t,l,b\}$, $|a_i^{p'} - a_j^p - \langle 2g^p, p'-p \rangle| \leq \ell^2$.
\end{enumerate}
Define the function $f: [0,1]^2 \to \mathbb{R}$
$$f(x) := \max_{p \in G} M\big(x;p,A^p,g^p\big).$$
Then for any $p \in G$ and any $x \in [0,1]^2$ with $\|x-p\|_\infty \leq \ell/2$ such that $f$ is differentiable at $x$, we have
\begin{itemize}
\item If $g_1^p \geq 10\ell$, then $\frac{\partial f}{\partial x_1} (x) \geq g_1^p - \ell$.
\item If $g_1^p \leq -10\ell$, then $\frac{\partial f}{\partial x_1} (x) \leq g_1^p + \ell$.
\item If $g_2^p \geq 10\ell$, then $\frac{\partial f}{\partial x_2} (x) \geq g_2^p - \ell$.
\item If $g_2^p \leq -10\ell$, then $\frac{\partial f}{\partial x_2} (x) \leq g_2^p + \ell$.
\end{itemize}
\end{lemma}

\begin{proof}
Let $p^* \in G$ and $x^* \in [0,1]^2$ with $\|x^*-p^*\|_\infty \leq \ell/2$ such that $f$ is differentiable at $x^*$. We only consider the case where $g_1^{p^*} \geq 10\ell$. The other cases are handled analogously.

First, observe that for all $x \in [0,1]^2$ with $\|x-p^*\|_\infty \leq \ell$ we have
$$f(x) = \max_{p \in G_{p^*}} M(x;p,A^{p},g^{p})$$
where $G_{p^*} := \{p \in G: \|p-p^*\|_\infty \leq \ell\}$ is the set of grid points adjacent to $p^*$. Indeed, for any $p \in G \setminus G_{p^*}$ we have $\|x-p\|_\infty \geq \ell \geq \ell/2 + 3/\Gamma$ and thus by \cref{lem:mesa-drop} it must be that $M(x;p,A^{p},g^{p}) \leq 0$. On the other hand, by \cref{lem:mesa-value} we have $f(x) \in [1/3,2/3]$, which implies that $f(x) > M(x;p,A^{p},g^{p})$ for all $p \in G \setminus G_{p^*}$.

Now, since $f$ is differentiable at $x^*$, $\nabla f(x^*)$ must correspond to the gradient of one of the linear affine pieces of $M(\cdot; p, A^{p}, g^{p})$ for some $p \in G_{p^*}$. By construction of the mesa function $M(\cdot; p, A^{p}, g^{p})$, its linear affine pieces $P_c, P_t$ and $P_b$ all have gradient equal to $g_1^p$ in the first coordinate. Furthermore, the linear affine piece $P_l$ has gradient $g_1^p + \Gamma \geq g_1^p$ in the first coordinate. Thus, as long as the piece $P_r$ does not appear in $f$ (i.e., $P_r(x;p,A^p,g^p) < f(x)$ whenever $P_r(x;p,A^p,g^p) = M(x; p, A^{p}, g^{p})$), we can deduce that $\frac{\partial f}{\partial x_1} (x^*) \geq g_1^p$. This implies the desired result that $\frac{\partial f}{\partial x_1} (x^*) \geq g_1^{p^*} - \ell$, since by assumption 1 we have $\|g^{p^*} - g^{p}\|_\infty \leq \ell$. 

It remains to prove that for all $p \in G_{p^*}$ the piece $P_r$ indeed does not appear in $f$. Since $g_1^p \geq g_1^{p^*} - \ell \geq 9\ell$, this is implied by the following slightly more general claim: If $g_1^p \geq 9\ell$ for some $p \in G$, then for all $x \in [0,1]^2$ we have
$$P_r(x; p, a_r^p, g^p) = M(x; p, A^p, g^p) \implies P_r(x; p, a_r^p, g^p) < f(x).$$
In the remainder of this proof, we prove this claim.

Let $p \in G$ with $g_1^p \geq 9\ell$. We begin by noting the following fact: for any $p\in G$ and $i,j \in \{c,r,t,l,b\}$, we have
\begin{equation}\label{eq:mesa-gap}
|a_i^p-a_j^p| \leq |a_i^{p'} - a_i^p - \langle 2g^p, p'-p \rangle| + |a_i^{p'} - a_j^p - \langle 2g^p, p'-p \rangle| \leq 2\ell^2
\end{equation}
by assumption 2, where $p'$ is some arbitrary gridpoint adjacent to $p$.

Furthermore, note that since $\Gamma \geq 6/\ell$, \cref{lem:mesa-drop} implies that $M(x; p, A^p, g^p) \leq 0$ for all $x \in [0,1]^2$ with $\|x-p\|_\infty \geq \ell$. On the other hand, by \cref{lem:mesa-value} we have $f(x) \geq 1/3$ for all $x \in [0,1]^2$. As a result, for any $x \in [0,1]^2$ with $\|x-p\|_\infty \geq \ell$, if $P_r(x; p, a_r^p, g^p) = M(x; p, A^p, g^p)$, then $P_r(x; p, a_r^p, g^p) \leq 0 < 1/3 \leq f(x)$, as desired. Thus, it remains to prove the claim for all $x \in ([p_1-\ell,p_1+\ell] \times [p_2-\ell,p_2+\ell]) \cap [0,1]^2 =: B_\ell(p) \cap [0,1]^2$.

First of all, note that if $p$ lies on the right-hand boundary of $G$, i.e., $p_1 = 1$, then $P_r(x; p, a_r^p, g^p) > P_c(x; p, a_c^p, g^p) \geq M(x; p, A^p, g^p)$ for all $x \in [0,1]^2$, and thus the claim holds trivially. Indeed
\begin{equation*}
\begin{split}
& \quad P_r(x; p, a_r^p, g^p) - P_c(x; p, a_c^p, g^p)\\
&= (x_1-p_1)(-\Gamma + g_1^p) + (x_2-p_2)g_2^p + a_r^p + \Gamma \ell/2 -(x_1-p_1)g_1^p - (x_2-p_2)g_2^p - a_c^p\\
&= -(x_1-p_1)\Gamma + \Gamma \ell/2 + a_r^p - a_c^p\\
&\geq \Gamma \ell/2 - 2\ell^2 = \ell(\Gamma/2 - 2\ell) > 0
\end{split}
\end{equation*}
where we used $x_1 \leq 1 = p_1$, $a_r^p - a_c^p \geq -2\ell^2$ by \eqref{eq:mesa-gap}, and $\Gamma \geq 6/\ell \geq 6 > 4\ell$.

Now consider the case where $p$ does not lie on the right-hand boundary of $G$. Then the point $p' := (p_1 + \ell, p_2)$ lies on the grid $G$. We will show that for all $x \in B_\ell(p)$, $P_r(x; p, a_r^p, g^p) = M(x; p, A^p, g^p) \implies P_r(x; p, a_r^p, g^p) < M(x; p', A^{p'}, g^{p'})$, which implies the claim, since $M(x; p', A^{p'}, g^{p'}) \leq f(x)$. For ease of notation, we drop the $p, p'$ superscripts, and use $A := A^p$, $A' := A^{p'}$, $a_i := a_i^p$, $a_i' := a_i^{p'}$, $g := g^p$, $g' := g^{p'}$. We will prove the contrapositive, namely $P_r(x; p, a_r, g) \geq M(x; p', A', g') \implies P_r(x; p, a_r, g) > M(x; p, A, g)$ for all $x \in B_\ell(p)$. We proceed by a case analysis. If $x$ is such that $P_r(x; p, a_r, g) \geq M(x; p', A', g')$, then at least one of the following five cases must occur.

\paragraph*{Case 1:} $P_r(x; p, a_r, g) \geq P_c(x; p', a_c', g')$. In this case, we show that $P_r(x; p, a_r, g) > P_c(x; p, a_c, g) \geq M(x; p, A, g)$. Indeed, we have
\begin{equation*}
\begin{split}
P_r(x; p, a_r, g) - P_c(x; p, a_c, g) &\geq P_c(x; p', a_c', g') - P_c(x; p, a_c, g)\\
&= (x_1-p_1-\ell)(g_1' - g_1) - \ell g_1 + (x_2-p_2)(g_2' - g_2) + a_c' - a_c\\
&\geq - 2\ell |g_1' - g_1| - \ell |g_2' - g_2| - \ell g_1 + a_c' - a_c\\
&\geq -3 \ell^2 + \ell g_1 - \ell^2 = \ell (g_1 - 4\ell) > 0
\end{split}
\end{equation*}
where we used assumptions 1 and 2, as well as $g_1 > 4\ell$. Indeed, note that assumption 2 implies that $a_c'-a_c -2\ell g_1 \geq -\ell^2$ and thus $a_c'-a_c -\ell g_1 \geq \ell g_1 - \ell^2$, which we used here.

\paragraph*{Case 2:} $P_r(x; p, a_r, g) \geq P_t(x; p', a_t', g')$. In this case, we show that $P_r(x; p, a_r, g) > P_t(x; p, a_t, g) \geq M(x; p, A, g)$.
Indeed, we have
\begin{equation*}
\begin{split}
P_r(x; p, a_r, g) - P_t(x; p, a_t, g) &\geq P_t(x; p', a_t', g') - P_t(x; p, a_t, g)\\
&= (x_1-p_1-\ell)(g_1' - g_1) - \ell g_1 + (x_2-p_2)(g_2' - g_2) + a_t' - a_t\\
&\geq - 2\ell |g_1' - g_1| - \ell |g_2' - g_2| - \ell g_1 + a_t' - a_t\\
&\geq \ell (g_1 - 4\ell) > 0
\end{split}
\end{equation*}
as in the previous case.

\paragraph*{Case 3:} $P_r(x; p, a_r, g) \geq P_b(x; p', a_b', g')$. In this case, we show that $P_r(x; p, a_r, g) > P_b(x; p, a_b, g) \geq M(x; p, A, g)$.
Indeed, we have
\begin{equation*}
\begin{split}
P_r(x; p, a_r, g) - P_b(x; p, a_b, g) &\geq P_b(x; p', a_b', g') - P_b(x; p, a_b, g)\\
&= (x_1-p_1-\ell)(g_1' - g_1) - \ell g_1 + (x_2-p_2)(g_2' - g_2) + a_b' - a_b\\
&\geq - 2\ell |g_1' - g_1| - \ell |g_2' - g_2| - \ell g_1 + a_b' - a_b\\
&\geq \ell (g_1 - 4\ell) > 0
\end{split}
\end{equation*}
as in the previous cases.

\paragraph*{Case 4:} $P_r(x; p, a_r, g) \geq P_r(x; p', a_r', g')$. This case is in fact impossible, since
\begin{equation*}
\begin{split}
P_r(x; p', a_r', g') - P_r(x; p, a_r, g) &= (x_1-p_1-\ell)(g_1' - g_1) + \Gamma \ell - \ell g_1 + (x_2-p_2)(g_2' - g_2) + a_r' - a_r\\
&\geq -2 \ell |g_1' - g_1| - \ell |g_2' - g_2|  + \Gamma \ell - \ell g_1 + a_r' - a_r\\
&\geq \ell(\Gamma + g_1 - 4 \ell) \geq \ell (g_1 - 4\ell) > 0
\end{split}
\end{equation*}
as in the previous cases.

\paragraph*{Case 5:} $P_r(x; p, a_r, g) \geq P_l(x; p', a_l', g')$. In this case, we show that $P_r(x; p, a_r, g) > P_c(x; p, a_c, g) \geq M(x; p, A, g)$. Indeed, assume towards a contradiction that $P_r(x; p, a_r, g) \leq P_c(x; p, a_c, g)$. Then
$$0 \leq P_c(x; p, a_c, g) - P_r(x; p, a_r, g) = (x_1-p_1)\Gamma + a_c - a_r - \Gamma \ell / 2$$
which yields $x_1 \geq (a_r-a_c)/\Gamma + \ell/2 + p_1$. But then, we have
\begin{equation*}
\begin{split}
& \quad P_l(x; p', a_l', g') - P_r(x; p, a_r, g)\\
&= (x_1-p_1-\ell)(2\Gamma + g_1' - g_1) + \Gamma \ell - \ell g_1 + (x_2-p_2)(g_2' - g_2) + a_l' - a_r\\
&\geq \Gamma (\ell + 2(x_1-p_1-\ell)) -2 \ell |g_1' - g_1| - \ell |g_2' - g_2| - \ell g_1 + a_l' - a_r\\
&\geq 2(a_r-a_c) -2 \ell |g_1' - g_1| - \ell |g_2' - g_2| - \ell g_1 + a_l' - a_r\\
&\geq \ell (g_1 - 8\ell) > 0
\end{split}
\end{equation*}
where we used assumptions 1 and 2, as well as $a_r-a_c \geq -2\ell^2$ by \eqref{eq:mesa-gap} and $g_1 > 8\ell$. We have thus obtained a contradiction and the claim holds in this case as well.
\end{proof}

\subsection{Proof of Proposition~\ref{prop:linear-kkt-CLS-hard}}
\label{sec:mesa-proof-of-prop}

We reduce from the problem of computing an $\eps$-KKT point of a continuously differentiable function $h: [0,1]^2 \to \mathbb{R}$ with $L$-Lipschitz gradient $\nabla h$, where both $h$ and $\nabla h$ are given as well-behaved arithmetic circuits. This problem is known to be \cls/-complete~\cite{FGHS22}. We do not need to define well-behaved arithmetic circuits here, since we will only use the fact that they can be evaluated in polynomial time (see \cite{FGHS22}).

Without loss of generality, we can assume that $h$ satisfies:
\begin{itemize}
    \item $h(x) \in [0.49,0.51]$ for all $x \in [0,1]^2$,
    \item $\nabla h(x) \in [-0.001,0.001]^2$ for all $x \in [0,1]^2$,
    \item $L \leq 1/100$.
\end{itemize}
Indeed, this is easy to achieve by scaling $h$ and $\eps$ by some sufficiently small factor, and by adding a constant offset to $h$. This does not change the set of solutions.

\paragraph*{\bf Construction.}
Set $\ell := 1/(2^n-1)$ for some sufficiently large $n \in \mathbb{N}$ such that $\ell \leq \eps/100$. Define the grid $G = \{0,\ell,2\ell,\cdots,1-\ell,1\}^2$.

Since $h$ and $\nabla h$ are given as well-behaved arithmetic circuits, we can construct the following two Boolean circuits in polynomial time:
\begin{itemize}
    \item $a: G \to [0.45,0.55]$ that satisfies $|a(p) - h(p)| \leq \ell^2/100$ for all $p \in G$,
    \item $g: G \to [-0.01,0.01]^2$ that satisfies $\|2g(p) - \nabla h(p)\|_\infty \leq \ell/100$ for all $p \in G$.
\end{itemize}
Note that this is where we use the ``half-gradient'' trick, namely we let $g(p)$ be close to $\nabla h(p)/2$, instead of $\nabla h(p)$.

Set $\Gamma := 12/\ell$, and $\delta' := \ell^2/100$. The following lemma states that we can now construct a circuit $\circuit$ which implements the corresponding mesa construction in a robust manner.

\begin{lemma}\label{lem:circuit-construction}
Let the following be given: $\delta' > 0$, $\ell = 1/(2^n-1)$ for some $n \in \mathbb{N}$, $\Gamma \geq 12/\ell$, and Boolean circuits\footnote{The outputs of the two Boolean circuits are represented in binary as described in \cref{sec:binary-circuits}.} computing functions:
\begin{itemize}
    \item $a: G \to [0.45,0.55]$
    \item $g: G \to [-0.01,0.01]^2$
\end{itemize}
where $G = \{0,\ell,2\ell,\cdots,1-\ell,1\}^2$. Then, in polynomial time we can compute $\delta > 0$ and construct a linear arithmetic circuit $\circuit$ which implements the mesa construction described by $(\ell, \Gamma, a, g)$ in a \emph{robust} manner, i.e., for every $\pi \in [-\delta,\delta]^m$ (where $m$ is the number of gates of $\circuit$), there exists a vector $\tau \in [-\delta',\delta']^{G\times\{c,r,t,l,b\}}$ such that for all $x \in [0,1]^2$
$$\circuit^\pi(x) = \max_{p \in G} M\big(x;p,A^{p,\tau},g(p)\big)$$
where $A^{p,\tau}$ is the following five-tuple
$$A^{p,\tau} := \big(a(p)+\tau_{p,c},a(p)+\tau_{p,r},a(p)+\tau_{p,t},a(p)+\tau_{p,l},a(p)+\tau_{p,b}\big).$$
\end{lemma}

We will now complete the proof of \cref{prop:linear-kkt-CLS-hard} assuming this lemma. The next section will then be dedicated to proving the lemma.

By \cref{lem:circuit-construction} we obtain $\delta > 0$ and a circuit $\circuit$ such that for all $x \in [0,1]^2$ and all perturbations $\pi \in [-\delta,\delta]^m$ we have
$$\circuit^\pi(x) = \max_{p \in G} M\big(x;p,A^{p},g(p)\big)$$
where $A^p = (a_c^p, a_r^p, a_t^p, a_l^p, a_b^p)$ satisfies $|a_i^p - a(p)| \leq \delta'$ for all $i$ and $p$. The following claim shows that the function computed by the circuit $\circuit^\pi$ satisfies the conditions of \cref{lem:mesa-bad-pieces}.

\begin{claim}
For all adjacent $p,p' \in G$, we have
\begin{enumerate}
    \item $\|g(p) - g(p')\|_\infty \leq \ell$,
    \item for all $\tau \in [-\delta',\delta']^{G\times\{c,r,t,l,b\}}$ and all $i,j \in \{c,r,t,l,b\}$
    \begin{enumerate}
        \item $a_i^p \in [0.4,0.6]$, and
        \item $|a_i^{p'} - a_j^p - \langle 2g(p), p'-p \rangle| \leq \ell^2$.
    \end{enumerate}
\end{enumerate}
\end{claim}

\begin{proof}
For the first point, note that by the $L$-Lipschitz-continuity of $\nabla h$, we have $\|\nabla h(p) - \nabla h(p')\|_\infty \leq L \ell \leq \ell/100$, since $L \leq 1/100$. Thus, by construction of $g(\cdot)$ we obtain
\begin{equation*}
\|g(p) - g(p')\|_\infty \leq \|\nabla h(p)/2 - \nabla h(p')/2\|_\infty + 2 \cdot \ell/200 \leq \ell/200 + 2 \cdot \ell/200 \leq \ell.
\end{equation*}

For the second point, note that $a_i^p \in [0.4,0.6]$, because $a(p) \in [0.45,0.55]$ by construction, and $|a_i^p - a(p)| \leq \delta' \leq \ell^2/100 \leq 0.05$. Next, by Taylor's theorem, we also have
$$|h(p') - h(p) - \langle \nabla h(p), p' - p \rangle| \leq \frac{L}{2} \|p' - p\|_2^2 \leq \frac{1}{100} \ell^2$$
where we used $L \leq 1/100$ and $\|p' - p\|_2 \leq \sqrt{2}\ell$. Now, for all $i,j \in \{c,r,t,l,b\}$, we can rewrite this as
\begin{equation*}
|a_i^{p'} - a_j^p - \langle 2g(p), p'-p \rangle| \leq \frac{1}{100} \ell^2 + 2\cdot \ell^2/50 + \ell^2/50 \leq \ell^2
\end{equation*}
where we used $|a_j^{p} - h(p)| \leq |a_j^{p} - a(p)| + |a(p) - h(p)| \leq \delta' + \ell^2/100 \leq \ell^2/50$ (since $\delta' \leq \ell^2/100$) and similarly $|a_i^{p'} - h(p')| \leq \ell^2/50$, as well as
$$|\langle 2g(p) - \nabla h(p), p'-p \rangle| \leq \|2g(p) - \nabla h(p)\|_2 \cdot \|p'-p\|_2 \leq 2 \ell \cdot \|2g(p) - \nabla h(p)\|_\infty \leq \ell^2/50$$
by construction of $g(\cdot)$.
\end{proof}

Now we can use \cref{lem:mesa-bad-pieces} to prove the following claim, which finishes the proof of \cref{prop:linear-kkt-CLS-hard}.

\begin{claim}
If $x \in [0,1]^2$ satisfies the $\eps/3$-KKT conditions with respect to the $\delta$-generalized circuit gradient of $\circuit$, then $x$ also satisfies the $\eps$-KKT conditions with respect to the gradient of $h$.
\end{claim}

\begin{proof}
We prove the contrapositive. Let $x \in [0,1]^2$ be a point that does not satisfy the $\eps$-KKT conditions with respect to the gradient of $h$. In other words, there exists $j \in \{1,2\}$ such that at least one of the two following statements holds
\begin{itemize}
\item $x_j > 0$ and $\frac{\partial h}{\partial x_j} (x) > \eps$
\item $x_j < 1$ and $\frac{\partial h}{\partial x_j} (x) < -\eps$
\end{itemize}
In the remainder of this proof, we only consider the case where $x_1 > 0$ and $\frac{\partial h}{\partial x_1} (x) > \eps$. The other three cases are handled analogously. Our goal is to show that $x$ cannot satisfy the $\eps/3$-KKT conditions with respect to the $\delta$-generalized circuit gradient of $\circuit$. Namely, we will show that all $u \in \widetilde{\partial}_\delta \circuit (x)$ satisfy $u_1 > \eps/3$. By the definition of $\delta$-generalized circuit gradient, it suffices to show that $\frac{\partial f^\pi}{\partial x_1} (x) > \eps/3$ for all $\pi \in [-\delta,\delta]^m$ such that $f^\pi$ is differentiable at $x$, where $f^\pi$ denotes the function computed by $\circuit^\pi$.

Let $\pi \in [-\delta,\delta]^m$ be such that $f^\pi$ is differentiable at $x$. Let $p \in G$ be such that $\|x-p\|_\infty \leq \ell/2$. By the $L$-Lipschitz-continuity of $h$ and the construction of $g(\cdot)$ we have
$$\|2g(p) - \nabla h(x)\|_\infty \leq \|2g(p) - \nabla h(p)\|_\infty + \|\nabla h(p) - \nabla h(x)\|_\infty \leq \ell/100 + L \cdot \ell/2 \leq \ell$$
where we used $L \leq 1/100$. As a result, we have
$$g_1(p) \geq \frac{1}{2} \cdot \frac{\partial h}{\partial x_1} (x) - \ell \geq \eps/2 - \ell \geq 10\ell$$
where we used $\ell \leq \eps/100$. As argued above, the assumptions of \cref{lem:mesa-bad-pieces} are satisfied and we can thus use it to deduce that $\frac{\partial f^\pi}{\partial x_1} (x) \geq g_1(p) - \ell \geq \eps/2 - 2\ell > \eps/3$, as desired.
\end{proof}

\section{Construction of the circuit}\label{sec:circuit-implementation}

In this section, we prove \cref{lem:circuit-construction}, which was crucially used in the last section to obtain \cref{prop:linear-kkt-CLS-hard}. In order to prove the lemma, we have to show that we can construct arithmetic circuits that implement the mesa construction from the last section in a robust manner.

In our description of the construction of the circuit, it will be convenient to extend the set of gates we can use.
Consider an arithmetic circuit $\circuit$ with $n$ inputs and one output, and using gates $+, -, c, \times c, \min$, $\max$, $\trunc_{[a,b]}$.\footnote{$\trunc_{[a,b]}$ denotes truncation to the interval $[a,b]$.} Let $m$ denote the number of $\min, \max, \trunc_{[a,b]}$ gates. For any perturbation vector $\pi = (\pi_i)_{i \in [m]} \in \mathbb{R}^m$, we let $\circuit^\pi$ denote the circuit perturbed by $\pi$, namely, for each $i \in [m]$, the $i$th gate of type $\min, \max, \trunc_{[a,b]}$ is perturbed as follows:
\begin{itemize}
    \item $x_i := \trunc_{[a,b]}(a_ix_j + b_ix_k+c_i)$ is replaced by $x_i := \trunc_{[a,b]}(a_ix_j + b_ix_k+c_i + \pi_i)$.
    \item $x_i := \min \{x_j,x_k\}$ is replaced by $x_i := \min \{x_j,x_k + \pi_i\}$.
    \item $x_i := \max \{x_j,x_k\}$ is replaced by $x_i := \max \{x_j,x_k + \pi_i\}$.
\end{itemize}

Our goal now is to prove the following lemma.

\begin{lemma}\label{lem:circuit-all-gates}
Let the following be given: $\delta' > 0$, $\ell = 1/(2^n-1)$ for some $n \in \mathbb{N}$, $\Gamma \geq 12/\ell$, and Boolean circuits computing functions:
\begin{itemize}
    \item $a: G \to [0.45,0.55]$
    \item $g: G \to [-0.01,0.01]^2$
\end{itemize}
where $G = \{0,\ell,2\ell,\cdots,1-\ell,1\}^2$. Then, in polynomial time we can compute $\delta > 0$ and construct a linear arithmetic circuit $\circuit$ (using gates $+, -, c, \times c, \min, \max, \trunc_{[a,b]}$) which implements the mesa construction described by $(\ell, \Gamma, a, g)$ in a \emph{robust} manner, i.e., for every $\pi \in [-\delta,\delta]^m$ (where $m$ is the number of $\min, \max, \trunc_{[a,b]}$ gates of $\circuit$), there exists a vector $\tau \in [-\delta',\delta']^{G\times\{c,r,t,l,b\}}$ such that for all $x \in [0,1]^2$
$$\circuit^\pi(x) = \max_{p \in G} M\big(x;p,A^{p,\tau},g(p)\big)$$
where $A^{p,\tau}$ is the following five-tuple
$$A^{p,\tau} := \big(a(p)+\tau_{p,c},a(p)+\tau_{p,r},a(p)+\tau_{p,t},a(p)+\tau_{p,l},a(p)+\tau_{p,b}\big).$$
\end{lemma}

Indeed, \cref{lem:circuit-all-gates} implies \cref{lem:circuit-construction}, because of the following lemma, which is proved in \cref{app:proof-perturbation-lemma}.

\begin{lemma}
\label{lem:perturbation}
Let $\circuit$ be an arithmetic circuit with $n$ inputs and one output, and using gates $+, -, c, \times c, \min, \max, \trunc_{[a,b]}$, and such that the output gate is a $\trunc_{[0,1]}$ gate. Then we can construct in polynomial time a number $K > 0$ and an arithmetic circuit $\circuitbar$ with the same number of inputs and outputs as $\circuit$, but that only uses $\trunc_{[0,1]}$ gates, and such that for any $\delta \leq 1/K$ and any $\delta$-perturbation $\pi$ of $\circuitbar$, there exists a $\delta K$-perturbation $\sigma$ of $\circuit$ that satisfies
$$f^\sigma(y) = \fbar^\pi(y)$$
for all $y \in [0,1]^n$, where $f^\sigma$ and $\fbar^\pi$ are the functions computed by $\circuit^\sigma$ and $\circuitbar^\pi$, respectively.
\end{lemma}

In the remainder of this section, we present an overview of the proof of \cref{lem:circuit-all-gates}, followed by the formal proof.

\paragraph{Rescaling the grid.}

\cref{lem:circuit-all-gates} uses the grid $G = \{0, 1/(2^n-1), \dots,
1\}^2$. This grid discretizes each dimension into $2^n$ points, and has points
along the boundaries at $0$ and $1$. Since we will use a standard binary
encoding to represent grid points, it will be more technically convenient to
use the grid $\widetilde{G} = \{0, 1/{2^n}, \dots, (2^{n-1} - 1)/2^n\}$, which
discretizes each dimension into $2^n$ points and has points along the
boundaries at $0$ and $1 - 1/2^n$. 

This can be achieved via a straightforward rescaling. From
\cref{lem:circuit-all-gates}
we are given Boolean
circuits $a : G \rightarrow [0.45, 0.55]$ and $g : G \rightarrow [-0.01,
0.01]^2$. We then build the following Boolean circuits.
\begin{itemize}
\item We build the circuit $\widetilde{a} : \widetilde{G} \rightarrow [0.45, 0.55]$. For each point $x = (x_1/2^n, x_2/2^n)$
we set 
\begin{equation*}
\widetilde{a}(x) = a(x_1/(2^n-1), x_2/(2^n-1)).
\end{equation*}
That is, $\widetilde{a}$ simply looks up the
corresponding point in $x' \in G$ and outputs $a(x')$. 

\item We build the circuit $\widetilde{g} : \widetilde{G} \rightarrow [-0.02, 0.02]^2$. 
For each point $x = (x_1/2^n, x_2/2^n)$ we set 
\begin{equation*}
\widetilde{g}(x) = \frac{1}{1 - 1/2^n} \cdot g(x_1/(2^n-1), x_2/(2^n-1))
\end{equation*}
So $\widetilde{g}$ looks up the corresponding point in $G$ and outputs the
corresponding gradient, but it also rescales that gradient to account for the
smaller space.
\end{itemize}
Observe that $\widetilde{a}$ and $\widetilde{g}$ can both be constructed in
polynomial time with respect to $a$ and $g$. 

In the rest of this section we will build a linear circuit that outputs, for
each point $x \in [0, 1 - 1/2^n]^2$ 
\begin{equation*}
\widetilde{f}(x) = \max_{p \in \widetilde{G}} M(x; p, \widetilde{a}(p), \widetilde{g}(p))
\end{equation*}
We can then build the final circuit for 
\cref{lem:circuit-all-gates} by undoing the rescaling in the following
way.
\begin{equation*}
f(x) = \widetilde{f}((1 - 1/2^n) \cdot x).
\end{equation*}
Clearly $f$ can be constructed in polynomial time once we have constructed
$f'$. Observe that, since we have undone the rescaling of the gradients in
$\widetilde{g}$, so we will have $f(x) = \max_{p \in G} M(x; p, a(p), g(p))$, as required.

\paragraph{\bf Important properties of the mesas.}

We will use two properties of the mesa construction repeatedly. 
\begin{enumerate}
\item From \cref{lem:mesa-value} we have that $f(x) \in [1/3,
2/3]$ for all $x \in [0, 1]$. Since $\widetilde{f}$ is a simple rescaling of
the domain of
$f$, which does not change the output values, we therefore have
$\widetilde f(x) \in [1/3, 2/3]$ for all $x \in [0, 1-1/2^n]$ as well. 

\item 
By \cref{lem:mesa-drop} and the choice of $\Gamma$ in the statement of
\cref{lem:circuit-all-gates}, we have 
$M(x; p, a(p), g(p)) \le 0$ whenever
$\| x - y \|_\infty > 3/4 \cdot 1/{(2^n-1)}$. In the rescaled 
$\widetilde f(x)$ we therefore have 
$\| x - y \|_\infty > 3/4 \cdot 1/{2^n}$, meaning that 
the region in which mesa takes positive values is contained within a
region of radius $3/4 \cdot 1/2^n$ around the center of the mesa. 
\end{enumerate}

\paragraph{\bf Overview.}

The rest of this section is dedicated to building a linear circuit that outputs
$\widetilde{f}$. We begin by building linear circuits that implement some
useful operations. In \cref{sec:binary-circuits} we build linear circuits that
can encode and decode values that are represented in binary, and we build a linear
circuit that can simulate a Boolean circuit.  In \cref{sec:mesa-circuits} we
build linear circuits that, given a point $p \in \tG$ encoded in
binary, can output the five affine pieces of $M(x; p, \ta(p), \tg(p))$. To do this, we first build a
linear circuit that can multiply a continuous variable by a binary variable,
and then we use that operation to build a circuit that can output a specific
affine function whose parameters are encoded in binary. Finally, in
\cref{sec:final-circuit} we build $\tf$ itself. 

Along the way we must consider how our circuits behave when they are evaluated
with perturbations. Recall from \cref{lem:circuit-all-gates} that the $\min$,
$\max$, and $\trunc$ operations will be perturbed, and that we must show that
these perturbations only additively affect the pieces of each of the mesas,
without affecting their gradients. 

To keep track of the effect of the perturbations, for each linear circuit $C$,
we will use $C^\pi$ to refer to the version of that circuit that is evaluated under
the perturbation vector $\pi$. Along the way, we will prove lemmas that bound
the deviation in the output of $C$ as a function of the perturbation vector~$\pi$.

\subsection{\bf Binary variables.}
\label{sec:binary-circuits}

Throughout the construction we will make use of variables that are
encoded in binary. An $n$-bit binary variable can represent values from the set
$V := \left\{ 0, 1/2^n, 2/2^n, \dots, (2^{n}-1)/2^n\right\}$. For each value $x \in
V$, we use $x_1, x_2, \dots, x_n$ to denote the $n$ bits used to
represent that number, where $x_1$ is the most significant bit. 

We will also need to encode negative numbers in binary, and the most
expedient way of doing this for our construction is to write an $n$ bit number
$x$
using $2n$ bits $x^+_1$ through $x^+_n$, and $x^-_1$ through $x^-_n$, where:
\begin{itemize}
\item If $x > 0$ then $x^+_1$ through $x^+_n$ hold a binary encoding of $x$ and
$x^-_i = 0$ for all $i$. 
\item If $x < 0$ then $x^-_1$ through $x^-_n$ hold a binary encoding of $-x$
and
$x^+_i = 0$ for all $i$. 
\item If $x = 0$ then $x^+_i = x^-_i = 0$ for all $i$.
\end{itemize}
Throughout this section, we will use the binary variable $x$ and the set of bits $\{x^j_i\}_{1 \le i \le n, j \in \{-,+\}}$ interchangeably to refer to the value $x$.

When we use variables in a linear circuit to hold a binary variable $x$, it is
possible that some bits used in the encoding of $x$ may not be in $\{0, 1\}$.
If it is the case that $x^j_i \in \{0, 1\}$ for all $i$ and $j$ then we say
that $x$ is a \emph{correct} binary encoding.

\subsubsection{The \decode circuit: Decoding a binary variable}

The first circuit that we build is the standard \emph{binary decoder}
circuit, which takes a correct binary encoding of a variable, and outputs a
continuous variable with that value.
Given a binary variable $x$ we define
\begin{align*}
\decode(x) & = \phantom{-} x^+_1 \cdot \frac{1}{2} + x^+_2 \cdot \frac{1}{4} + \dots +
x^+_n \cdot \frac{1}{2^n} \\
& \phantom{=} -x^-_1 \cdot \frac{1}{2} - x^-_2 \cdot \frac{1}{4} - \dots -
x^-_n \cdot \frac{1}{2^n} \\
\end{align*}
This operation simply follows the definition of the binary encoding that we are
using, and so if $x$ is a correct binary variable, then the circuit will output
a continuous variable holding $x$.

Recall that $\decode^\pi$ is the version of $\decode$ that is evaluated under
the perturbation vector~$\pi$. 
Since $\decode$ does not use $\min$, $\max$ or truncation, 
we have that $\decode^\pi(x) = \decode(x)$ for all $x$ and $\pi$.

\subsubsection{The \extract circuit: Extracting bits from a continuous variable}

Given a continuous variable $x \in [0, 1]$, an integer $n \ge 1$, and a
positive rational constant $L > 0$, the $\extract(x, n, L)$ circuit will output
a binary variable $b$ consisting of bits $b^+_1$, $b^+_2$, \dots, $b^+_n$, which are
the first $n$ bits of $x$. Since a
linear circuit must compute a continuous function, it is not possible to
correctly extract bits for every value of $x \in [0, 1]$. Thus, our circuit
will fail for some inputs, and the constant $L$ will control the values for
which these failures occur. 

We implement \extract in the following way. We fix $x_0 = x$, and we define
\begin{align*}
b^+_i &= \trunc_{[0, 1]} \left( \left( x_{i-1} - \frac{1}{2^i} \right) \cdot
\frac{1}{L} \right), \\
x_i &= x_{i-1} - \frac{b^+_i}{2^i}.
\end{align*}
Since we only require \extract to work for positive inputs, we can simply set
$b^-_i = 0$ for all $i$.

It is relatively straightforward to show that if $x$ does not lie in a
\emph{bad region}
$\left[\frac{k}{2^n}, \frac{k}{2^n} + L\right]$ for any integer $k \in \{0, 1, \dots,
2^n - 1\}$, then this circuit will correctly extract the first $n$ bits of $x$.
This is because when $x$ is not in a bad region, the value inside the truncation
gate will either be less than or equal to $0$, or greater than or equal to $1$,
and thus the value will be truncated to either $0$ or $1$, meaning that the bit
is correctly decoded. 

On the other hand, if $x$ does lie in a bad region, then the value inside
the truncation gate may lie in the range $(0, 1)$, meaning that the bit $b^+_i$ will
not be correctly decoded, and all bits $b^+_j$ with $j > i$ may also have
incorrect values. 

Recall that $\extract^\pi$ is the $\extract$ circuit evaluated under the
perturbation vector
$\pi$. The following lemma shows that the circuit still functions correctly
when evaluated with perturbations, but the bad regions are slightly
expanded. 

\begin{lemma}
\label{lem:extract}
If $x$ does not lie in 
\begin{equation*}
\left[\frac{k}{2^n} - \frac{\max_i |\pi_i|}{L}, \;\; \frac{k}{2^n} + L +
\frac{\max_i |\pi_i|}{L} \right] 
\end{equation*}
for any integer $k \in \{0, 1, \dots, 2^n-1\}$, then $\extract^\pi(x,
n, L)$ will correctly extract the leading~$n$ bits of $x$. 
\end{lemma}
\begin{proof}
We will inductively prove for each $i$ that each $b^+_j$ with $j < i$ correctly
contains bit $j$ of $x$, and that  $x_i = x - \sum_{j = 1}^{i}
\frac{b^+_i}{2^i}$. The base case, when $i = 0$, is trivial.

For the inductive step, 
if $\pi_i$ is the perturbation associated with the truncation operation used to
define $b^+_i$, then we have
\begin{equation*}
b^+_i = \trunc \left( \left( x_{i-1} - \frac{1}{2^i} \right)
\cdot \frac{1}{L} + \pi_i \right).
\end{equation*}
Hence, if $x_{i-1} - 1/2^i \ge L - \pi_i/L$ then we will have $b^+_i =
1$, while if $x_{i-1} - 1/2^i \le - \pi_i/L$ then we will have $b^+_i =
0$, and we note that one of these two cases must be true since $x$ is not in a
bad region.
Thus, $b^+_i$ is decoded correctly.
Since $x_{i-1} = x - \sum_{j = 1}^{i-1} \frac{b^+_j}{2^j}$, and
since $x_i = x_{i-1} - \frac{b^+_i}{2^i}$, we have
that $x_i$ is computed correctly as well.
\end{proof}

\subsubsection{Evaluating a Boolean circuit}

A \emph{Boolean circuit} is defined by a tuple $(V, X, Y, G)$, where $V$ is a
set of variables, $X \subset V$ is a set of \emph{input variables}, and $Y
\subset V$ is a set of \emph{output variables}. Each variable can hold a value
from the set $\{0, 1\}$. We can therefore interpret the set $X = \{x_1, x_2,
\dots, x_k\}$ as a vector $x \in \{0, 1\}^k$, and the set $Y = \{y_1, y_2,
\dots, y_l\}$ as a vector $y \in \{0, 1\}^l$. The set $G$ contains gates of the
following form.
\begin{itemize}
\item \textbf{Not:} $v_1 = 1 - v_2$, where $v_1, v_2 \in V$.
\item \textbf{And:} $v_1 = 1$ if $v_2 = v_3 = 1$, and $v_1 = 0$ otherwise,
where $v_1, v_2, v_3 \in V$.
\end{itemize}
We insist that every non-input variable is the output of exactly one gate.
We also require that the circuit be acyclic in the sense that one can
topologically order the set $V = \{v_1, v_2, \dots, v_n\}$ such that 
the gate that outputs a value to variable $v_i$ only takes inputs from
variables
$v_j$ with $j < i$. 

Therefore each Boolean circuit computes a function $F : \{0, 1\}^k \rightarrow \{0, 1\}^l$, where the value $y = F(x)$ is the result of fixing $x$ at the input variables of the circuit, and then evaluating the gates to obtain the output value $y$. 

We can simulate a Boolean circuit via a linear circuit in the following way.
For each not-gate $v_i$ with input $v_j$ we set
\begin{equation*}
v_i = 1 - v_j
\end{equation*}
in the linear circuit. For each and-gate $v_i$ with inputs $v_j$ and $v_k$ we
set
\begin{equation*}
v_i = \trunc_{[0, 1]}(4 \cdot (v_j + v_k - 1.5)).
\end{equation*}
We use $\eval(G, x)$ to refer to a linear circuit that simulates $G$ in this
way on inputs given by the binary variable $x$. 

Recall that $\eval^\pi(G, x)$ refers to the version of \eval that is perturbed
by the vector $\pi$. The following lemma states that, if perturbations in $\pi$
are in the range $(-1, 1)$, then $\eval^\pi(G, x)$ correctly evaluates the
Boolean circuit with no perturbations appearing in the output variables. When
we prove \cref{lem:circuit-all-gates}, we will set $\delta$ to be much smaller
than $1$, meaning that we will have $\pi_i \in (-1, 1)$, so the lemma implies
that from now on we can assume that all Boolean circuits we evaluate will be
correctly simulated. 

\begin{lemma}
\label{lem:eval}
Let $G$ be a Boolean circuit with input variables $x_1$ through $x_k$ and
output variables $y_1$ through $y_l$. If $\pi \in (-1, 1)$, each $x_i \in
\{0, 1\}$, and we compute $y_1$
through $y_l$ using the linear circuit $\eval^\pi(G, x)$, then each $y_i$ will
lie in $\{0, 1\}$, and will hold the value that would be computed by $G$ on
input $x$. 
\end{lemma}
\begin{proof}
We will prove inductively that each gate in the Boolean circuit is computed
correctly. 

For each not-gate $v_i$ with input $v_j$, the linear circuit directly computes
$v_i = 1 - v_j$. Since by the inductive hypothesis we have that $v_j \in \{0,
1\}$ is the correct value for gate $v_j$, it is clear that $v_i$ will be
computed correctly. 

For an and-gate $v_i$ with inputs $v_j$ and $v_k$, the inductive hypothesis
implies that $v_j, v_k \in \{0, 1\}$ and that both are computed correctly. If 
$\pi_i$ is the perturbation that is associated with the truncation operation
for $v_i$ then we have 
\begin{equation*}
v_i = \trunc_{[0, 1]}(4 \cdot (v_j + v_k - 1.5) + \pi_i) 
\end{equation*}
If
at least one of the two inputs is equal to 0 then we have
\begin{align*}
4 \cdot (v_j + v_k - 1.5) + \pi_i &\le -2 + \pi_i  \\ &
< -1,
\end{align*}
where the last line uses the fact that $\pi_i < 1$. Hence, the value
will be truncated to $0$, so we will have $v_i = 0$. 

On the other hand, if $v_j = v_k = 1$ then we have
\begin{align*}
4 \cdot (v_j + v_k - 1.5) + \pi_i &= 2 + \pi_i  \\ &
> 1,
\end{align*}
where the last line uses the fact that $\pi_i > -1$. Hence the value will be
truncated to $1$, so we will have $v_i = 1$.

Therefore in both cases we have that $v_i \in \{0, 1\}$ and that $v_i$ is
correctly computed. 
\end{proof}

Given an expression $e$ over binary variables, we write $E(e)$ to be the linear
circuit that is defined in the following way. First we construct a Boolean circuit $G$
that evaluates $e$, then we evaluate that circuit using \eval.

\subsection{Computing a mesa}
\label{sec:mesa-circuits}

We now build circuits that will allow us to output the five affine pieces of a
mesa. 

\subsubsection{The \bm circuit: multiplying by a single bit}

As mentioned in the technical overview, if we are to output an affine function
with a particular gradient, we need to be able to multiply two variables
together. Although it is, by definition, impossible two multiply two continuous
variables in a linear circuit, we
show that we can multiply a continuous variable with a variable encoded in
binary. 

We break this operation down into two steps. In this section we implement the
$\bm(x, b)$ circuit, which multiplies a continuous variable $x \in [-2, 2]$ with a bit
$b \in \{0, 1\}$. 
We define the circuit \bm in the following way.
\begin{equation*}
\bm(x, b) = 4 \cdot (1-b) + \max\bigl(-4, x-8 \cdot (1-b)\bigr).
\end{equation*}

Recall that $\bm^\pi$ is the $\bm$ circuit perturbed by the vector $\pi$, and
also recall that we will choose $\delta$ such that $\pi \in (-1, 1)$. 
The following lemma shows that $\bm^\pi$ is always correct when $b=0$, and
when $b=1$, the circuit outputs the correct answer (additively) perturbed by a value that 
is at most $\max_i |\pi_i|$.

\begin{lemma}
\label{lem:bm}
If $x \in [-2, 2]$ and each $\pi_i \in (-1, 1)$ then we have the following.
\begin{itemize}
\item If $b = 0$ then $\bm^\pi(x, b) = 0$. 
\item If $b = 1$ then $\bm^\pi(x, b) = x + \sigmap{bm}(\pi)$, where
$\sigmap{bm}$ is a function satisfying $|\sigmap{bm}(\pi)| \le \max_i
|\pi_i|$. 
\end{itemize}
\end{lemma}
\begin{proof}
Throughout, we will assume that $\pi_i$ is the perturbation associated with the
max gate in the \bm circuit. 

If $b=0$ then we have
\begin{align*}
\bm^\pi(x, 0) &= 4 \cdot (1-b) + \max\bigl(-4, x-8 \cdot (1-b) + \pi_i \bigr) \\
&= 4 + \max\bigl(-4, x-8 + \pi_i \bigr) \\
&=0,
\end{align*}
where the second line uses the fact that $x \le 2$ and $\pi_i < 1$, meaning that
the $-4$ term is the maximum. 

If $b=1$ then we have
\begin{align*}
\bm^\pi(x, 1) &= 4 \cdot (1-b) + \max\bigl(-4, x-8 \cdot (1-b) + \pi_i \bigr) \\
&= \max\bigl(-4, x + \pi_i \bigr) \\
&= x + \pi_i,
\end{align*}
where the second line uses the fact that $x \ge -2$ and $\pi_i > -1$,
meaning that the $x + \pi_i$ term is the maximum. 
So setting $\sigmap{bm}(\pi) = \pi_i$ completes the proof. 
\end{proof}

\subsubsection{The \ctb circuit: multiplying by a binary variable}

With the \bm circuit at hand, it is now straightforward to implement the
\ctb(x, y) circuit, which multiplies a continuous variable $x \in [-2, 2]$ with
a binary variable $y$. Specifically, we just multiply each bit individually,
and then sum the results. We use the following definition.
\begin{align*}
\ctb(x, y) &=  \phantom{-} \sum_{i = 1}^n \left( \frac{1}{2^{i}} \cdot \bm(x, y^+_i) \right) \\
&\phantom{=} - \sum_{i = 1}^n \left( \frac{1}{2^{i}} \cdot \bm(x, y^-_i) \right).
\end{align*}

Recall that $\ctb^\pi$ is the $\ctb$ circuit perturbed by the vector $\pi$. The
following lemma shows that 
$\ctb^\pi$ introduces a perturbation of magnitude at most $\max_i |\pi_i|$.

\begin{lemma}
\label{lem:ctb}
If $x \in [-2, 2]$, each $y^j_i \in \{0, 1\}$, and each $\pi_i \in (-1, 1)$, then we have
\begin{equation*}
\ctb^\pi(x, y_1, \dots, y_n) = x \cdot y + \sigmap{ctb}(\pi)
\end{equation*}
where $\sigmap{ctb}$ is a function satisfying $|\sigmap{ctb}(\pi)| \le \max_i
|\pi_i|$.
\end{lemma}
\begin{proof}
We will prove this for the case where $y \ge 0$. The case where $y \le 0$ is
entirely symmetric. 

Let $\sigmap{bm}^i$ be the function associated with the $\bm(x, y_i)$
operation.  Observe that $x$, each $y^+_i$, and each $\pi_i$ all satisfy
the preconditions of \cref{lem:bm}. Thus we have
\begin{align*}
\ctb^\pi(x, y_1, \dots, y_n) &= 
\sum_{i \; : \; y^+_i = 1} \left( \frac{1}{2^{i}} ( x + \sigmap{bm}^i(\pi) ) \right) - 0
\\
&= x \cdot y + \sum_{i : y_i = 1} \left( \frac{1}{2^{i}}  \cdot
\sigmap{bm}^i(\pi) \right), 
\end{align*}
where the first line uses the fact that $y^-_i = 0$ for all $i$, so the
negative summands introduce no extra perturbations. 
We set $\sigmap{ctb}(\pi) = \sum_{i : y_i = 1} \frac{1}{2^{i}}  \cdot \sigmap{bm}^i(\pi)
$. Since for each $i$ we have $|\sigmap{bm}^i(\pi)| \le \max_j \pi_j$, and since $\sum_{i : y_i
= 1} \frac{1}{2^{i}} \le 1$, we have that $|\sigmap{ctb}(\pi)| \le \max_j \pi_j$.
\end{proof}

\subsubsection{\affine: Computing an affine function}

We now build a linear circuit that outputs an affine function that has value
$a$ at point $p = (p^1, p^2) \in \reals^2$ and has gradient $g = (g^1, g^2)$, where each of these values will be presented in binary to the circuit.
Specifically, given a continuous variable $x \in [0, 1]^2$  we define the circuit 
$\affine(x, p^1, p^2, a, g^1, g^2)$ which will output the value of the affine
function defined by $p$, $a$, and $g$ at the point $x$. 
We define $\affine(x, p^1, p^2, a, g^1, g^2)$ to be
\begin{align*}
&\phantom{=} \phantom{+}\; \ctb(x_1 - \decode(p^1_1, \dots, p^1_n), \;\; g^1_1,
\dots, g^1_n) \\
   & \;\;\;  + \ctb(x_2 - \decode(p^2_1, \dots, p^2_n), \; \; g^2_1, \dots, g^2_n) \\
   & \;\;\; + \decode(a_1, \dots, a_n).
\end{align*}

Recall that $\affine^\pi$ is the $\affine$ circuit perturbed by the vector
$\pi$. The following lemma states that $\affine^\pi$ correctly outputs an
affine function defined by $p$, $a$, and $g$, while introducing an additive
perturbation of magnitude at most $2 \cdot \max \pi_i$. In particular, note
that the gradient of the outputted function is exactly $g$. 

\begin{lemma}
\label{lem:affine}
If $x_1 - p^1, x_2 - p^2 \in [-2, 2]$, if each $p^1_i, p^2_i, g^1_i, g^2_i,
a_i \in \{0, 1\}$, and if each $\pi_i \in (-1, 1)$, then 
\begin{equation*}
\affine^\pi(x_1, x_2, p^1, p^2, a, g^1, g^2) = (x_1 - p^1) \cdot g^1 + (x_2 -
p^2) \cdot g^2 + a + \sigmap{aff}(\pi)
\end{equation*}
where $\sigmap{aff}$ is a function satisfying $|\sigmap{aff}(\pi)| \le 2 \cdot \max |\pi_i|$. 
\end{lemma}
\begin{proof}
Let $\sigmap{ctb}^1$ and 
$\sigmap{ctb}^1$ be the perturbation functions associated with the first and
second \ctb operations, respectively. 
Since all of our parameters meet the preconditions of \cref{lem:ctb}, we
obtain
\begin{equation*}
\affine^\pi(x_1, x_2, p^1, p^2, a, g^1, g^2) =  (x_1 - p^1) \cdot g^1 + (x_2 -
p^2) \cdot g^2 + a + \sigmap{ctb}^1(\pi) + \sigmap{ctb}^2(\pi).
\end{equation*}
We set 
$\sigmap{aff}(\pi) = \sigmap{ctb}^1(\pi) + \sigmap{ctb}^1(\pi)$. Since 
$|\sigmap{ctb}^j(\pi)| \le \max_i |\pi_i|$ for all $j$, we have 
$|\sigmap{aff}(\pi)| \le 2 \cdot \max_i |\pi_i|$.
\end{proof}

We should remark at this point that it is important that we keep track of
perturbations in the way that we have been doing so far. To prove
\cref{lem:circuit-all-gates}, it is not enough to just show that
$\affine^\pi$ outputs the affine function perturbed by at most 
$2 \cdot \max |\pi_i|$, because this formulation could allow different
perturbations to occur at different values of $x$, which would then affect the
gradient of the affine function. Instead, our formulation encapsulates the
perturbation in the function $\sigmap{aff}(\pi)$, which makes it clear that the
perturbation \emph{does not} depend on $x$. So we get an affine function
with exactly the gradient that we asked for that has been additively
perturbed everywhere by $\sigmap{aff}(\pi)$.

\subsubsection{Constructing the mesa pieces}

We now build circuits that output the five affine functions defined
in~\cref{sec:mesa} that are used to define a mesa.

Recall that $E(e)$ denotes the evaluation of an expression $e$ over binary
variables using the \eval circuit. 
We will also make use of the
$\ell$ and $\Gamma$ constants used in the statement of
\cref{lem:circuit-all-gates}, and we note that both values can be represented
in binary using polynomially many bits.

If we are given continuous variables $x = (x_1, x_2)$ and
binary encodings of $p = (p_1, p_2)$, $a$, and $g = (g_1, g_2)$ we can build
the five affine pieces given in \cref{sec:mesa} as follows. 
\begin{align*}
P_c(x, p, a, g) & = \affine(x_1, x_2, p_1, p_2, a, g_1, g_2) \\
P_r(x, p, a, g) & = \affine(x_1, x_2, E(p_1 + \ell/2), E(p_2 + \ell/2),
E((g_1+g_2)\ell/2 + a), E(-\Gamma + g_1), g_2) \\
P_t(x, p, a, g) & = \affine(x_1, x_2, E(p_1+\ell/2), E(p_2+\ell/2), E((g_1 + g_2)\ell/2 + a), g_1, E(-\Gamma + g_2)) \\
P_l(x, p, a, g) & = \affine(x_1, x_2, E(p_1 - \ell/2), E(p_2 - \ell/2), E(-(g_1 + g_2)\ell/2 + a), E(\Gamma + g_1), g_2) \\
P_b(x, p, a, g) & = \affine(x_1, x_2, E(p_1 - \ell/2), E(p_2 - \ell/2), E(-(g_1 + g_2)\ell/2 + a), g_1, E(\Gamma + g_2)) 
\end{align*}
These circuits simply implement the definitions given in~\cref{sec:mesa}.

The mesa itself is the minimum of these five functions, but we will not compute
that minimum at this stage. This is because we will first apply the
averaging trick to each of the pieces and then take the minimum. 
The reasons for this will be discussed later. 

Recall that for each circuit $P_i$ we have that $P^\pi_i$ is the version of
that circuit perturbed by the vector $\pi$. Since each piece is implemented by
a single call to \affine, we get that $P^\pi_i$ correctly outputs the
corresponding piece of the mesa with additive perturbation whose magnitude is
at most $2 \cdot \max_j |\pi_j|$. 

\begin{lemma}
\label{lem:p}
If $\pi \in (-1, 1)$, $x_1 - p_1 \pm \ell/2 \in [-2, 2]$, and $x_2 - p_2 \pm
\ell/2 \in [-2, 2]$ then for each $i \in \{c, r, t, l, b\}$ we have $P^\pi_i(x)
= P_i(x) + \sigmap{p}^i(\pi)$ where
$|\sigmap{p}^i(\pi) | \le 2 \cdot \max_j |\pi_j|$. 
\end{lemma}
\begin{proof}
First note that \cref{lem:eval} guarantees that the calls to $E$ will
be correctly computed. Hence, we can apply \cref{lem:affine} to argue that
each $P_i$ will be computed to be the correct affine function perturbed by
$\sigmap{aff}$. 
Since 
$|\sigmap{aff}(\pi) | \le 2 \cdot \max_j |\pi_j|$, this completes the proof.
\end{proof}

\subsection{The construction}
\label{sec:final-circuit}

We now proceed to construct the circuit $\tf$, which for each $x \in [0, 1 -
1/2^n]^2$ will output
\begin{equation*}
\tf(x) = \max_{p \in \tG} M(x; p, \ta(p), \tg(p)).
\end{equation*}

As mentioned in the technical overview, we will divide the points in $\tG$ into
four sets, $S_1$ through $S_4$, and we will build circuits $g_1$ through $g_4$
that output
\begin{equation*}
g_i =  \max_{p \in S_i} M(x; p, \ta(p), \tg(p)).
\end{equation*}
The sets themselves are shown on the right-hand side of \cref{fig:grid}.
They simply divide the points into sub-grids of double the width. We define
this formally using the following points. 
\begin{align*}
a^1 &= (0, 0) &  a^2 &= (0, 1) \\
a^3 &= (1, 0) &  a^4 &= (1, 1)
\end{align*}
Then we define 
\begin{equation*}
S_i = \{ y \in G \; : \; y = a^i + (2a/2^n, 2b/2^n)
\text{ for some $a, b \in \mathbb{N}$} \}.
\end{equation*}
Once we have built the $g_i$ circuits, we will then
build $\tf$ in the following way.
\begin{equation*}
\tf(x) = \max(g_1(x), \max(g_2(x), \max(g_3(x), g_4(x)))).
\end{equation*}

\paragraph{\bf The definition of the $g_i$ circuits.}

To complete the construction we must specify how the $g_i$ circuits are
constructed. This circuit will make use of the averaging trick, and we fix $k =
12$ to be the number of samples that we will average over.

The first step for computing $g_i$ is to extract the bits of $x$. We do this
using the \extract circuit, and we carefully select $k$ points close to $x$ and
attempt to decode each of them. 
Formally, 
for each $i \in \{1, 2, 3, 4\}$ and each $j \in \{1, \dots, k\}$ we 
extract the binary variable $y^{i, j} = (y^{i,j}_1, y^{i,j}_2)$ in the
following way. 
\begin{equation*}
y^{i,j}_l = \extract\left(a^i_l + x_l + \frac{9/8}{2^n} - \frac{j/8k}{2^n}, \; n - 1, \; \frac{1}{24k \cdot 2^n} 
\right)
\end{equation*}
Since we are looking to extract a point from $S_i$, we add $a^i$ to the point,
and extract only the first $n-1$ bits, to ensure
that we get a binary encoding of a point from $S_i$. The other parameters here are chosen so that
there can be at most two values of $j$ such that $y^{i, j}$ fails to be decoded
correctly. 

The next step is to construct the five pieces of the mesa defined by $y^{i,j}$,
which we do in the following operation. 
Recall that $E(e)$ denotes the evaluation of an expression $e$ over binary
variables using the \eval circuit. 
Then, for each $d \in \{c, r, t, l, b\}$ 
we define
\begin{equation*}
p^{i, j}_d(x) = \min\Bigl(1, \; \; P_d(x, y^{i,j}, E(\ta(y^{i,j})),
E(\tg(y^{i,j}))) \Bigr).
\end{equation*}
Here the minimum with 1 is used to prevent outputs from bad decodes from being
excessively large. This cannot affect the final output because we know that $\tf(x)
\le 2/3$, and so any piece that is in the output will not be affected by the
minimum operation.

Next we average over the 
$p^{i, j}_d(x)$ values to obtain a single averaged piece.
Formally, for each $d \in \{c, r, t, l, b\}$ we define
\begin{equation*}
p^i_d(x) = \sum_{j=1}^k \frac{1}{k} \cdot p^{i, j}_d(x).
\end{equation*}
Finally we output the mesa by taking the minimum over the five averaged pieces. 
\begin{equation*}
g_i(x) = \min \biggl( 
p^i_c(x), \; \min \Bigl(
p^i_r(x), \; \min \bigl(
p^i_t(x), \; \min(
p^i_l(x),\;  
p^i_b(x))\bigr)\Bigr)\biggr)
\end{equation*}
It is worth remarking that we must first average each piece and then take the
minimum over averaged pieces, as we have done here. The alternative approach of
first constructing $k$ different mesas, and then averaging over them, does not work. This is because each mesa would be
perturbed slightly differently, meaning that the locations at which the pieces
of the mesa meet would move slightly. So when we then average over those mesas
we could end up, for example, averaging the center piece of one mesa with
the right piece of another, which would create a new undesired gradient.

\paragraph{\bf Correctness.}

We now prove that this construction satisfies \cref{lem:circuit-all-gates}. We
begin by proving that the $g_i$ functions are correct. 
Recall that $g^\pi_i$ is the circuit $g_i$ perturbed by the vector $\pi$. 
We will prove two properties.
\begin{itemize}
\item If $\| x - y \|_\infty \le 3/4 \cdot 1/2^n$ for some $y \in S_i$, meaning
that we are suitably close to a point in $S_i$, then $g^\pi_i$ outputs the mesa
defined by $y$, in which all pieces have been additively perturbed by a small
amount. This is because all of the $k$ points that we decoded gave us binary
encodings of $y$, since $y$ is
not close to any of the bad regions of the \extract operation. Thus each
$p^i_d$ averages over $k$ pieces that all have identical gradients, but each of
which has been perturbed slightly differently. The averaging operation
therefore results in a mesa whose pieces have gradients that are exactly correct, but
where each has been perturbed by the average of the perturbations from the
$k$ different $p^{i,j}_d$ operations.

\item If we are not close to a point in $S_i$, then $g_i$ outputs a value
strictly less than 1/4. In this case the \extract operations may decode
different points in $S_i$, however since $\| x - y \|_\infty > 3/4 \cdot 1/2^n$
for all points $y \in S_i$, any mesa from $S_i$ that we evaluate at $x$ will output a value
that is less than or equal to $0$. We may also have at most two points that are
incorrectly decoded, and for these points we bound the value at $1$ using the
$\min$ operation in $p^{i,j}_d$. Since
$k=12$ we get that $g_i(x) \le 10/12 \cdot 0 + 2/12 \cdot 1 = 1/6$. While
perturbations may increase this value, we show that if each element of $\pi$ is
suitably small, then $g^\pi_i(x) < 1/4$. 
\end{itemize}
The two points above are shown formally in the following lemma.

\begin{lemma}
\label{lem:g}
Suppose that $\max |\pi_i| \le (\frac{1}{24k \cdot 2^n})^2/2$. 
For each $i \in \{1, 2, 3, 4\}$ we have
\begin{itemize}
\item If $\| x - y \|_\infty \le 3/4 \cdot 1/2^n$ for some $y \in S_i$, then
\begin{equation*}
g^\pi_i(x) = M(x, y, A, \tg(y)) 
\end{equation*}
where $A = (
\ta(y) + \sigmap{g}^c(\pi), 
\ta(y) + \sigmap{g}^r(\pi), 
\ta(y) + \sigmap{g}^t(\pi), 
\ta(y) + \sigmap{g}^l(\pi), 
\ta(y) + \sigmap{g}^b(\pi))$
and each $|\sigmap{g}^j(\pi)| \le 7 \cdot \max_l |\pi_l|$.
\item Otherwise $g^\pi_i(x) < 1/4$.
\end{itemize}
\end{lemma}
\begin{proof}
We fix $t_j = \frac{9/8}{2^n} - \frac{j/8k}{2^n}$ and $L = \frac{1}{24 k
\cdot 2^n}$ to be the arguments to the \extract circuit used to extract
$y^{i,j}$.
By \cref{lem:extract} we have that an \extract operation for 
$x_1$ will fail only if it lies in the region
\begin{equation*}
\left[a^i_1 + \frac{m}{2^{n-1}} - t_j - \frac{\max_i |\pi_i|}{L},
\;\; \frac{m}{2^{n-1}} - t_j + L + \frac{\max_i |\pi_i|}{L} \right] 
\end{equation*}
for some integer $m \in \{0, 1, \dots, 2^{n-1}-1\}$. By assumption we have that 
$\max_l |\pi_l| \le L^2/2$, so we can define the following \emph{bad regions}.
\begin{align*}
R^1_j &= \left[a^i_1 + \frac{m}{2^{n-1}} -t_j - L/2 \;\; \frac{m }{2^{n-1}} -
t_j + 3L/2 \right], \\
R^2_j &= \left[a^i_2 + \frac{m}{2^{n-1}} -t_j - L/2 \;\; \frac{m }{2^{n-1}} -
t_j + 3L/2 \right]. 
\end{align*}
So the \extract operations used in $g^\pi_i$ can fail only if $x_1 \in
R^1_j$ for some $m$, or if $x_2 \in R^2_j$ for some $m$.

For each $l \in \{1, 2\}$, we can show that the set of bad regions
$\{R^l_j\}_{j=1,\dots,k}$ do not overlap with the following chain of inequalities.
\begin{align*}
\frac{m}{2^n} -t_j + 3L/2 
&= a^i_l + \frac{m}{2^{n-1}} - \frac{9/8}{2^n} + \frac{j/(8k)}{2^n} + \frac{3/(2 \cdot 24k)}{2^n} \\
&= a^i_l + \frac{m}{2^{n-1}} - \frac{9/8}{2^n} + \frac{6j/(48k) + 3/(48k)}{2^n}  \\
&< a^i_l + \frac{m}{2^{n-1}} - \frac{9/8}{2^n} + \frac{6(j+1)/(48k)}{2^n} - \frac{1/(48k)}{2^n} \\
&= a^i_l + \frac{m}{2^{n-1}} -t_j - L/2.
\end{align*}
Hence both $x_1$ and $x_2$ can each lie in at most one bad region.

Since all bad regions lie
between $a^i_l + \frac{m}{2^{n-1}} - \frac{9/8}{2^n}$ and 
$a^i_l + \frac{m}{2^{n-1}} - \frac{7/8}{2^n}$, if either $x_1$ or $x_2$ lies in
a bad region, then $\| x - y \|_\infty \ge 3/4 \cdot 1/2^n$ for all $y \in
S_i$.

We can now prove the lemma via a case analysis. 
\begin{itemize}
\item If there exists a $y \in S_i$ with $\| x -  y \|_\infty \le 3/4 \cdot
1/2^n$ then note that all of the \extract operations used in $g^\pi_i$
circuits will each return $y$, since we would have to cross a bad region to
output any other point, and there are no bad regions within distance $3/4 \cdot
1/2^n$ of $y$.

Thus for each $j$ and $d \in \{c, r, t, l, b\}$
we have that the $P_d$ operation used in $p^{i,j}_d(x)$ will output $P_d(x, y,
\ta(y), \tg(y)) +
\sigmap{p}^d(\pi)$,
where 
$|\sigmap{p}^d(\pi) | \le 2 \cdot \max_j |\pi_j|$ by \cref{lem:p}. 
This value is then perturbed by the $\min$ gate used in $p^{i,j}_d$, which
introduces a perturbation that we will call $\pi^{i,j}_d$.
Each $\min$ gate used in the definition of $g_i(x)$ can add
an additional perturbation of magnitude at most $\max_{l} |\pi_l|$, and at most
four such perturbations can be added when computing $g_i$. So we have that 
\begin{equation*}
g^\pi_i(x) = \min_{d \in \{c, r, t, l, b\}} \Bigl( \min \bigl(1, 
P_d(x, y, \ta(y), \tg(y)) + \sigmap{p}^d(\pi) + \pi^{i,j}_d \bigr) \Bigr) +
\sigmap{mins}(\pi),
\end{equation*}
where 
$|\sigmap{mins}(\pi)| \le 4 \max_l |\pi_l|$.

We must show that the min with $1$ used in each term above does not affect the
output. For all $d$ we have
\begin{align*}
2/3 &\ge f(x) \\
&\ge M(x; y, \ta(y), \tg(y)) \\
& \ge P_d(x, y, \ta(y), \tg(y)). 
\end{align*} 
Observe that 
\begin{align*}
\sigmap{p}^d(\pi) + \pi^{i,j}_d &\le 3 \cdot \max_l |\pi_l| \\
&\le \left(\frac{1}{24k \cdot 2^n}\right)^2/2 \\
&< 1/3, 
\end{align*}
where the final line holds for all $n \ge 1$. 
Hence, the minimum never affects the output of $g^\pi_i(x)$, and so we have
\begin{equation*}
g^\pi_i(x) = \min_{d \in \{c, r, t, l, b\}} \Bigl( P_d(x, y, \ta(y), \tg(y)) +
\sigmap{p}^d(\pi) + \pi^{i,j}_d \Bigr) + \sigmap{mins}(\pi).
\end{equation*}
When we account for all of the perturbations above, we have that each piece is
additively perturbed by at most $7 \cdot \max_l |\pi_l|$, as required.

\item On the other hand, if there is no $y \in S_i$ such that 
$\| x -  y \|_\infty \le 3/4 \cdot 1/2^n$ then note that $M(x; y, \ta(y),
\tg(y)) \le 0$ for
all $y \in S_i$. Since $x_1$ and $x_2$ can each fall into at most one bad region, 
there
will be at most two indices $j$ for which the $p^{i,j}_d$ circuits
will receive incorrectly decoded points, 
and for
these circuits we will have $p^{i,j}_d(x) \le 1$ due to the min operation used in
the definition of 
$p^{i,j}_d$. For every other index $j$, the \extract operation will output some
valid point $y \in S_i$. If $\pi^{i,j}_d$ is the perturbation associated with
the $\min$ gate in $p^{i,j}_d$, then we can apply \cref{lem:p} to obtain
\begin{align*}
p^{i,j}_d(x) &\le P_d(x, y, \ta(y), \tg(y)) + \sigmap{p}(\pi) + \pi^{i,j}_d \\
&\le 0 + \sigmap{p}(\pi) + \pi^{i,j}_d\\
&\le 3 \cdot \max_l |\pi_l|,
\end{align*}
where the first inequality uses the fact that 
$\| x -  y \|_\infty > 3/4 \cdot 1/2^n$ which implies that the mesa associated with $y$ has
value less than or equal to $0$ at $x$.

Since $k=12$, the circuit $p^i_d(x)$ averages over 10 values that are at most
$3 \cdot \max_l |\pi_l|$, and 2 values that are at most $1$, 
so we have
\begin{equation*}
p^i_d(x) \le \frac{10}{12} \cdot 3 \cdot \max_l |\pi_l| + \frac{2}{12}
\end{equation*}
As argued in the first case, the $g_i$ circuit introduces an additive perturbation of
magnitude at most 
$4 \cdot \max_l |\pi_l|$.
So we have
\begin{equation*}
g^\pi_i(x) \le 1/6 + 7 \cdot \max_l |\pi_l|.
\end{equation*}
Since $\max_l |\pi_l| \le (\frac{1}{24k \cdot 2^n})^2/2 < 1/84$ for all $n \ge
1$, we
have $g^\pi_i(x) < 1/4$, as required.
\end{itemize}
\end{proof}

The following final lemma now shows that $\tf$ correctly computes the perturbed
mesa function that is desired by \cref{lem:circuit-all-gates}. As mentioned in
the technical overview, any spurious gradients that arise from bad bit decodes
do not appear in the output. This is because the previous lemma assures that
$g_i(x) < 1/4$ whenever $g_i$ fails to decode a point. Since $\tf(x) \ge 1/3$
for all points $x$, there is therefore some other
function $g_j(x) \ge 1/3$ that outputs the mesa that defines $\tf(x)$, and thus
all gradients from $g_i$ do not appear in the final output of the function. The
following lemma proves this, and also shows that if $\pi$ is suitably small,
then the perturbations do not affect this property.

\begin{lemma}
\label{lem:f}
Suppose that $\max |\pi_i| \le (\frac{1}{24k \cdot 2^n})^2/2$. 
For all $x \in [0, 1 - 1/2^n]^2$ we have 
\begin{equation*}
\tf^\pi(x) = M(x, y, A, \tg(y)) 
\end{equation*}
where $A = (
\ta(y) + \sigmap{f}^c(\pi), 
\ta(y) + \sigmap{f}^r(\pi), 
\ta(y) + \sigmap{f}^t(\pi), 
\ta(y) + \sigmap{f}^l(\pi), 
\ta(y) + \sigmap{f}^b(\pi))$
and each $|\sigmap{f}^j(\pi)| \le 11 \cdot \max_l |\pi_l|$.
\end{lemma}
\begin{proof}
If $\pi_a$, $\pi_b$, and $\pi_c$ are the perturbations associated with the
three min gates used to compute $\tf$, then 
we have
\begin{equation}
\label{eqn:fmax}
\tf^\pi(x) = \max(g^\pi_1(x), \max(g^\pi_2(x), \max(g^\pi_3(x), g^\pi_4(x) + \pi_a) + \pi_b) + \pi_c).
\end{equation}
From \cref{lem:g} we have that either $g^\pi_i(x)$ outputs a perturbed mesa 
for some $y \in S_i$ that is within distance $3/4 \cdot 1/2^n$
of $x$, or that $g_i(x) < 1/4$. 

We rule out the case in which $g_i(x) < 1/4$ for
all $i$. To do this, we first recall that by assumption we have $f(x) \ge 1/3$.
Hence, there is a mesa defined by $y$ such that $M(x; y, \ta(y), \tg(y)) \ge 1/3$, so we must
also have $\| x - y \|_\infty \le 3/4 \cdot 1/2^n$.
Therefore, 
\cref{lem:g}
implies that there is an index $i$ and $d \in \{c, r, t, l, b\}$ such that 
\begin{align*}
g_i(x) &\ge M(x; y, \ta(y), \tg(y)) + \sigmap{g}^d(\pi) \\
&\ge M(x; y, \ta(y), \tg(y)) - 7 \cdot \max_i |\pi_i| \\
&> M(x; y, \ta(y), \tg(y)) - 1/12 \\
&\ge 1/4
\end{align*}
where the third line
uses the fact that 
$\max |\pi_i| \le (\frac{1}{24k \cdot 2^n})^2/2 < 1/84$ for all $n$, and the
final line uses the fact that $M(x; y, \ta(y), \tg(y)) \ge 1/3$.

Finally note that the perturbations contributed by $\pi_a$, $\pi_b$, and
$\pi_c$, can perturb the final output by a value of magnitude at most $4 \cdot
\max_i |\pi_i|$. So once we incorporate the perturbations of magnitude at most
$7 \cdot \max_i |\pi_i|$ arising from $\sigmap{g}(\pi)$, each mesa piece can be
additively perturbed by a value of magnitude at most $11 \cdot \max_i |\pi_i|$.
\end{proof}

To complete the proof of \cref{lem:circuit-all-gates} it therefore suffices to pick
$\delta = \min(\delta'/11, (\frac{1}{24k \cdot 2^n})^2/2)$.

\section{Conclusions, Open Problems}
Open problems include the question of whether our result continues to hold for restrictions of QPs such as the bilinear case, where there are no squared terms. \cite{BabR21} show that it is \cls/-complete to find a KKT-point of a \emph{multilinear} degree-five polynomial. They use this multilinearity to show \cls/-completeness also for the problem of finding a mixed equilibrium of degree-five polytensor games, on the way to showing \cls/-completeness for finding a mixed equilibrium of a congestion game. One of the main open problems arising from their work is whether it is \cls/-hard to find a mixed equilibrium of a ``network coordination'' game, i.e., a degree-2 polytensor game. This would follow if our result continued to hold without squared terms. Other special cases of interest include: no linear terms (the \np/-hardness results of \cite{AhmadiZ22-polytope} apply to QPs that have no linear terms, i.e., quadratic forms). Our result also exploits exponential ratios between coefficients, leaving open the question of whether it should hold if coefficients are presented in unary. In connection with their discussion of KKT solutions, \cite{PardalosV92} also ask about the special case where there is a single local minimum (hence the global minimum). This latter problem would have to be treated as a promise problem.

\appendix

\section{From Approximate to Exact Solutions}
\label{app:approx-to-exact}

The following shows that our hardness result also applies to the problem of computing an $\eps$-KKT point of a quadratic polynomial over $[0,1]^2$, where $\eps > 0$ is allowed to be exponentially small. The proof technique for this is standard, see, e.g., \cite{EY10-Nash-FIXP}.

\begin{lemma}
Given an instance $p$ of \qKKT, we can compute in polynomial time an $\eps > 0$ such that from any $\eps$-KKT point of $p$ we can obtain an exact KKT point of $p$ in polynomial time.
\end{lemma}

\begin{proof}
Let $p$ be an instance of \qKKT over $n$ variables. For any sets $I_0, I_1 \subseteq [n]$ with $I_0 \cap I_1 = \emptyset$ we define a corresponding LP$(I_0,I_1)$ over variables $z, x_i$:
\begin{equation*}
\begin{tabular}{rlr}
$\min \quad$ & $z$ & \\
s.t.$\quad$ & $z \geq \frac{\partial p}{\partial x_i} (x)$ & $\forall i \in [n] \setminus I_0$\\
& $z \geq - \frac{\partial p}{\partial x_i} (x) \quad$ & $\forall i \in [n] \setminus I_1$\\
& $x_i = 0$ & $\forall i \in I_0$\\
& $x_i = 1$ & $\forall i \in I_1$\\
& $0 \leq x_i \leq 1$ & $\forall i \in [n]$\\
& $z \geq 0$ &
\end{tabular}
\end{equation*}
Note that this is indeed an LP, and that it is feasible and admits an optimal solution. Since LPs have solutions with bit complexity that is polynomially bounded in their description length, and the description length of LP$(I_0,I_1)$ for all $I_0,I_1$ is bounded by a polynomial in the description length of $p$, we can compute in polynomial time a value $\eps > 0$ such that for all $I_0,I_1$ the optimal value of LP$(I_0,I_1)$ is $0$ or strictly larger than $\eps$.

Now consider any $\eps$-KKT point $x^*$ of $p$. Letting $I_0 := \{i \in [n]: x_i^* = 0\}$ and $I_1 := \{i \in [n]: x_i^* = 1\}$, observe that $(\eps, x^*)$ is feasible for LP$(I_0,I_1)$ with objective function value $\eps$. As a result, the LP$(I_0,I_1)$ must have optimal value $0$. Now, it suffices to solve this LP to obtain an optimal solution $(0,x)$ and note that $x$ must be an exact KKT point of $p$.
\end{proof}

Note that the proof also gives all the ingredients needed to show the existence of rational solutions with polynomially bounded bit complexity. Indeed, since a solution is guaranteed to exist (because $p$ is a continuous function over a compact domain and thus admits a global minimum), there must exist $I_0$ and $I_1$ such that LP$(I_0,I_1)$ has optimal value $0$. Since all these LPs have solutions of polynomially bounded bit complexity, the same also holds for \qKKT. We note that both this proof sketch and the proof of the lemma can be extended to general compact domains of the form $Ax \leq b$.

\section{Proof of Lemma~\ref{lem:perturbation}}
\label{app:proof-perturbation-lemma}

Let $\circuit$ be as in the statement of \cref{lem:perturbation}. Let $x_1, \dots, x_N$ denote the variables computed by the circuit, where $x_1, \dots, x_n$ are the inputs, and $x_N$ is the output.

\paragraph{\bf Preprocessing.}
We begin by noting that we can assume that $\circuit$ only consists of gates of the form $x_i := ax_j+bx_k+c$ and of the form $x_i := \trunc_{[a,b]}(x_j)$. Indeed, given $\circuit$ we can construct $\circuitbar$ that only uses this restricted set of gates as follows:
\begin{itemize}
    \item A gate of type $+$, $-$, $c$, or $\times c$ can be directly simulated by a single affine linear gate of the form $x_i := ax_j+bx_k+c$.
    \item A gate $x_i := \trunc_{[a,b]}(a_ix_j + b_ix_k+c_i)$ can be easily simulated by combining the two types of allowed gates.
    \item A gate $x_i := \min \{x_j,x_k\}$ can be simulated by instead letting
    $$x_i := x_j - \trunc_{[0,3B]}(x_j-x_k)$$
    which can be computed using only the two types of allowed gates. Here $B \geq 1$ is a bound on the value of any variable in the perturbed circuit $\circuit^\sigma$ when $\sigma \in [-1,1]^m$ and the inputs lie in $[0,1]^n$. Formally, for any $j \in [N]$, $B \geq 1$ must satisfy
    $$\max_{\sigma \in [-1,1]^m} \max_{y \in [0,1]^n} f_j^\sigma(y) \in [-B,B]$$
    where $f_j^\sigma(y)$ denotes the value of the $j$th variable of circuit $\circuit^\sigma$ on input $y$. Note that such a bound $B$ can be computed in polynomial time given the description of $\circuit$, see, e.g., \cite{FGHS22}.
    \item A gate $x_i := \max \{x_j,x_k\}$ can be simulated by first replacing it by $x_i := - \min \{-x_j,-x_k\}$, and then simulating the $\min$ gate as shown above.
\end{itemize}
We let $\xbar_1, \dots, \xbar_N$ denote the variables in circuit $\circuitbar$ that correspond to the original variables $x_1, \dots, x_N$ of $\circuit$. Note that $\circuitbar$ will also contain some additional variables that are needed to simulate the original gates. For any perturbation $\sigma$ of $\circuit$, we let $x_1^\sigma, \dots, x_N^\sigma$ denote the variables of circuit $\circuit^\sigma$. We define $\xbar_1^\pi, \dots, \xbar_N^\pi$ similarly for $\circuitbar^\pi$. Finally, note that gates of the form $\xbar_i := a\xbar_j+b\xbar_k+c$ are not perturbed, while gates of the form $\xbar_i := \trunc_{[a,b]}(\xbar_j)$ are perturbed as $\xbar_i^\pi := \trunc_{[a,b]}(\xbar_j^\pi + \pi_i)$.

We can now check that the new circuit $\circuitbar$ satisfies, for any $\delta \leq 1$: given any $\delta$-perturbation $\pi$ of $\circuitbar$, there exists a $\delta$-perturbation $\sigma$ of $\circuit$ such that $x_i^\sigma = \xbar_i^\pi$ whenever the input to the circuit lies in $[0,1]^n$. Indeed, this can be shown by induction by using the following observations
\begin{itemize}
    \item The simulation of a gate of type $+$, $-$, $c$, or $\times c$ does not introduce any gates that are perturbed.
    \item The simulation of a gate of type $x_i := \trunc_{[a,b]}(a_ix_j + b_ix_k+c_i)$ yields that $\xbar_i^\pi = \trunc_{[a,b]}(a_i\xbar_j^\pi + b_i\xbar_k^\pi+c_i + \pi_i)$, and thus $\xbar_i^\pi = x_i^\sigma$, if we set $\sigma_i := \pi_i$.
    \item The simulation of a gate of type $x_i := \min \{x_j,x_k\}$ yields that $\xbar_i^\pi = \xbar_j^\pi - \trunc_{[0,3B]}(\xbar_j^\pi-\xbar_k^\pi + \pi_i)$. We claim that this implies $\xbar_i^\pi = \min \{\xbar_j^\pi, \xbar_k^\pi - \pi_i\}$, and thus $\xbar_i^\pi = \min \{x_j^\sigma,x_k^\sigma + \sigma_i\} = x_i^\sigma$ by setting $\sigma_i := -\pi_i$. Indeed, if $\xbar_j^\pi \leq \xbar_k^\pi - \pi_i$, then $\xbar_i^\pi = \xbar_j^\pi - \trunc_{[0,3B]}(\xbar_j^\pi-\xbar_k^\pi + \pi_i) = \xbar_j^\pi$. On the other hand, if $\xbar_j^\pi \geq \xbar_k^\pi - \pi_i$, then $\xbar_i^\pi = \xbar_j^\pi - \trunc_{[0,3B]}(\xbar_j^\pi-\xbar_k^\pi + \pi_i) = \xbar_j^\pi - (\xbar_j^\pi-\xbar_k^\pi + \pi_i) = \xbar_k^\pi - \pi_i$
    where we used the fact that $\xbar_j^\pi-\xbar_k^\pi + \pi_i = x_j^\sigma-x_k^\sigma + \pi_i \leq 2B + \pi_i \leq 3B$, by construction of~$B$.
    \item Similarly, the simulation of a gate of type $x_i := \max \{x_j,x_k\}$ yields that $\xbar_i^\pi = - \min \{- x_j^\sigma, - x_k^\sigma - \pi_i\} = \max \{x_j^\sigma, x_k^\sigma + \pi_i\} = \max \{x_j^\sigma, x_k^\sigma + \sigma_i\} = x_i^\sigma$ by setting $\sigma_i := \pi_i$.
\end{itemize}

\paragraph{\bf Construction.}
Consider now a circuit $\circuit$ that satisfies the conditions of the lemma, but only consists of $x_i := ax_j+bx_k+c$ and $x_i := \trunc_{[a,b]}(x_j)$ gates. It remains to show that we can construct a circuit $\circuitbar$ that only uses $\trunc_{[0,1]}$ gates and satisfies the lemma. We let $\trunc := \trunc_{[0,1]}$.

We begin by computing in polynomial time a bound $B \geq 2$ such that for any $i \in [N]$, $B$ satisfies
$$\max_{\sigma \in [-1,1]^m} \max_{y \in [0,1]^n} f_j^\sigma(y) \in [-B,B].$$
We also make sure that we have $B \geq a + b + 1$ for every gate $x_i := ax_j+bx_k+c$, and $B \geq b-a$ for every gate $x_i := \trunc_{[a,b]}(x_j)$.

Since all gates in $\circuitbar$ can only output values in $[0,1]$, we will encode values of the original circuit, which lie in $[-B,B]$, inside the interval $[1/4,3/4]$. A value $x \in [-B,B]$ can be encoded inside $[1/4,3/4]$ by using the map
$$\phi: \mathbb{R} \to \mathbb{R}, \quad x \mapsto \frac{1}{2} + \frac{x}{4B}.$$
Note that this transformation maps $[-2B,2B]$ to $[0,1]$. A value $x \in [1/4,3/4]$ can be decoded to $[-B,B]$ by using the map
$$\phi^{-1}: \mathbb{R} \to \mathbb{R}, \quad x \mapsto 4B(x-1/2).$$
The circuit $\circuitbar$ is constructed as follows:
\begin{itemize}
    \item For every input variable $x_i$ of $\circuit$, let $\zbar_i$ denote the corresponding input variable of $\circuitbar$. We compute the variable $\xbar_i$ of $\circuitbar$ by encoding $\zbar_i$ into $[1/4,3/4]$, i.e., by letting $\xbar_i := \trunc(\phi(\zbar_i)) = \trunc(1/2 - \zbar_i/4B)$.
    \item Every gate of the type $x_i := ax_j + bx_k + c$ is replaced by the gate $\xbar_i := \trunc(\phi(a\phi^{-1}(\xbar_j) + b\phi^{-1}(\xbar_k) + c))$. This corresponds to decoding $\xbar_j$ and $\xbar_k$, applying the original gate operation, and then encoding again. Note that the expression inside the truncation operator is affine linear, so the gate is well-defined.
    \item Every gate of the type $x_i := \trunc_{[a,b]}(x_j)$ is replaced by the gate $\zbar_i := \trunc(\psi(\phi^{-1}(\xbar_j)))$ and the gate $\xbar_i := \trunc(\phi(\psi^{-1}(\zbar_i)))$, where $\psi$ is a linear transformation mapping $[a,b]$ to $[0,1]$, i.e.,
    $$\psi: \mathbb{R} \to \mathbb{R}, \quad x \mapsto \frac{x-a}{b-a}.$$
    The first expression corresponds to decoding $\xbar_j$, then using a linear map that maps $[a,b]$ to $[0,1]$, and then truncating to $[0,1]$. The second expression corresponds to taking the inverse linear map from $[0,1]$ to $[a,b]$, and then encoding again.
    \item Special case: for the output gate of $\circuit$, which by assumption is of type $x_N := \trunc_{[0,1]}(x_j)$, we just replace it by $\xbar_N := \trunc(\phi^{-1}(\xbar_j))$. This corresponds to the same construction as above, except that (i) $a = 0$ and $b = 1$, and (ii) we leave the output un-encoded.
\end{itemize}

\paragraph{\bf Correctness.}
Let $K := 4B^N$ and consider any $\delta \leq 1/K$. Let $\pi$ be any perturbation vector for $\circuitbar$ with $|\pi_i| \leq \delta$ for all $i$. We will show that there exists a perturbation vector $\sigma$ for $\circuit$ with $|\sigma_i| \leq \delta K$ such that $f^\sigma(y) = \fbar^\pi(y)$ for all $y \in [0,1]^n$.

We prove below by induction that we can construct a $\delta K$-perturbation $\sigma$ of $\circuit$ such that for all $i \in [N-1]$ there exists some $\eps_i$ with $|\eps_i| \leq 4 B^i \delta$ such that
\begin{equation}\label{eq:perturbation-induction}
\xbar_i^\pi = \phi(x_i^\sigma + \eps_i)
\end{equation}
whenever the input to the circuit lies in $[0,1]^n$.

From this we then obtain the desired statement as follows. Recall that the last gate of $\circuit$ is $x_N := \trunc(x_j)$ and that it is replaced by $\xbar_N := \trunc(\phi^{-1}(\xbar_j))$ in $\circuitbar$. As a result, we have that $\xbar_N^\pi = \trunc(\phi^{-1}(\xbar_j^\pi) + \pi(\xbar_N))$ where $\pi(\xbar_i)$ denotes the entry of $\pi$ which is used to perturb the gate computing $\xbar_i$. Using the fact that $\xbar_j^\pi = \phi(x_j^\sigma + \eps_j)$ by \eqref{eq:perturbation-induction}, we obtain
$$\xbar_N^\pi = \trunc(x_j^\sigma + \eps_j + \pi(\xbar_N)) = \trunc(x_j^\sigma + \sigma(x_N)) = x_N^\sigma$$
by setting $\sigma(x_N) := \eps_j + \pi(\xbar_N)$, which satisfies $|\sigma(x_N)| \leq 4B^{N-1}\delta + \delta \leq 4B^N\delta = \delta K$. As a result we have that $f^\sigma(y) = x_N^\sigma = \xbar_N^\pi = \fbar^\pi(y)$ for all $y \in [0,1]^n$, as desired.

It remains to prove \eqref{eq:perturbation-induction}. It is easy to see that the statement holds for the input variables, i.e., for $i \in [n]$. Indeed, by definition we have $x_i^\sigma = \zbar_i^\pi$ for all input variables, i.e., for all $i \in [n]$. Thus, we can write
$$\xbar_i^\pi = \trunc(\phi(\zbar_i^\pi) + \pi(\xbar_i)) = \trunc(\phi(x_i^\sigma + 4B\pi(\xbar_i))) = \trunc(\phi(x_i^\sigma + \eps_i)) = \phi(x_i^\sigma + \eps_i)$$
by setting $\eps_i := 4B\pi(\xbar_i)$, which satisfies $|\eps_i| \leq 4B \delta \leq 4B^i \delta$, as desired. We also used the fact that $\phi(x_i^\sigma + \eps_i) \in [0,1]$, since $x_i^\sigma + \eps_i \in [-2B,2B]$, because $x_i^\sigma \in [-B,B]$ by construction of $B$, and $|\eps_i| \leq 1 \leq B$ since $\delta \leq 1/K = 1/4B^N$.

Now consider any $i \in [N-1]$ with $i > n$, and assume that the induction hypothesis holds for smaller values of $i$. We handle the following two cases separately:
\begin{itemize}
    \item If the $i$th gate is of the type $x_i := ax_j + bx_k + c$, then by construction we have that
    \begin{equation*}
    \begin{split}
    \xbar_i^\pi &= \trunc(\phi(a\phi^{-1}(\xbar_j^\pi) + b\phi^{-1}(\xbar_k^\pi) + c) + \pi(\xbar_i))\\
    &= \trunc(\phi(ax_j^\sigma + a\eps_j + bx_k^\sigma + b\eps_k + c) + \pi(\xbar_i))\\
    &= \trunc(\phi(x_i^\sigma + a\eps_j + b\eps_k + 4B\pi(\xbar_i))\\
    &= \trunc(\phi(x_i^\sigma + \eps_i)) = \phi(x_i^\sigma + \eps_i)
    \end{split}
    \end{equation*}
    where we set $\eps_i := a\eps_j + b\eps_k + 4B\pi(\xbar_i)$, which satisfies $|\eps_i| \leq (a+b)4B^{i-1}\delta + 4B \delta \leq 4B^i\delta$ by construction of $B$, and where we also used the induction hypothesis, namely $\xbar_j^\pi = \phi(x_j^\sigma + \eps_j)$ and $\xbar_k^\pi = \phi(x_k^\sigma + \eps_k)$. We also used the fact that $\phi(x_i^\sigma + \eps_i) \in [0,1]$, because $x_i^\sigma + \eps_i \in [-2B,2B]$ as before.
    \item If the $i$th gate is of the type $x_i := \trunc_{[a,b]}(x_j)$, then by construction we have that
    \begin{equation*}
    \begin{split}
    \zbar_i^\pi = \trunc(\psi(\phi^{-1}(\xbar_j^\pi)) + \pi(\zbar_i)) &= \trunc(\psi(\phi^{-1}(\xbar_j^\pi) + (b-a)\pi(\zbar_i)))\\
    &= \psi(\trunc_{[a,b]}(\phi^{-1}(\xbar_j^\pi) + (b-a)\pi(\zbar_i)))\\
    &= \psi(\trunc_{[a,b]}(x_j^\sigma + \eps_j + (b-a)\pi(\zbar_i)))\\
    &= \psi(\trunc_{[a,b]}(x_j^\sigma + \sigma(x_i))) = \psi(x_i^\sigma)
    \end{split}
    \end{equation*}
    where we set $\sigma(x_i) := \eps_j + (b-a)\pi(\zbar_i)$, which satisfies $|\sigma(x_i)| \leq 4B^{i-1}\delta + B\delta \leq 4B^i \delta \leq \delta K$, and where we also used the induction hypothesis $\xbar_j^\pi = \phi(x_j^\sigma + \eps_j)$. Thus, it follows that
    \begin{equation*}
    \begin{split}
    \xbar_i^\pi = \trunc(\phi(\psi^{-1}(\zbar_i^\pi)) + \pi(\xbar_i)) &= \trunc(\phi(x_i^\sigma) + \pi(\xbar_i))\\
    &= \trunc(\phi(x_i^\sigma + 4B\pi(\xbar_i)))\\
    &= \trunc(\phi(x_i^\sigma + \eps_i)) = \phi(x_i^\sigma + \eps_i)
    \end{split}
    \end{equation*}
    where we set $\eps_i := 4B\pi(\xbar_i)$, which satisfies $|\eps_i| \leq 4B \delta \leq 4B^i \delta$, and where we also used the fact that $\phi(x_i^\sigma + \eps_i) \in [0,1]$, since $x_i^\sigma + \eps_i \in [-2B,2B]$ as before.
\end{itemize}
Thus, \eqref{eq:perturbation-induction} follows by induction, and the proof of \cref{lem:perturbation} is complete.

\bigskip
\subsubsection*{Acknowledgments}
We thank the anonymous reviewers for comments and suggestions that helped improve the presentation of the paper.
J.F.\ and R.S.\ were supported by EPSRC Grant EP/W014750/1.
P.W.G.\ was supported by EPSRC Grant EP/X040461/1.
A.H.\ was supported by the Swiss State Secretariat for Education, Research and Innovation (SERI) under contract number MB22.00026.

\bibliographystyle{alphaurl}
\bibliography{references}

\end{document}